\documentclass[12pt]{article}
\usepackage{amsmath, amssymb, amsthm}
\usepackage{graphicx,amssymb,amsfonts,amsbsy,amsmath,amsthm}
\usepackage{mathtools}
\usepackage{hyperref}
\usepackage{floatrow}
\usepackage{float}
\usepackage{subfig}
\usepackage{amssymb}
\usepackage{amsthm}
\usepackage{epstopdf}
\usepackage{anysize}
\usepackage{color,soul}
\usepackage{graphicx}
\usepackage{multirow}
\usepackage{enumerate}
\usepackage{csquotes}
\newtheorem{theorem}{Theorem}[section]

\newtheorem{remark}[theorem]{Remark}

\input epsf.sty
\usepackage{lineno}

\begin{document}
\begin{center}
	\baselineskip .2in {\large\bf Self-diffusion Driven Pattern Formation in Prey-Predator System with Complex Habitat under Fear Effect}
\end{center}

\begin{center}
	\baselineskip .2in {\bf Debaldev Jana$^{1}$, Saikat Batabyal$^{1}$, M. Lakshmanan$^{2}$}
	
	{\small\it $^1$Department of Mathematics \& SRM Research Institute,} \\ 
	SRM Institute of Science and Technology, Kattankulathur-603 203, Tamil Nadu, India \\
	{\small\it $^2$ Department of Nonlinear Dynamics, School of Physics,}\\
	Bharathidasan University, Tiruchirapalli - 620 024, India\\
	
\end{center}
%
%
%
%

\begin{abstract}
In the present work, we explore the influence of habitat complexity on the activities of prey and predator of a spatio-temporal system by incorporating self diffusion.	
First we modify the Rosenzweig-MacArthur predator-prey model by incorporating the effects of habitat complexity on the carrying capacity and fear effect of prey and predator functional response.	
We establish conditions for the existence and stability of all feasible equilibrium points of the non-spatial model and later we prove the existence of Hopf and transcritical bifurcations in different parametric phase-planes analytically and numerically. 
The stability of the spatial system is studied and we discuss the conditions for Turing instability. Selecting suitable control parameter from the Turing space, the existence conditions for stable patterns are derived using the amplitude equations. Results obtained from theoretical analysis of the amplitude equations are justified by numerical simulation results near the critical parameter value. Further, from numerical simulation, we illustrate the effect of diffusion of the dynamical system in the spatial domain by different pattern formations. Thus our model clearly shows that the fear effect of prey and predator's functional response make an anti-predator behaviour including habitat complexity which helps the prey to survive in the spatio-temporal domain through diffusive process. \\


\end{abstract}

\section{Introduction}\label{S1}
Structural complexity of any habitat (habitat complexity) is known to significantly affect the natural activities of the native species. It reduces the available resources (space, food, etc.) of prey resulting in the increasing possibility of inter-specific competition among the prey population which reduces the carrying capacity. Again, it reduces the foraging rate of predator due to decreasing encounter rate between the predator and prey. Reduced foraging intention of predator due to increasing habitat complexity makes the prey to be gradually fearless. Also, increasing habitat complexity reduces the rate of mobility as well as diffusion rate both of which affect the diffusive tendency of each of the species.
Habitat complexity is a morphological characteristics within the structure of a habitat \cite{BET91,LH04}. It may shape the predator-prey interaction and it has been demonstrated that the dynamics of locally interacting populations largely depend upon the attributes of the local habitats \cite{A01,J06}. Habitat complexity is observed
in almost all ecological systems, whether it is terrestrial or aquatic.
Marine habitat becomes complex in the presence of oyster and coral reefs, mangroves, sea grass beds and salt marshes \cite{HET11}. Experimental ecologists have measured the strength of habitat complexity in a significant way for the purpose of their needs. 
Marine ecologists measure the strength of habitat complexity by the amount of sea grass blades per square meter \cite{JH94}, the amount of shell material in polychaete worm tubes \cite{B85} or surface area to volume ratios \cite{CW83}. \\

Ecological niche is the way of life, which is defined by the full set of conditions, resources, and interactions it needs. Every species capacitates itself into an ecological community and has its own tolerable limits for certain ecological factors \cite{MS09}. Niche \cite{AM80} in ecology is also termed as the behaviour of a group of species and how they interact to the distribution of resources and competitors, for example abundance in resources and scarcity of predator, etc. \cite{JET15}, but in return to this instinct these species alter these factors themselves, for example the growth rate is increased due to the abundance of resources or increase in predator decreases the growth rate and increases the protection strategies, etc. \cite{A01,HM91,HH95,H87,HET06,F80,GR03,C95,K05,K11,K98,M74,M86,R95,SET12,S87,T84,OET13}. Environmental niche differs from one species to another in the context of type and number of variables comprising the ecological dimensions and the relative importance of particular set of environmental variables for a species which may vary according to the geographic and biotic contexts \cite{A01,M74,M86,R95,T84}. Ecological niche relies on the topology of the ecological structures which depends on the space \cite{HM91}, habitat physical and behavioural \cite{K11,K98,JR16} complexity and many other constraints \cite{H87,HM91} where the ecological communities belong to. In a niche one species changes or alters the dynamical behaviour of another species in order to maintain ecological fitness \cite{K11,K98}. Ecological niche is defined as for mapping species’ distributions across space and time. It gives the prediction of the impact of climate change on species ranges, on the basis of several other factors which affect the distribution, including species dispersal, biotic interactions, land use and topography. Predators search for their preys and preys avoid their predators: this is the fundamental relationship between prey and predator. Habitat complexity reduces the available resource (living space, food etc.) for the interacting species \cite{MA99,FC90}and consequently reduces the diffusion possibility. For example, aquatic weeds or submerged vegetations or reefs decrease the interaction chances of predator and prey, increases the inter-specific competition among preys, and decreases the diffusion rate with its increasing strength. \\

Pattern formation in the dynamical system presumes as an attractive proposition of the environment. Spatio temporal dynamics shows patterns \cite{BP11,SUN14,UVT12,WANG14,YZ14,ZSJ12} and topographical relationship among the species. Biodiversity of the environment is governed by the spatio-temporal dynamics. The discontinuities of population distribution in time and space makes patterns and it shows environmental irregularities. Diffusion has a powerful impact in pattern formation \cite{T52}. 
There is a certain chance of increasing stability of the system when diffusion happens. Diffusive instability of environmental systems raises many questions.  
Mathematical conditions and numerical findings are very special issues to observe the spatial variability of the populations.
A large amount of literature exists on these aspects, see for example references \cite{BB12,BA15,ZWW12,XUE12,LR03}.  In general, very often spatio temporal dynamics is modelled by reaction-diffusion \cite{MK15,ZS15,SJ72} equations and the study of such systems enriches our understanding of the performance of spatio temporal instabilities in the ecosystem with a large number of degrees of freedom. 

\subsection{Model construction}
Now we concentrate on the different functioning of any prey-predator relationship. Then the assumptions are:\\

$\bullet$ Increasing habitat complexity decreases the available resources (space, food, etc.) to prey population and decreasing resource increases the inter-specific competition of prey which directly affects the carrying capacity of prey \cite{JB14}. Habitat complexity reduces the space for prey population and decreases the prey carrying capacity. Therefore, the environmental carrying capacity should be a decreasing function of habitat complexity. For simplicity, we assume that environmental carrying capacity is reduced linearly with increasing habitat complexity. If we assume that the impact of habitat complexity on prey is the dimensionless parameter $c_1 (0<c_1<1)$ and $k$ is the environmental carrying capacity of prey in the absence of habitat complexity then it would be $k(1-c_1)$ in the presence of complexity \cite{JB14}. \\

$\bullet$ Previously, many experimental results have verified that the behaviour of prey can change drastically in the presence of predator and may be the effect is more powerful than the effect of direct predation  \cite{L98,CC08,L10,C11,ZC11,WZZ16,WZ17,SS18}.Prey shows anti-predator responses including habitat changes, foraging, vigilance, etc. and also they can react to the predation risk, which may increase their chance of survival in the environment \cite{C11,ZC11,WZZ16,WZ17,UM18}. Reproduction of prey depends on the fear of preys, which may affect the psychological condition of the juvenile preys and also the survival of adult preys \cite{CC08,ZC11,WZZ16,UM18}. If we consider the logistic growth function of prey ($rx(1-\frac{x}{k})=rx-\frac{rx^2}{k}, r$: reproduction rate, $k$: carrying capacity), then the production term ($r$) will be multiplied by the fear effect term $f(K,y)=\frac{1}{1+Ky}$ ($K:$ predator's fear effect on prey, $y$: predator density). Reduced foraging intention of predator due to increasing habitat complexity makes prey to be gradually fearless. If we assume the impact of habitat complexity on predator is the dimensionless parameter $c_2 (0<c_2<1)$, then, in the presence of habitat complexity the fear effect term can be modified as $f(K,y)=\frac{1}{1+(1-c_2)Ky}$. Finally the logistic growth function of prey becomes $\frac{rx}{1+(1-c_2)K y}-\frac{rx^2}{(1-c_1)k}$.\\

$\bullet$ Different field and laboratory experiments confirm that habitat complexity reduces predation rates by decreasing encounter rates between predator and prey \cite{JB14}. The foraging efficiency of predator generally decreases with increasing habitat complexity. Predator’s functional response, defined as the amount of prey catch per predator per unit of time, is thus affected by the structure of the habitat \cite{AL01,A01}. The effect of habitat complexity should, therefore, be incorporated in predator’s response function. If we assume that the impact of habitat complexity on predator is the dimensionless parameter $c_2 (0<c_2<1)$, then the traditional Holling type II functional response ($\frac{\alpha x}{1+\alpha h x}, \alpha$: attack rate of predator, $h$: food handling time, $x$: prey density) is modified by incorporating the strength of habitat complexity ($c_2$) as $\frac{\alpha(1-c_2) x}{1+\alpha(1-c_2) h x}$ \cite{JB14} and it directly contributes to the growth of the predator as $\frac{\theta \alpha(1-c_2) x}{1+\alpha(1-c_2) h x}$, where $\theta$ is the food conversion efficiency.\\

$\bullet$ Habitat complexity has a significant role on community development according to resource availability. The rate of space sequestration is changed and this makes the effects of habitat complexity varied through space. Diversity indices and species abundances are observed with fine-scale complexity. Sometimes it has been noticed that complexity effects also appear to diminish with time \cite{M84}. Environmental factors and other related effects have significant role to play in the diffusive domain of a predator-prey model in the habitat complexity. The aptitude of species for living in a given habitat is to migrate to better suitable regions for their own survival and existence. So the predator follows the prey for its own survival in the environment. Now the prey and predator will travel from one to another region for some biological reasons with diffusion rates $D_1$ and $D_2$. As a result, the parameters $\left(1-c_1\right)D_1$ and $\left(1-c_2\right)D_2$ in an ecosystem tend to give rise to a uniform density of population in the habitat complexity. Consequently, it may be expected that diffusion, when it occurs, plays an important role in increasing the stability in a habitat complexity of mixed populations and resources. Along with self-diffusion, the study of spatio-temporal diffusion suggests that space, or perhaps the habitat complexity, may lead to an important factor in the community development of predator-prey communities.\\

Incorporating self-diffusion, we obtain the following system of coupled reaction-diffusion equations to represent the above system of prey and predator populations:
\begin{equation}\label{eq:03}
\begin{split}
\frac{\partial x}{\partial t} &= \frac{rx}{1+(1-c_2)Ky}-\frac{rx^2}{(1-c_1)k}-\frac{\alpha (1-c_2) xy}{1+\alpha (1-c_2)hx}+(1-c_1)D_1\nabla^2 x, \\
\frac{\partial y}{\partial t} &= \frac{\theta\alpha (1-c_2) xy}{1+\alpha (1-c_2)hx}-d.y+(1-c_2)D_2\nabla^2 y.
\end{split}
\end{equation}
where $\nabla^2\equiv\frac {\partial^2}{\partial X^2}+\frac {\partial^2}{\partial Y^2}$, $D_{1}$ and $D_{2}$ are diffusion constants of prey ($x(X,Y,t)$) and predator ($y(X,Y,t)$) respectively and we assume $x(X,Y,0)> 0, y(X,Y,0)> 0.$ Note here that in \ref{eq:03} $X$ and $Y$ represent spatial coordinates, while $x$ and $y$ represent prey and predator densities.\\

In the present study, the effect of habitat complexity on prey-predator interaction is studied when fear effect of predator on prey is incorporated and the effect of self diffusion included. We bring out how different kinds of interesting spatio temporal patterns emerge in the system \ref{eq:03}. The manuscript is organized as follows: Existence and stability of equilibrium points of the non-spatial model is given in Section \ref{S2} and also the existence of Hopf and transcritical bifurcations are investigated in this section. Impact of population density on the spatial variability along with time is the mathematical key findings. Conditions for existence/non-existence of Turing patterns, weak nonlinear analysis and stability analysis of amplitude equations are studied in Sections \ref{S3}, \ref{S4} and \ref{S5}, respectively. Numerical simulation is described in Section \ref{S6}. Finally in Section \ref{S7} we briefly present our conclusions.

\section{Existence and stability of equilibrium points of non-spatial model}\label{S2}

Ecological stability can correspond to different types of stability during a limit starting from regeneration via resilience (returning quickly to a previous state), to constancy and then to persistence. The precise definition depends on the ecosystem in question, the variable or variables of interest, and therefore the overall context. In the context of conservation ecology, stable populations are often defined as the ones that do not go extinct. Researchers usually use the notion of Lyapunov stability to mathematical models. Local stability suggests the stability domain of a system over small short-lived disturbances. Now we carry out the stability analysis of the non-spatial model system to start with, which is written as follows:

\begin{equation}\label{eq:2}
\begin{split}
\dot{x} &= \frac{rx}{1+(1-c_2)Ky}-\frac{rx^2}{(1-c_1)k}-\frac{\alpha (1-c_2) xy}{1+\alpha (1-c_2)hx}, \\
\dot{y} &= \frac{\theta\alpha (1-c_2) xy}{1+\alpha (1-c_2)hx}-d.y,
\end{split}
\end{equation}
where $\dot{x}$ denotes the derivative of the variable $x$ with respect to time $t$ and so also $\dot{y}$.
The system (\ref{eq:2}) has three equilibrium points and their feasibility and stability are given below.
\begin{theorem}
	\label{thm:1}
	The trivial equilibrium point $E_0(0,0)$ always exists and is a saddle point.
\end{theorem}
\begin{proof}
	The Jacobian matrix of system (\ref{eq:2}) around $E_0$ is 
	$J(E_0)=\left( \begin{array}{cc}
	r & 0
	\\
	0 & -d
	\end{array} \right)$, the eigenvalues are $\lambda_1=r>0$ and $\lambda_2=-d<0$. So, $E_0$ is a saddle point.
\end{proof}
\begin{theorem}
	\label{thm:2}
	The prey only equilibrium point $E_1(x_1,y_1)$, where $x_1=(1-c_1)k,y_1=0$, always exists $(\because c_1\in(0,1))$ and is locally asymptotically stable in the parametric region $\Gamma_1=\{c_1,c_2\in(0,1),\theta>\theta_1: \{c_1<\hat{c}_1,c_2>c_2^*\}\mbox{or}\{c_2<\hat{c}_2,c_1>c_1^*\}\}$, non-hyperbolic on the boundary of the region $\Gamma_1$ and unstable otherwise, where $c_1^*=1-\frac{d}{\alpha k(1-c_2)(\theta-hd)}$, $c_2^*=1-\frac{d}{\alpha k(1-c_1)(\theta-hd)}$, $\hat{c}_1=1-\frac{d}{\alpha k (\theta-hd)}=\hat{c}_2$ and $\theta_1=hd+\frac{d}{\alpha k}$.
\end{theorem}
\begin{proof}
	Prey only equilibrium means an equilibrium which contains only prey density, and no predator density. That means at this equilibrium point predator population density is zero.
	The Jacobian matrix of system (\ref{eq:2}) around $E_1$ is 
	$$J(E_1)=\left( \begin{array}{cc}
	-r &~~ -rkK(1-c_1)(1-c_2)-\frac{\alpha k(1-c_1)(1-c_2)}{1+\alpha k h (1-c_1)(1-c_2)}
	\\
	0 &~ \frac{\theta \alpha k(1-c_1)(1-c_2)}{1+\alpha k h (1-c_1)(1-c_2)}-d
	\end{array} \right)$$ and the eigenvalues are $\lambda_1=-r<0$ and $\lambda_2=\frac{\theta \alpha k(1-c_1)(1-c_2)}{1+\alpha k h (1-c_1)(1-c_2)}-d$. So, $E_1$ is asymptotically stable, non-hyperbolic and unstable if $\lambda_2<0$, $\lambda_2=0$ and $\lambda_2>0$, respectively, which is straightforward to check.
\end{proof}


\begin{theorem}
	\label{thm:4}
	The coexistence equilibrium point $E^*(x^*,y^*)$, where $x^*=\frac{d}{\alpha(1-c_2)(\theta-hd)}, y^*=\frac{1}{2(1-c_2)K}\bigg(\sqrt{\Delta}-1-\frac{r\theta dK}{\alpha^2(\theta-hd)^2(1-c_1)(1-c_2)}\bigg)$ and $\Delta=\bigg(\frac{r\theta dK}{\alpha^2(\theta-hd)^2(1-c_1)(1-c_2)}-1\bigg)^2+\frac{4r\theta K(1-c_1)}{\alpha(\theta-hd)(1-c_2)}$, is feasible in the region $\{c_1,c_2\in(0,1),\theta>\theta_1: \{c_2<c_2^*,c_1<\hat{c}_1\}\mbox{or}\{c_1<c_1^*,c_2<\hat{c}_2\}\}$.
	Let\\
	{\bf H$_1$:} $\Gamma_2=\bigg\{c_1,c_2\in(0,1),\alpha>\frac{1}{hk},\theta>\max\{\theta_1,\theta_2\}:		\{c_1<\min\{\bar{c}_1,\hat{c}_1\},\tilde{c}_2<c_2<c_2^*\}or\{c_2<\min\{\bar{c}_2,\hat{c}_2\},\tilde{c}_1<c_1<c_1^*\bigg\},$\\
	{\bf H$_2$:} $\Gamma_3=\bigg\{c_1,c_2\in(0,1),\alpha>\frac{1}{hk},\theta>\theta_2: 
	\{c_1<\bar{c}_1,c_2<\tilde{c}_2\}or\{c_2<\bar{c}_2,c_1<\tilde{c}_1\}
	\bigg\}$,\\
	{\bf H$_3$:} $\Gamma_2^*=\bigg\{c_1,c_2\in(0,1),\alpha>\frac{1}{hk},\theta>\theta_2: 
	\{c_1<\bar{c}_1,c_2=\tilde{c}_2\}or\{c_2<\bar{c}_2,c_1=\tilde{c}_1\}
	\bigg\}$ and\\
	{\bf H$_4$:} $\Gamma_1^*=\bigg\{c_1,c_2\in (0,1),\alpha>\frac{1}{hk},\theta>\theta_2: \{c_1<\hat{c}_1,c_2=c_2^*\}or\{c_2<\hat{c}_2,c_1=c_1^*\}\bigg\}$\\
	{\bf H$_4$:} $\bar{K}=\frac{k(\theta-hd)^2(1-c_1)(1-c_2)\{\alpha kh(1-c_1)(1-c_2)(\theta-hd)-(\theta+hd)\}}{r\theta^2(\theta+hd)^2}$.\\
	Then\\
	(i) $E^*$ is locally asymptotically stable if {\bf H$_1$} or \{{\bf H$_2$} and $K>\bar{K}$\} hold, unstable focus if \{{\bf H$_2$} and $0<K<\bar{K}$\} hold.\\
	(ii) The system (\ref{eq:2}) undergoes a Hopf bifurcation around $E^*$ if {\bf H$_3$} or \{{\bf H$_2$} and $K=\bar{K}$\} hold.\\
	(iii) The system (\ref{eq:2}) undergoes a transcritical bifurcation between $E^*$ and $E_1$ if {\bf H$_4$} holds,\\   
	(iv) $E^*$ is globally asymptotically stable if {\bf H$_1$} and $hd<\theta<h(r+d)$ hold.\\
	Here $\tilde{c}_1=1-\frac{\theta+hd}{\alpha k h(1-c_2)(\theta-hd)},\tilde{c}_2=1-\frac{\theta+hd}{\alpha k h(1-c_1)(\theta-hd)},\bar{c}_1=1-\frac{\theta+hd}{\alpha k h(\theta-hd)}=\bar{c}_2.$
\end{theorem}
\begin{proof}
	Coexistence equilibrium point is the most import fact in the dynamical system where non-zero constant values of both prey and predator densities are allowed. We observe the behaviour of the dynamical system around the coexistence equilibrium point and so that we analyse the bifurcation (Hopf and Transcritical) around the coexistence equilibrium point. So, in this context, it is very much important to show the feasible region of the coexistence equilibrium point of the prey-predator system.
	$(i)$ The Jacobian matrix of the system (\ref{eq:2}) at $E^*$ is $J(E^*)=\left( \begin{array}{cc}
	a_{11} & a_{12}
	\\
	a_{21} & 0
	\end{array} \right)$ where
	$$\begin{array}{l}
	a_{11}=-\frac{rx^*}{(1-c_1)k}+\frac{\alpha^2(1-c_2)^2hx^*y^*}{\{1+\alpha(1-c_2)hx^*\}^2},\\
	a_{12}=-\frac{rx^*K(1-c_2)}{\{1+(1-c_2)Ky^*\}^2}-\frac{\alpha(1-c_2)x^*}{1+\alpha(1-c_2)hx^*}<0,\\
	a_{21}=\frac{\theta\alpha(1-c_2)y^*}{\{1+\alpha(1-c_2)hx^*\}^2}>0.
	\end{array}$$
	The characteristic equation is given by $\lambda^2-a_{11}\lambda-a_{12}a_{21}=0$. Here det$(J(E^*))=-a_{21}a_{12}>0$. So, $E^*$ is asymptotically stable if tr$J(E^*)=a_{11}<0$ (real part of the eigenvalues). Making use of the above form of $a_{11}$, which after a straightforward mathematical calculation yields the condition
	\begin{equation}\label{eq:3}
	\frac{2\{\alpha kh(\theta-hd)(1-c_1)(1-c_2)-(\theta+hd)\}}{(\theta+hd)}<\sqrt{\Delta}-1-\frac{r\theta d K}{\alpha^2(\theta-hd)^2(1-c_1)(1-c_2)}.
	\end{equation}
	If {\bf H$_1$} holds, then inequality (\ref{eq:3}) holds and if {\bf H$_2$} holds then $K>\bar{K} (\mbox{threshold value})$ is needed for the stability of $E^*$. Further, $E^*$ loses its stability if {\bf H$_2$} and $0<K<\bar{K}$ hold.\\
	
	$(ii)$ Proof for the existence of Hopf bifurcation and limit cycle is given in Section (\ref{s3}). \\
	
	$(iii)$ Proof for the existence of transcritical bifurcation is given in Section (\ref{s4}). \\
	
	$(iv)$ Assume the right hand sides of system (\ref{eq:2}) are $P(x,y)$ and $Q(x,y)$ respectively. We consider the Dulac function as
	\[B(x,y)=\{1+\alpha h(1-c_2)x\}\{1+(1-c_2)Ky\}x^{-1}y^p,\] 
	where $p$ will be defined later. We calculate the divergence of the vector 
	\[D=\frac{\partial(P(x,y)B(x,y))}{\partial x}+\frac{\partial(Q(x,y)B(x,y))}{\partial y}=x^{-1}y^p\{f_(x)(1-c_2)Ky+f_2(x)\}.\]
	Here
	$$\begin{array}{l}
	f_1(x,p)=-\frac{2\alpha hr(1-c_2)}{(1-c_1)k}x^2+\{(p+2)\alpha(\theta-hd)(1-c_2)-\frac{r}{(1-c_1)k}\}x-(p+2)d,\\
	f_2(x,p)=-\frac{2\alpha hr(1-c_2)}{(1-c_1)k}x^2+\{(p+1)\alpha(\theta-hd)(1-c_2)+\alpha hr(1-c_2)-\frac{r}{(1-c_1)k}\}x-(p+1)d.
	\end{array}$$
	If $hd<\theta<h(r+d)$ holds, then
	\[f_2(x,p)-f_1(x,p)=\alpha\{h(r+d)-\theta\}(1-c_2)x+d\geq 0~\mbox{for} ~ x\in[0,\infty).\]
	So, $D\leq 0$ for $(x,y)\in \mathbb{R}^2_{+}$ if
	
	\begin{equation}\label{eq:4}
	f_2(x,p)\leq 0, ~\mbox{for} ~x\in[0,\infty).\end{equation} 
	Now it is sufficient to find a $p$ for which (\ref{eq:4}) holds and for so
	
	\begin{equation}\label{eq:5}
	l(p+1)=\bigg(\alpha(p+1)(1-c_2)(\theta-hd)+\alpha h r(1-c_2)-\frac{r}{(1-c_1)k}\bigg)^2-\frac{8\alpha hrd(p+1)(1-c_2)}{(1-c_1)k}\leq 0.\end{equation} 
	Let, $\bar{p}=p+1$, then (\ref{eq:5}) becomes

	\begin{equation}\label{eq:6}
	\begin{split}
	l(\bar{p}) &=(1-c_2)^2\alpha^2(\theta-hd)^2\bar{p}^2+2(1-c_2)\bigg[\bigg(\alpha h r(1-c_2)-\frac{r}{(1-c_1)k}\bigg)\alpha(\theta-hd)-\frac{4\alpha hrd}{(1-c_1)k}\bigg]\bar{p}\\
	&+\bigg\{\alpha h r(1-c_2)-\frac{r}{(1-c_1)k}\bigg\}^2\leq 0.
	\end{split}
	\end{equation} 
	The existence of $\bar{p}$ satisfying (\ref{eq:6}) means
	
	$$\begin{array}{l}
	l\bigg(-\frac{\bigg(\alpha h r(1-c_2)-\frac{r}{(1-c_1)k}\bigg)\alpha(\theta-hd)-\frac{4\alpha hrd}{(1-c_1)k}}{\alpha^2(\theta-hd)^2(1-c_2)}\bigg)\leq 0 \\
	\implies \bigg\{\bigg(\alpha h r(1-c_2)-\frac{r}{(1-c_1)k}\bigg)\alpha(\theta-hd)-\frac{4\alpha hrd}{(1-c_1)k}\bigg\}^2-\bigg(\alpha h r(1-c_2)-\frac{r}{(1-c_1)k}\bigg)^2\alpha^2(\theta-hd)^2\geq 0\\
	\implies (1-c_1)(1-c_2)\leq \frac{\theta+hd}{\alpha hk(\theta-hd)}\implies \Gamma_3.
	\end{array}$$
	Thus there exists a $p$ such that $D\leq 0$ for $(x,y)\in \mathbb{R}^2_{+}$. By the Bendixson-Dulac theorem (\cite{B83}) system (\ref{eq:2}) has no periodic orbits in the positive plane. Because $E^*$ is the only positive equilibrium in the plane, each of the positive solutions tends to $E^*$. By the consideration of the local asymptotically stability of $E^*$, we come to the conclusion that $E^*$ is globally asymptotically stable if both {\bf H$_1$} and the condition $hd<\theta<h(r+d)$ hold.
\end{proof}

\subsection{Existence of Hopf bifurcation and limit cycle}\label{s3}
On the basis of Poincar$\acute{\mbox{e}}$-Bendixson theorem \cite{M07}, in this section we analyze the existence of Hopf bifurcation and limit cycle.


\begin{theorem}\label{thm:6}
	Assume {\bf H$_2$} holds, then system (\ref{eq:2}) undergoes a Hopf bifurcation around $E^*$ at $K=\bar{K}$. If {\bf H$_2$} and $0<K<\bar{K}$ hold, then there exists a limit cycle of the system (\ref{eq:2}). 
\end{theorem}
\begin{proof}
	From an analysis of the characteristic equation for $E^*$, namely $\lambda^2-\mbox{tr}(J(E^*))\lambda+\mbox{det}(J(E^*))=0$, one can establish that there exists a Hopf bifurcation iff $K=\bar{K}$ such that:
	$$\begin{array}{l}
	(i)~ \mbox{tr}(J(E^*))|_{K=\bar{K}}=0,\\
	(ii)~ \mbox{det}|_{K=\bar{K}}>0 ~\mbox{which is equivalent to purely imaginary eigenvalues},\\
	(iii)~ \frac{d\mbox{tr}(J(E^*))}{dK}|_{K=\bar{K}}\neq 0.
	\end{array}$$
	As the condition tr$(J(E^*))=0$ gives $a_{11}(\bar{K})=0$ and det$(J(E^*))=-a_{12}(J(E^*))a_{21}(J(E^*))>0$, the conditions (i) and (ii) are satisfied. So we need to verify only $\frac{d\mbox{tr}(J(E^*))}{dK}|_{K=\bar{K}}\neq 0$. By the assumption of this theorem, we get
	
	\begin{equation}\label{eq:7}
	\begin{split}
	\frac{d\mbox{tr}(J(E^*))}{dK}\bigg|_{K=\bar{K}} & =\frac{da_{11}(J(E^*))}{dK}\bigg|_{K=\bar{K}} \\
	& =\frac{\alpha^2(\theta-hd)^2(1-c_2)\bigg[
		\bigg(\alpha h r(1-c_2)-\frac{r}{(1-c_1)k}\bigg)\alpha(\theta-hd)-\frac{r\alpha(\theta+hd)}{(1-c_1)k}\bigg]}{\frac{\bar{K}^2r\theta^2\alpha^2(\theta+hd)}{(1-c_1)k}}>0.
	\end{split}
	\end{equation}
	Hence the condition (iii) is satisfied which implies that system (\ref{eq:2}) undergoes a Hopf bifurcation at $K=\bar{K}$.\\
	
	Now we will prove that the system (\ref{eq:2}) has a unique limit cycle. Already we know that if {\bf H$_2$} and the condition $0<K<\bar{K}$ hold, then $E^*$ is unstable and $E_1$ is a saddle point. From {\bf H$_2$}, we get
	\[x_1=(1-c_1)k>\frac{d}{\alpha(1-c_2)(\theta-hd)}=x^*.\]
	Let $L:x-x_1=0$, $M:y-\tilde{y}=0$ ($\tilde{y}>0$ will be specified later) and $N:2(x_1-x^*)(y-\tilde{y})+\tilde{y}(x-x^*)=0$.
	
	\begin{equation}\label{eq:8}
	\frac{dL}{dt}\bigg|_{x=x_1}=\frac{dx}{dt}\bigg|_{x=x_1}=\frac{rx_1}{1+(1-c_2)Ky}-\frac{rx_1^2}{(1-c_1)k}-\frac{\alpha (1-c_2) x_1y}{1+\alpha (1-c_2)hx_1}<0 ~ \mbox{for} ~y\in(0,\infty).
	\end{equation}
	
	\begin{equation}\label{eq:9}
	\frac{dM}{dt}\bigg|_{y=\tilde{y}}=\frac{dy}{dt}\bigg|_{y=\tilde{y}}=\frac{\theta\alpha (1-c_2) x\tilde{y}}{1+\alpha (1-c_2)hx}-d\tilde{y}<0~\mbox{for}~x\in(0,x^*).
	\end{equation}
	Since $0<x<x_1$, we have $x-x^*<x_1-x^*$ and $y=\tilde{y}-\frac{\tilde{y}(x-x^*)}{2(x_1-x^*)}>\frac{\tilde{y}}{2}$.When $0<x<x^*,y>0$, then
	
	\begin{equation}\label{eq:10}
	\frac{dN}{dt}\bigg|_{x=x_1,y=\tilde{y}}=2(x_1-x^*)\frac{dy}{dt}+\tilde{y}\frac{dx}{dt}<\frac{1}{\{1+\alpha h(1-c_2)x\}\{1+(1-c_2)Ky\}}f(\tilde{y})
	\end{equation}
	where 
	
	\begin{equation}\label{eq:11}
	\begin{split}
	f(\tilde{y})&=-\frac{\alpha K(1-c_2)^2x\tilde{y}^3}{4}\\
	&+\bigg\{2\alpha(\theta-hd)(1-c_2)^2(x_1-x^*)^2-\frac{x}{2}\bigg\{\frac{\alpha rhK(1-c_2)^2x^*}{(1-c_1)k}+\frac{rK(1-c_2)x}{(1-c_1)k}+\alpha(1-c_2)\bigg\}x\bigg\} \tilde{y}^2\\
	&+\bigg\{2\alpha(\theta-hd)(1-c_2)(x_1-x^*)^2+\bigg\{ -\frac{\alpha rh(1-c_2)x^2}{(1-c_1)k}+\bigg\{\alpha rh(1-c_2)-\frac{r}{(1-c_1)k}\bigg\}x+r\bigg\}x\tilde{y}.
	\end{split}
	\end{equation}
	Because $\alpha K(1-c_2)^2x>0,\frac{dN}{dt}<0$ for sufficiently large $\tilde{y}>0$. If {\bf H$_2$} and $0<K<\bar{K}$ hold, then by Poincar$\acute{\mbox{e}}$-Bendixson theorem \cite{M07}, system (\ref{eq:2}) has a unique limit cycle.
\end{proof}

\subsection{Existence of transcritical bifurcation}\label{s4}
The system (\ref{eq:2}) undergoes a transcritical bifurcation between $E^*$ and $E_1$ for appropriate choice of parameters. It has been shown that when $(c_1,c_2)\in \Gamma_2$, the coexistence equilibrium E* is stable but the axial equilibrium $E_1$ is unstable and when $(c_1,c_2)\in \Gamma_1$, the coexistence equilibrium E* is unstable (in this case $E^*$ is biologically infeasible) but the axial equilibrium $E_1$ is stable. The two equilibria coincide when $(c_1,c_2)\in\Gamma_4$ and exchange their stability. By using Sotomayer's theorem, we prove the existence of transcritical bifurcation. On $(c_1,c_2)\in\Gamma_4$, there exists only one equilibrium point which is $E_1$ and the Jacobian matrix is calculated at $E_1$ as
\[J=Df(x_1,y_1:(c_1,c_2)\in\Gamma_4)=\left( \begin{array}{cc}
-r &~~ -rkK(1-c_1)(1-c_2)
\\
0 & 0
\end{array} \right).\]
$J$ has an eigenvalue $\lambda=0$. Let $u$ and $v$ be the corresponding eigenvectors
to the eigenvalue $\lambda=0$ for $J$ and $J^T$ respectively. Then we can calculate
$u=\left( \begin{array}{c}
-\frac{rkK(1-c_1)(1-c_2)}{r}
\\
0
\end{array} \right)$ and $v=\left( \begin{array}{c}
0
\\
1
\end{array} \right)$. Using the expressions for $u$ and $v$, we get
$$\begin{array}{l}
(i)~v^Tf_{c_1}(x_1,y_1;(c_1,c_2)\in\Gamma_4)|_{c_2<\hat{c}_2}=v^Tf_{c_2}(x_1,y_1;(c_1,c_2)\in\Gamma_4)|_{c_1<\hat{c}_1}=0,\\
(ii)~v^T[Df_{c_1}(x_1,y_1;(c_1,c_2)\in\Gamma_4)u|_{c_2<\hat{c}_2}]=v^T[Df_{c_2}(x_1,y_1;(c_1,c_2)\in\Gamma_4)u|_{c_1<\hat{c}_1}]=-\frac{\alpha k(\theta-hd)^2}{\theta}\neq 0 ,\\
(iii)~v^T[D^2f_{c_1}(x_1,y_1;(c_1,c_2)\in\Gamma_4)(u,u)|_{c_2<\hat{c}_2}]=-\frac{2\alpha d^2}{rk\theta}\neq 0.
\end{array}$$
Thus, by Sotomayor's theorem, we conclude that the model
system undergoes a transcritical bifurcation as $(c_1, c_2)$ passes
through the critical line $\Gamma_4$.

\section{Conditions for existence/non-existence of Turing patterns}\label{S3}

In this section, we investigate the Turing instability. We uncover both spatial and spatio-temporal patterns, and provide the details of the Turing analysis. We derive conditions where the unique positive interior equilibrium point $(x^*,y^*)$ becomes unstable due to the action of diffusion under a small perturbation to the positive interior equilibrium point. We first linearize the model \eqref{eq:03} about the homogeneous steady state, by introducing both space and time-dependent fluctuations around $(x^*,y^*)$. This is given as
\begin{subequations}\label{eq:A}
	\begin{align}
	x & =x^* + \hat{x}(\xi,t),\\
	y & =y^* + \hat{y}(\xi,t),
	\end{align}
\end{subequations}
where $| \hat{x}(\xi,t)|\ll x^*$, \mbox{and} $| \hat{y}(\xi,t)|\ll y^*$. Conventionally, we choose
\[
\left[ {\begin{array}{cc}
	\hat{x}(\xi,t)  \\
	\hat{y}(\xi,t) \\
	\end{array} } \right]
=
\left[ {\begin{array}{cc}
	\epsilon_1  \\
	\epsilon_2 \\
	\end{array} } \right]
e^{\lambda t + i\kappa\xi},
\]
where  $\epsilon_i$ for $i=1,2$ are the corresponding amplitudes, $\kappa$ is the wave number, $\lambda$ is the growth rate of perturbation in time $t$ and $\xi$ is the spatial coordinate.
Substituting \eqref{eq:A} into \eqref{eq:03} and ignoring higher order terms including nonlinear terms, we obtain the characteristic equation which is given as
\begin{align}\label{eq:B}
({\bf J} - \lambda{\bf I} - \kappa^2{\bf D})
\left[ {\begin{array}{cc}
	\epsilon_1  \\
	\epsilon_2 \\
	\end{array} } \right]=0,
\end{align}
where
\[
\quad
\bf {D} =
\left[ {\begin{array}{cc}
	(1-c_1)D_1  & 0  \\
	0    & (1-c_2)D_2
	\end{array} } \right],
\]
\begin{align}\label{eq:ABC}
J &=\left(
\begin{array}{cc}
-\frac{rx^*}{(1-c_1)k}+\frac{\alpha^2(1-c_2)^2hx^*y^*}{\{1+\alpha(1-c_2)hx^*\}^2} & ~~ -\frac{rx^*K(1-c_2)}{\{1+(1-c_2)Ky^*\}^2}-\frac{\alpha(1-c_2)x^*}{1+\alpha(1-c_2)hx^*} \\
& \\
\frac{\theta\alpha(1-c_2)y^*}{\{1+\alpha(1-c_2)hx^*\}^2} &  0 \\
\end{array}
\right) &=\left(\begin{array}{cc} {J_{11} } & {J_{12} } \\ {J_{21} } & {J_{22} }  \\  \end{array}\right),
\end{align} 

and $\bf{I}$ is a $2\times 2$ identity matrix.
For the non-trivial solution of \eqref{eq:B}, we require that
\[
\left|
\begin{array}{cc}
J_{11}-\lambda -\kappa^2(1-c_1)D_1 & J_{12} \\
J_{21}  & J_{22}-\lambda -\kappa^2(1-c_2)D_2 \\

\end{array} \right|=0,
\]

To determine the stability domain associated with $(x^*,y^*)$, we rewrite the dispersion relation as a quadratic polynomial function given as
\begin{align}\label{eq:C}
\lambda_{\kappa}^2 - \boldsymbol {tr}(J_\kappa)\lambda_{\kappa} + \boldsymbol {\Delta_\kappa}=0,
\end{align}
with the coefficients
\begin{align*}
\boldsymbol {tr}(J_\kappa)&= (J_{11} + J_{22}) - \kappa^2((1-c_1)D_1 + (1-c_1)D_2) ,\\
\boldsymbol {\Delta_\kappa} &= \kappa^4(1-c_1)(1-c_2)D_1D_2 - \kappa^2\big( (1-c_2)D_2J_{11} + (1-c_1)D_1J_{22}\big)+\Delta_0,\\
\end{align*}
where $\Delta_0$ is the determinant of the Jacobian matrix of \eqref{eq:2} i.e. the determinant of the Jacobian matrix for the non-spatial case.\\
The roots of the above equation are,
\[\lambda_{\kappa}=\frac{1}{2}[tr(J_\kappa)\pm \sqrt{tr^2(J_\kappa)-4\Delta_\kappa} ]\]

It is a well known fact that Turing bifurcation is the main reason for breaking the spatial symmetry and generation of patterns which are oscillatory and stationary with space and time respectively \cite{Z12}. The conditions for Turing instability are $Re(\lambda_\kappa)=0$, $Im(\lambda_\kappa)=0$, at $\kappa=\kappa_T\neq0.$ \\\\
The general conditions of Turing bifurcation are in the following:
\begin{subequations}\label{eq:D}
	\begin{align}
	J_{11}+J_{22}<0 \label{eq:D1}\\ 
	J_{11}J_{22} - J_{12}J_{21}>0  \label{eq:D2}\\
	(1-c_2)D_2J_{11} + (1-c_1)D_1J_{22}>2\sqrt{(1-c_1)(1-c_2)D_1D_2\Delta_0} \label{eq:D4}
	\end{align}
\end{subequations}
Now we differentiate the constant term of \eqref{eq:C} w.r.t. $\kappa^2$ and we obtain
$\kappa^2_{min}=\frac{(1-c_2)D_2J_{11}+(1-c_1)D_1J_{22}}{2(1-c_1)(1-c_2)D_1D_2}$  and also get $\kappa^2_{T}=\sqrt{\frac{\Delta_0}{(1-c_1)(1-c_2)D_1D_2}}$ from the condition of Turing instability ($Im(\lambda_\kappa)=0$). At the Turing threshold, the spatial symmetry of the dynamical system is broken and there arises stationary and oscillatory patterns in time and space respectively with the wavelength $\lambda_{T}=\frac{2\pi}{\kappa_{T}}$ \cite{W07,ZSJ12}. \\

One dimensional diffusion problem may experience rapid change in the very beginning, but then the evolution of the solution becomes very slow. The solution is usually very smooth, and after some time, one cannot recognize the initial shape of it. This is in sharp contrast to solutions where the initial shape is preserved and the solution is basically moving to initial condition. Small time steps are not much convenient and not required by accuracy as the diffusion process converges to a stationary state. Incorporating the one dimensional diffusion in the model system \eqref{eq:2}, we obtain the following reaction-diffusion equations: 
\begin{equation}\label{eq:333}
\begin{split}
\frac{\partial x}{\partial t} &= \frac{rx}{1+(1-c_2)Ky}-\frac{rx^2}{(1-c_1)k}-\frac{\alpha (1-c_2) xy}{1+\alpha (1-c_2)hx}+(1-c_1)D_1x_{X}, \\
\frac{\partial y}{\partial t} &= \frac{\theta\alpha (1-c_2) xy}{1+\alpha (1-c_2)hx}-dy+(1-c_2)D_2y_{X},
\end{split}
\end{equation}
where  $D_{1}$ \mbox{and} $D_{2}$ are diffusion constants of prey ($x(X,t)$) and predator ($y(X,t)$) respectively and $x(X,0)> 0, y(X,0)> 0$ and we consider our domain as $\Omega =[0,\pi]$.
We prescribe Neumann boundary conditions
\begin{equation}
\frac{\partial x}{\partial X} = \frac{\partial y}{\partial X} = 0.
\end{equation}

\begin{theorem}
	\label{thm:ABC}
	Let us assume the predator-prey model is \eqref{eq:333}, and its Jacobian matrix $J$ of reaction terms is given by \eqref{eq:ABC}. Consider a parameter set for which $J$ is stable, then there exist diffusion coefficients $D_{1}$ \mbox{and} $D_{2}$, such that the spatially homogeneous steady state $(x^{*},y^{*})$ of \eqref{eq:333}, can be driven unstable due to Turing instability.
\end{theorem}

\section{Weakly Nonlinear Analysis}\label{S4}

Instabilities are studied in general as linear instabilities by which we mean the perturbations are sufficiently small. It has been shown that dynamical systems have a major slow down at the bifurcation threshold point. Eigenvalues corresponding to the dynamical system are close to zero  when the bifurcation parameter is very close to a Turing bifurcation parameter \cite{DB96}. This indicates that the critical modes are slow and the dynamical system shows the presence of active slow modes \cite{W10,WW11,Z11}. At this point Turing patterns appear and the structure and stability of Turing patterns are explained by the study of amplitude equations \cite{I00}. The study of weakly nonlinear analysis helps to understand the etymology of amplitude equations. In this section we derive amplitude equations with the help of the standard multiple-scale analysis \cite{S08,Y13,Z14,GP17}.
Let us assume that the Turing patterns \cite{WW10} have three pairs of modes  $\left( \textbf{k}_{i},\textbf{-k}_{i}, i = 1, 2, 3 \right)$ making an angle of $\frac{2\pi}{3}$ between each pair with the following conditions $|\textbf{k}_{i}|=\kappa_T$ and $\sum_{i=1}^{3}\textbf{k}_{i}=0.$ In this section, we will obtain the amplitude equations with the help of standard multiple-scale analysis. For this purpose, the solution of the model system \eqref{eq:03} is expanded as follows:

\begin{align}\label{eq:E1}
\textbf{x}=\left( {\begin{array}{c}
	x  \\
	y  
	\end{array} } \right)=\tilde{\textbf{x}}.\sum_{j=1}^{3}[A_{j}\mbox{exp}(i\textbf{k}_{j}.\textbf{r})+\bar{A}_{j}\mbox{exp}(-i\textbf{k}_{j}.\textbf{r})],
\end{align}
where $\tilde{\textbf{x}}$ is the eigenvector of the linearized operator. Basically $\tilde{\textbf{x}}$ is the direction of the eigen modes in the concentration space (i.e. the ratio of x and y). $A_{j}$ and conjugate $\bar{A}_{j}$ are the amplitudes associated with the modes $\textbf{k}_{j}$ and $-\textbf{k}_{j}$ respectively.

System \eqref{eq:03} can be rewritten as

\begin{equation}\label{eq:E2}
\begin{split}
\frac{\partial x}{\partial t}&=f+(1-c_1)D_1\nabla^2 x,\\
\frac{\partial y}{\partial t}&=g+(1-c_2)D_2\nabla^2 y,
\end{split}
\end{equation}
where $f=\frac{rx}{1+(1-c_2)Ky}-\frac{rx^2}{(1-c_1)k}-\frac{\alpha (1-c_2) xy}{1+\alpha (1-c_2)hx}$ and $g=\frac{\theta\alpha (1-c_2) xy}{1+\alpha (1-c_2)hx}-dy$.\\

Now using Taylor's series expansion to system \eqref{eq:E2}, we obtain the linearized model as follows:

\begin{equation}\label{eq:E3}
\begin{split}
\frac{\partial x}{\partial t}&=x\frac{\partial f}{\partial x}+y\frac{\partial f}{\partial y} + \sum_{j=2}^{3} \frac{1}{j\,!}\bigg(x\frac{\partial }{\partial x}+y\frac{\partial }{\partial y}\bigg)^{j}f +(1-c_1)D_1\nabla^2 x,\\
\frac{\partial y}{\partial t}&=x\frac{\partial g}{\partial x}+y\frac{\partial g}{\partial y} + \sum_{j=2}^{3} \frac{1}{j\,!}\bigg(x\frac{\partial }{\partial x}+y\frac{\partial }{\partial y}\bigg)^{j}g +(1-c_2)D_2\nabla^2 y.
\end{split}
\end{equation}

System \eqref{eq:E3} can be rewritten in the following form:

\begin{align}\label{eq:E4}
\frac{\partial }{\partial t}\textbf{X}=\textbf{L}\textbf{X}+\textbf{P},
\end{align}
where
\begin{align}\label{eq:E5}
\textbf{X}=\left( {\begin{array}{c}
	x  \\
	y  
	\end{array} } \right),
\end{align}

\begin{align}\label{eq:E6}
\textbf{L}=\left( {\begin{array}{cc}
\frac{\partial f}{\partial x} + (1-c_1)D_1\nabla^2 &  \frac{\partial f}{\partial y} \\
\frac{\partial g}{\partial x} & \frac{\partial g}{\partial y} + (1-c_2)D_2\nabla^2 
	\end{array} } \right),
\end{align}
and
 
\begin{align}\label{eq:E7}
\textbf{P}=\left( {\begin{array}{c}
 \sum_{j=2}^{3} \frac{1}{j\,!}\bigg(x\frac{\partial }{\partial x}+y\frac{\partial }{\partial y}\bigg)^{j}f  \\
\sum_{j=2}^{3} \frac{1}{j\,!}\bigg(x\frac{\partial }{\partial x}+y\frac{\partial }{\partial y}\bigg)^{j}g 
	\end{array} } \right).
\end{align}

$K$ is our bifurcation parameter and we make series expansions near the Turing threshold parameter $K_T$ in the following:

\begin{align}\label{eq:E8}
\textbf{X}=\left( {\begin{array}{c}
	x  \\
	y  
	\end{array} } \right) = \epsilon\left( {\begin{array}{c}
	x_1  \\
	y_1  
	\end{array} } \right)+\epsilon^2\left( {\begin{array}{c}
	x_2  \\
	y_2  
	\end{array} } \right)+\epsilon^3\left( {\begin{array}{c}
	x_3  \\
	y_3  
	\end{array} } \right)+o(\epsilon^3), \mbox{where}~ |\epsilon|<<1,
\end{align}

such that

\begin{align}\label{eq:E9}
K_T-K= \epsilon K_1+\epsilon^2 K_2+\epsilon^3 K_3+o(\epsilon^3).
\end{align}

\begin{align}\label{eq:E10}
\frac{\partial }{\partial t}= \epsilon \frac{\partial }{\partial T_{1}}+\epsilon^2 \frac{\partial }{\partial T_{2}}+o(\epsilon^2), ~~~ \mbox{where},~ T_1=\epsilon t, ~ T_2=\epsilon^2 t.
\end{align}
Here, $T_1=\epsilon t, ~ T_2=\epsilon^2 t$ are the two time scales of the system.

\begin{align}\label{eq:E11}
\frac{\partial A}{\partial t}= \epsilon \frac{\partial A_{1}}{\partial T_{1}}+\epsilon^2 \frac{\partial A_{2}}{\partial T_{2}}+o(\epsilon^2),~~ \mbox{where}~ A= \mbox{Amplitude}.
\end{align}
\\
Now if we linearize each $f_{x_{i}x_{j}}$ and $g_{x_{i}x_{j}}$ by using Taylor's series expansion at $K=K_T$, we get,

\begin{equation}\label{eq:E12}
\begin{split}
f_{x_{i}x_{j}} &=f^T_{x_{i}x_{j}}-(\epsilon K_1+\epsilon^2 K_2+\epsilon^3 K_3)f^{'T}_{x_{i}x_{j}}+o(\epsilon^2), \\
g_{x_{i}x_{j}} &=g^T_{x_{i}x_{j}}-(\epsilon K_1+\epsilon^2 K_2+\epsilon^3 K_3)g^{'T}_{x_{i}x_{j}}+o(\epsilon^2),
\end{split}
\end{equation}
where $f^T_{x_{i}x_{j}}$ and $g^T_{x_{i}x_{j}}$ are the values of $f_{x_{i}x_{j}}$ and $g_{x_{i}x_{j}}$ at $K=K_T,$ respectively, and $f^{'T}_{x_{i}x_{j}}=\frac{d}{dK}(f_{x_{i}x_{j}})$ and $g^{'T}_{x_{i}x_{j}}=\frac{d}{dK}(g_{x_{i}x_{j}})$ at  $K=K_T$ 
[Here we denote $f_{x_{i}x_{j}}=\frac{\partial }{\partial x_{j}}(\frac{\partial f}{\partial x_{i}})$ and $x_1=x$ and $x_2=y$]. \\
\\
Now we use the trick of \eqref{eq:E12} and expand the terms of $\textbf{P}$ by using Taylor's series as follows: 

\begin{align}\label{eq:E13}
\textbf{P}= \epsilon^2 \textbf{P}_{2}+\epsilon^3 \textbf{P}_{3}+o(\epsilon^3),
\end{align}
where
\begin{align}\label{eq:E14}
\textbf{P}_{2}=\left( {\begin{array}{c}
\frac{1}{2}(x^2_{1}f^T_{xx}+y^2_{1}f^T_{yy}+2x_{1}y_{1}f^T_{xy}) \\ \\
\frac{1}{2}(x^2_{1}g^T_{xx}+y^2_{1}g^T_{yy}+2x_{1}y_{1}g^T_{xy}) 
	\end{array} } \right).
\end{align} 
and

\begin{align}\label{eq:E15}
\textbf{P}_{3}=\left( {\begin{array}{c}
	(x_{1}x_{2}f^T_{xx}+y_{1}y_{2}f^T_{yy}+(x_{1}y_{2}+x_{2}y_{1})f^T_{xy})+\frac{1}{6}((x_1\frac{\partial }{\partial x}+y_1\frac{\partial }{\partial y})^{3}f^T)-\frac{K_1}{2}((x_1\frac{\partial }{\partial x}+y_1\frac{\partial }{\partial y})^{2}f^{'T}) \\ \\
	(x_{1}x_{2}g^T_{xx}+y_{1}y_{2}g^T_{yy}+(x_{1}y_{2}+x_{2}y_{1})g^T_{xy})+\frac{1}{6}((x_1\frac{\partial }{\partial x}+y_1\frac{\partial }{\partial y})^{3}g^T)-\frac{K_1}{2}((x_1\frac{\partial }{\partial x}+y_1\frac{\partial }{\partial y})^{2}g^{'T}) 
	\end{array} } \right).
\end{align} 
$\textbf{P}_{2}$ and $\textbf{P}_{3}$ are terms associated with second and third order of $\epsilon$ of nonlinear term $\textbf{P}.$ At the same time, we linearize the operator $\textbf{L}$ of \eqref{eq:E4} at $K=K_T$ and we obtain the following relations:

\begin{align}\label{eq:E16}
\textbf{L}=\textbf{L}_{T}+(K_T-K)\textbf{M},
\end{align}
where
\begin{align}\label{eq:E17}
\textbf{L}_{T}=\left( {\begin{array}{cc}
	\frac{\partial f^T}{\partial x} + (1-c_1)D_1\nabla^2 &  \frac{\partial f^T}{\partial y} \\
	\frac{\partial g^T}{\partial x} & \frac{\partial g^T}{\partial y} + (1-c_2)D_2\nabla^2 
	\end{array} } \right),
\end{align}
and
\begin{align}\label{eq:E18}
\textbf{M}=\left( {\begin{array}{cc}
	m_{11} & m_{12} \\
	m_{21} & m_{22} 
	\end{array} } \right).
\end{align}
At $E^*$ with $K=K_T,$ each component of $\textbf{M}$ can be described as follows:
\begin{center}
	\begin{minipage}[b]{0.5\textwidth}
		\centering
		\begin{align*}
		{m_{11}}&=\frac{d}{dK}\bigg(f_x(x^*,y^*)\bigg) \\
		{m_{21}}&=\frac{d}{dK}\bigg(g_x(x^*,y^*)\bigg)
		\end{align*}
	\end{minipage}%
	\begin{minipage}[b]{0.5\textwidth}
		\raggedleft
		\begin{align*}
		{m_{12}}&=\frac{d}{dK}\bigg(f_y(x^*,y^*)\bigg) \\
		{m_{22}}&=\frac{d}{dK}\bigg(g_y(x^*,y^*)\bigg)
		\end{align*}
	\end{minipage}
\end{center} 
Now substituting Eqs. \eqref{eq:E8}-\eqref{eq:E18} into Eq. \eqref{eq:E4} and comparing the coefficients of $\epsilon^j,~~\forall j=1,2~ \mbox{and}~ 3,$ we obtain the following relations:
\begin{subequations}\label{eq:19}
	\begin{align}
	\textbf{L}_{T}\left( {\begin{array}{c}
		x_1  \\
		y_1  
		\end{array} } \right) &=0  \label{eq:19A}\\ 
		\textbf{L}_{T}\left( {\begin{array}{c}
		x_2  \\
		y_2  
		\end{array} } \right) &=\frac{\partial}{\partial T_1}\left( {\begin{array}{c}
		x_1  \\
		y_1  
		\end{array} } \right)-K_1\textbf{M}\left( {\begin{array}{c}
		x_1  \\
		y_1  
		\end{array} } \right)-\textbf{P}_{2}   \label{eq:19B}\\
		\textbf{L}_{T}\left( {\begin{array}{c}
		x_3  \\
		y_3  
		\end{array} } \right) &=\frac{\partial}{\partial T_1}\left( {\begin{array}{c}
		x_2  \\
		y_2  
		\end{array} } \right)+\frac{\partial}{\partial T_2}\left( {\begin{array}{c}
		x_1  \\
		y_1  
		\end{array} } \right)-K_1\textbf{M}\left( {\begin{array}{c}
		x_2  \\
		y_2  
		\end{array} } \right)-K_2\textbf{M}\left( {\begin{array}{c}
		x_1  \\
		y_1  
		\end{array} } \right)-\textbf{P}_{3}   \label{eq:19C}
	\end{align}
\end{subequations}
Note that in the above, $\textbf{L}_{T}$ is a linear and $\left( {\begin{array}{c}
		x_1  \\
		y_1  
\end{array} } \right)$ corresponds to linear combination of eigenvectors associated with the eigenvalue 0. Considering the first order in $\epsilon$, the solution of \eqref{eq:19A} can be obtained as:
\begin{align}\label{eq:E20}
\left( {\begin{array}{c}
	x_1  \\
	y_1  
	\end{array} } \right)= \left( {\begin{array}{c}
	\xi_1  \\
	1  
	\end{array} } \right)\bigg(\sum_{j=1}^{3}W_{j}\mbox{exp}(i\textbf{k}_{j}.\textbf{r})+c.c.\bigg),
\end{align}
where $\xi_1=-\frac{f^T_y}{f^T_x-(1-c_1)D_1\kappa^2_T}=-\frac{g^T_y-(1-c_2)D_2\kappa^2_T}{g^T_x}$, $W_j$ are the amplitudes of the modes $\mbox{exp}(i\textbf{k}_{j}.\textbf{r}), j = 1, 2, 3,$ and \mbox{c.c.} represents complex conjugate. \\ 
Using Fredholm solvability condition, the R.H.S. of \eqref{eq:19B} must be orthogonal with the zero eigenvector to the  adjoint operator of $\textbf{L}_{T},$ denoted as $\textbf{L}^{+}_{T}.$
Therefore we obtain the following equation:
\begin{equation}\label{eq:E21}
 \textbf{L}^{+}_{T}\left( {\begin{array}{c}
 	x_2  \\
 	y_2  
 	\end{array} } \right) =0.
\end{equation}
The solution form of the eigenvector corresponding to zero eigenvalue of  Eq.\eqref{eq:E21} is given by,
\begin{align}\label{eq:E22}
\left( {\begin{array}{c}
	x_2  \\
	y_2  
	\end{array} } \right)= \left( {\begin{array}{c}
	1  \\
	\xi_2  
	\end{array} } \right)\bigg(\sum_{j=1}^{3}W_{j}\mbox{exp}(i\textbf{k}_{j}.\textbf{r})+c.c.\bigg),
\end{align}
where $\xi_2=-\frac{f^T_x-(1-c_1)D_1\kappa^2_T}{g^T_x}=-\frac{f^T_y}{g^T_y-(1-c_2)D_2\kappa^2_T}$.\\
We substitute Eq.\eqref{eq:E20} into Eq.\eqref{eq:19B} and  we get,

\begin{equation}\label{eq:E23}
\begin{split}
	\textbf{L}_{T}\left( {\begin{array}{c}
	x_2  \\\\
	y_2  
	\end{array} } \right) &=\left( {\begin{array}{c}
	\xi_1\frac{\partial W_{j}}{\partial T_1}  \\\\
	\frac{\partial W_{j}}{\partial T_1}  
	\end{array} } \right)-K_1\left( {\begin{array}{c}
	\xi_1m_{11}W_j+m_{12}W_j  \\\\
	\xi_1m_{21}W_j+m_{22}W_j  
	\end{array} } \right)-\left( {\begin{array}{c}
	u_1\bar{W}_l\bar{W}_m  \\\\
	u_2\bar{W}_l\bar{W}_m  
	\end{array} } \right) \\
   & =\left( {\begin{array}{c}
   	F^{j}_{x}  \\\\
   	F^{j}_{y}  
   	\end{array} } \right) ~~~ (say)
\end{split}
\end{equation}
where $j\neq l,m;~ l\neq m;~\forall~ j,l,m=1,2~ \mbox{and}~ 3.$ For the sake of easy calculation, we consider,
\begin{equation}\label{eq:E24}
\begin{split}
u_1 &=(\xi^2_{1}f^T_{xx}+f^T_{yy}+2\xi_{1}f^T_{xy}), \\
u_2 &=(\xi^2_{1}g^T_{xx}+g^T_{yy}+2\xi_{1}g^T_{xy}).
\end{split}
\end{equation}
Since the left eigenvector with zero eigenvalue of \eqref{eq:E22} is orthogonal to $(F^{j}_{x}~~ F^{j}_{y})^T$, we have
\begin{equation}\label{eq:E25}
\begin{split}
\begin{pmatrix} x_2 & y_2 \end{pmatrix}\left( {\begin{array}{c}
	F^{j}_{x}\\ 
	F^{j}_{y}  
	\end{array} } \right)=0.
\end{split}
\end{equation}
Now we incorporate Eqs.\eqref{eq:E22}-\eqref{eq:E23} into Eq.\eqref{eq:E25} and compare the coefficients of $\mbox{exp}(i\textbf{k}_{j}.\textbf{r})$. Then we obtain the following equations for $j=1,2 ~\mbox{and}~ 3$:
\begin{equation}\label{eq:E26}
\begin{split}
(\xi_1+\xi_2)\frac{\partial W_1}{\partial T_1}&=K_1[(\xi_1m_{11}+m_{12})+\xi_2(\xi_1m_{21}+m_{22})]W_1+(u_1+\xi_2u_2)\bar{W}_2\bar{W}_3\\
(\xi_1+\xi_2)\frac{\partial W_2}{\partial T_1}&=K_1[(\xi_1m_{11}+m_{12})+\xi_2(\xi_1m_{21}+m_{22})]W_2+(u_1+\xi_2u_2)\bar{W}_1\bar{W}_3\\
(\xi_1+\xi_2)\frac{\partial W_3}{\partial T_1}&=K_1[(\xi_1m_{11}+m_{12})+\xi_2(\xi_1m_{21}+m_{22})]W_3+(u_1+\xi_2u_2)\bar{W}_1\bar{W}_2
\end{split}
\end{equation}
We write $\begin{pmatrix} x_2 & y_2 \end{pmatrix}^T$ in \eqref{eq:19B} as a sum of Fourier series terms as follows:
\begin{equation}\label{eq:E27}
\begin{split}
\left( {\begin{array}{c}
	x_2  \\
	y_2  
	\end{array} } \right) &= \left( {\begin{array}{c}
	x^0_{2}  \\
	y^0_{2}  
	\end{array} } \right)+\sum_{j=1}^{3}\left( {\begin{array}{c}
	x^j_{2}  \\
	y^j_{2}  
	\end{array} } \right)\mbox{exp}(i\textbf{k}_{j}.\textbf{r})+\sum_{j=1}^{3}\left( {\begin{array}{c}
	x^{jj}_{2}  \\
	y^{jj}_{2}  
	\end{array} } \right)\mbox{exp}(2i\textbf{k}_{j}.\textbf{r})+\left( {\begin{array}{c}
	x^{12}_{2}  \\
	y^{12}_{2}  
	\end{array} } \right)\mbox{exp}(i(\textbf{k}_{1}-\textbf{k}_{2}).\textbf{r})\\
    & +\left( {\begin{array}{c}
	x^{23}_{2}  \\
	y^{23}_{2}  
	\end{array} } \right)\mbox{exp}(i(\textbf{k}_{2}-\textbf{k}_{3}).\textbf{r})+\left( {\begin{array}{c}
	x^{31}_{2}  \\
	y^{31}_{2}  
	\end{array} } \right)\mbox{exp}(i(\textbf{k}_{3}-\textbf{k}_{1}).\textbf{r})+c.c.
\end{split}
\end{equation}
Now we substitute Eq. \eqref{eq:E27} into Eq. \eqref{eq:19B} and equate the coefficients of $\mbox{exp(0)}$, $\mbox{exp}(2i\textbf{k}_{j}.\textbf{r})$ and $\mbox{exp}(i(\textbf{k}_{3}-\textbf{k}_{1})$. Then we obtain the following results for $j= 1,2 ~\mbox{and}~3$,

\begin{equation}\label{eq:E28}
\begin{split}
\left( {\begin{array}{c}
	x^0_{2}  \\\\
	y^0_{2}  
	\end{array} } \right) &= \left( {\begin{array}{c}
	\frac{u_2f^{T}_{y}-u_1g^{T}_{y}}{f^{T}_{x}g^{T}_{y}-f^{T}_{y}g^{T}_{x}} \\\\
	\frac{u_2f^{T}_{x}-u_1g^{T}_{x}}{f^{T}_{y}g^{T}_{x}-f^{T}_{x}g^{T}_{y}}  
	\end{array} } \right) (|W_1|^2+|W_2|^2+|W_3|^2)=\left( {\begin{array}{c}
	R_1 \\\\
	R_2  
	\end{array} } \right) (|W_1|^2+|W_2|^2+|W_3|^2)     ~~~ (say)\\\\
\left( {\begin{array}{c}
	x^{jj}_{2}  \\\\
	y^{jj}_{2}  
	\end{array} } \right) &=\left( {\begin{array}{c}
	\frac{u_2f^{T}_{y}-u_1(g^{T}_{y}-4(1-c_2)D_2\kappa^{2}_{T})}{(f^{T}_{x}-4(1-c_1)D_1\kappa^{2}_{T})(g^{T}_{y}-4(1-c_2)D_2\kappa^{2}_{T})-f^{T}_{y}g^{T}_{x}} \\\\
	\frac{u_2(f^{T}_{x}-4(1-c_1)D_1\kappa^{2}_{T})-u_1g^{T}_{x}}{f^{T}_{y}g^{T}_{x}-(f^{T}_{x}-4(1-c_1)D_1\kappa^{2}_{T})(g^{T}_{y}-4(1-c_2)D_2\kappa^{2}_{T})}  
	\end{array} } \right) \frac{(W_j^2)}{2}=\left( {\begin{array}{c}
	S_1 \\\\
	S_2 
	\end{array} } \right) \frac{(W_j^2)}{2}    ~~~ (say)\\\\
\left( {\begin{array}{c}
	x^{12}_{2}  \\\\
	y^{12}_{2}  
	\end{array} } \right) &=\left({\begin{array}{c}
	\frac{u_2f^{T}_{y}-u_1(g^{T}_{y}-3(1-c_2)D_2\kappa^{2}_{T})}{(f^{T}_{x}-3(1-c_1)D_1\kappa^{2}_{T})(g^{T}_{y}-3(1-c_2)D_2\kappa^{2}_{T})-f^{T}_{y}g^{T}_{x}} \\\\
	\frac{u_2(f^{T}_{x}-3(1-c_1)D_1\kappa^{2}_{T})-u_1g^{T}_{x}}{f^{T}_{y}g^{T}_{x}-(f^{T}_{x}-3(1-c_1)D_1\kappa^{2}_{T})(g^{T}_{y}-3(1-c_1)D_2\kappa^{2}_{T})}  
		\end{array} } \right) W_1\bar{W}_2=\left({\begin{array}{c}
		T_1 \\\\
		T_2  
		\end{array} } \right) W_1\bar{W}_2 ~~~ (say) \\\\
\left( {\begin{array}{c}
	x^{23}_{2}  \\\\
	y^{23}_{2}  
	\end{array} } \right) &=\left({\begin{array}{c}
	\frac{u_2f^{T}_{y}-u_1(g^{T}_{y}-3(1-c_2)D_2\kappa^{2}_{T})}{(f^{T}_{x}-3(1-c_1)D_1\kappa^{2}_{T})(g^{T}_{y}-3(1-c_2)D_2\kappa^{2}_{T})-f^{T}_{y}g^{T}_{x}} \\\\
	\frac{u_2(f^{T}_{x}-3(1-c_1)D_1\kappa^{2}_{T})-u_1g^{T}_{x}}{f^{T}_{y}g^{T}_{x}-(f^{T}_{x}-3(1-c_1)D_1\kappa^{2}_{T})(g^{T}_{y}-3(1-c_2)D_2\kappa^{2}_{T})}  
	\end{array} } \right) W_2\bar{W}_3=\left({\begin{array}{c}
	T_1 \\\\
	T_2  
	\end{array} } \right) W_2\bar{W}_3 ~~~ (say) \\\\
\left( {\begin{array}{c}
	x^{31}_{2}  \\\\
	y^{31}_{2}  
	\end{array} } \right) &=\left({\begin{array}{c}
	\frac{u_2f^{T}_{y}-u_1(g^{T}_{y}-3(1-c_2)D_2\kappa^{2}_{T})}{(f^{T}_{x}-3(1-c_1)D_1\kappa^{2}_{T})(g^{T}_{y}-3(1-c_2)D_2\kappa^{2}_{T})-f^{T}_{y}g^{T}_{x}} \\\\
	\frac{u_2(f^{T}_{x}-3(1-c_1)D_1\kappa^{2}_{T})-u_1g^{T}_{x}}{f^{T}_{y}g^{T}_{x}-(f^{T}_{x}-3(1-c_1)D_1\kappa^{2}_{T})(g^{T}_{y}-3(1-c_2)D_2\kappa^{2}_{T})}  
	\end{array} } \right) W_3\bar{W}_1=\left({\begin{array}{c}
	T_1 \\\\
	T_2  
	\end{array} } \right) W_3\bar{W}_1 ~~~ (say)
\end{split}
\end{equation}
We incorporate Eq.\eqref{eq:E20} and Eq.\eqref{eq:E27} into Eq.\eqref{eq:E15} and after carrying out  long algebric manipulation, Eq.\eqref{eq:E15} can be rewritten as follows:
\begin{align}\label{eq:E29}
\textbf{P}_{3}=\left( {\begin{array}{c}
	H_1|W_1|^2W_1+H_2(|W_2|^2+|W_3|^2)W_1-H_3\bar{W_2}\bar{W_3} \\ \\
    I_1|W_1|^2W_1+I_2(|W_2|^2+|W_3|^2)W_1-I_3\bar{W_2}\bar{W_3}
	\end{array} } \right),
\end{align} 
where,
\begin{equation}\label{eq:E30}
\begin{split}
H_1 &=\bigg(R_1+\frac{S_1}{2}\bigg)(\xi_1f^{T}_{xx}+f^{T}_{xy})+\bigg(R_2+\frac{S_2}{2}\bigg)(\xi_1f^{T}_{xy}+f^{T}_{yy})+\frac{1}{2}(\xi^{3}_1f^{T}_{xxx}+3\xi^{2}_1f^{T}_{xxy}+3\xi_1f^{T}_{xyy}+f^{T}_{yyy}),\\
H_2 &=(R_1+{T_1})(\xi_1f^{T}_{xx}+f^{T}_{xy})+(R_2+{T_2})(\xi_1f^{T}_{xy}+f^{T}_{yy})+(\xi^{3}_1f^{T}_{xxx}+3\xi^{2}_1f^{T}_{xxy}+3\xi_1f^{T}_{xyy}+f^{T}_{yyy}),\\
H_3 &=K_1(\xi^{2}_1f^{'T}_{xx}+2\xi_1f^{'T}_{xy}+f^{'T}_{yy}), \\
I_1 &=\bigg(R_1+\frac{S_1}{2}\bigg)(\xi_1g^{T}_{xx}+g^{T}_{xy})+\bigg(R_2+\frac{S_2}{2}\bigg)(\xi_1g^{T}_{xy}+g^{T}_{yy})+\frac{1}{2}(\xi^{3}_1g^{T}_{xxx}+3\xi^{2}_1g^{T}_{xxy}+3\xi_1g^{T}_{xyy}+g^{T}_{yyy}),\\
I_2 &=(R_1+{T_1})(\xi_1g^{T}_{xx}+g^{T}_{xy})+(R_2+{T_2})(\xi_1g^{T}_{xy}+g^{T}_{yy})+(\xi^{3}_1g^{T}_{xxx}+3\xi^{2}_1g^{T}_{xxy}+3\xi_1g^{T}_{xyy}+g^{T}_{yyy}),\\
I_3 &=K_1(\xi^{2}_1g^{'T}_{xx}+2\xi_1g^{'T}_{xy}+g^{'T}_{yy}). 
\end{split}
\end{equation}
After using Eqs.\eqref{eq:E20}, \eqref{eq:E27} and \eqref{eq:E29} into Eq.\eqref{eq:19C} and comparing the coefficients of $\mbox{exp}(i\textbf{k}_{j}.\textbf{r}),~\forall j=1,2,3,$ one can check that Eq.\eqref{eq:19C} becomes
\begin{equation}\label{eq:E31}
\begin{split}
\textbf{L}_{T}\left( {\begin{array}{c}
	x_3  \\\\
	y_3  
	\end{array} } \right) &=\frac{\partial }{\partial T_1}\left( {\begin{array}{c}
      x^{j}_2      \\\\
	  y^{j}_2
	\end{array} } \right)+\left( {\begin{array}{c}
	\xi_1\frac{\partial W_{j}}{\partial T_2}  \\\\
	\frac{\partial W_{j}}{\partial T_2}  
	\end{array} } \right)-K_1\left( {\begin{array}{c}
	m_{11}x^{j}_2+m_{12}y^{j}_2  \\\\
	m_{21}x^{j}_2+m_{22}y^{j}_2  
	\end{array} } \right)-K_2\left( {\begin{array}{c}
	\xi_1m_{11}W_j+m_{12}W_j  \\\\
	\xi_1m_{21}W_j+m_{22}W_j  
	\end{array} } \right)\\
& -\left( {\begin{array}{c}
	H_1|W_1|^2W_1+H_2(|W_2|^2+|W_3|^2)W_1-H_3\bar{W_2}\bar{W_3} \\ \\
	I_1|W_1|^2W_1+I_2(|W_2|^2+|W_3|^2)W_1-I_3\bar{W_2}\bar{W_3}
	\end{array} } \right) \\
& =\left( {\begin{array}{c}
	G^{j}_{x}  \\\\
	G^{j}_{y}  
	\end{array} } \right) ~~~ (say).
\end{split}
\end{equation}
Let us assume that $x^{j}_2=\xi_1y^{j}_2,~\forall j=1,2,3.$ Using Fredholm solvability condition for Eq.\eqref{eq:19C}, $\left( {\begin{array}{c}
	G^{j}_{x}  \\\\
	G^{j}_{y}  
	\end{array} } \right)$ must be orthogonal to the left eigenvectors of $\textbf{L}_{T}$ with zero eigenvalue. We obtain the following relations:
\begin{equation}\label{eq:E32}
\begin{split}
(\xi_1+\xi_2)\bigg(\frac{\partial y^{1}_2}{\partial T_1}+\frac{\partial W_1}{\partial T_2}\bigg) &=[(\xi_1m_{11}+m_{12})+\xi_2(\xi_1m_{21}+m_{22})](K_1y^{1}_2+K_2W_1)\\
& +[(H_1+I_1\xi_2)|W_1|^2+(H_2+I_2\xi_2)(|W_2|^2+|W_3|^2)]W_1-(H_3+I_3\xi_2)\bar{W_2}\bar{W_3},\\
(\xi_1+\xi_2)\bigg(\frac{\partial y^{2}_2}{\partial T_1}+\frac{\partial W_2}{\partial T_2}\bigg) &=[(\xi_1m_{11}+m_{12})+\xi_2(\xi_1m_{21}+m_{22})](K_1y^{2}_2+K_2W_2)\\
& +[(H_1+I_1\xi_2)|W_2|^2+(H_2+I_2\xi_2)(|W_3|^2+|W_1|^2)]W_2-(H_3+I_3\xi_2)\bar{W_3}\bar{W_1},\\
(\xi_1+\xi_2)\bigg(\frac{\partial y^{3}_2}{\partial T_1}+\frac{\partial W_3}{\partial T_2}\bigg) &=[(\xi_1m_{11}+m_{12})+\xi_2(\xi_1m_{21}+m_{22})](K_1y^{3}_2+K_2W_3)\\
& +[(H_1+I_1\xi_2)|W3|^2+(H_2+I_2\xi_2)(|W_1|^2+|W_2|^2)]W_3-(H_3+I_3\xi_2)\bar{W_1}\bar{W_2}.
\end{split}
\end{equation}
Now, the amplitude equations $A_j(j=1,2,3)$ can be expressed as follows:

\begin{align}\label{eq:E33}
A_j=\epsilon W_j+\epsilon^2y^{j}_{2}+o(\epsilon^3) ~~~\forall j=1,2,3.
\end{align} 
By the analysis of the symmetries, up to the third order in the perturbations, the spatio temporal evolutions of the amplitudes $A_j(j=1,2,3)$ are described through the amplitude equations. Now using the above equations, we obtain the amplitude equations from Eq. \eqref{eq:E11} as follows:
\begin{equation}\label{eq:E34}
\begin{split}
\tau_0\frac{\partial A_1}{\partial t} &=\mu A_1+h\bar{A_2}\bar{A_3}-[g_1|A_1|^2+g_2(|A_2|^2+|A_3|^2)]A_1,\\
\tau_0\frac{\partial A_2}{\partial t} &=\mu A_2+h\bar{A_3}\bar{A_1}-[g_1|A_2|^2+g_2(|A_3|^2+|A_1|^2)]A_2,\\
\tau_0\frac{\partial A_3}{\partial t} &=\mu A_3+h\bar{A_1}\bar{A_2}-[g_1|A_3|^2+g_2(|A_1|^2+|A_2|^2)]A_3,
\end{split}
\end{equation}
where,
\begin{equation}\label{eq:E35}
\begin{split}
\mu &=\frac{K_T-K}{K_T}, \\
\tau_0 &=\frac{\xi_1+\xi_2}{K_T[(\xi_1m_{11}+m_{12})+\xi_2(\xi_1m_{21}+m_{22})]}, \\
h &=-\frac{H_3+\xi_2I_3}{K_T[(\xi_1m_{11}+m_{12})+\xi_2(\xi_1m_{21}+m_{22})]}, \\
g_1 &= -\frac{H_1+\xi_2I_1}{K_T[(\xi_1m_{11}+m_{12})+\xi_2(\xi_1m_{21}+m_{22})]}, \\
g_2 &=-\frac{H_2+\xi_2I_2}{K_T[(\xi_1m_{11}+m_{12})+\xi_2(\xi_1m_{21}+m_{22})]}.
\end{split}
\end{equation} \\



\begin{remark}
Amplitude equations are also obtained by using bifurcation parameters $c_1$ and $c_2$.
\end{remark}

\section{Stability Analysis of Amplitude Equations}\label{S5}
The amplitude equations can be re-expressed in terms of magnitudes and phases, 
\begin{equation}\label{eq:E36}
A_j=\rho_j\mbox{exp}(i\varphi_j)
\end{equation} with equal phase angle $\phi=\frac{2\pi}{3}$ and they can be decomposed to modes $\rho_j=|A_j|,$ where $\phi=\sum_{j=1}^{3}\varphi_j.$ \\
Now we have from \eqref{eq:E34},
\begin{equation}\label{eq:E37}
\begin{split}
\frac{\partial \phi}{\partial t}=\sum_{j=1}^{3}\frac{\partial \varphi_j}{\partial t}=\frac{1}{i}\sum_{j=1}^{3}\frac{1}{A_j}\frac{\partial A_j}{\partial t} &=-\frac{h}{\tau_0}(\rho^{2}_2\rho^{2}_3+\rho^{2}_3\rho^{2}_1+\rho^{2}_1\rho^{2}_2)(\rho_1\rho_2\rho_3)^{-1}(i\cos\phi+\sin\phi)\\ 
& +(i\tau_0)[3\mu-(g_1+2g_2).(|A_1|^2+|A_2|^2+|A_3|^2)].
\end{split}
\end{equation} 
Taking the real part, we have
\begin{equation}\label{eq:E38}
\begin{split}
\tau_0\frac{\partial \phi}{\partial t}= -h(\rho^{2}_2\rho^{2}_3+\rho^{2}_3\rho^{2}_1+\rho^{2}_1\rho^{2}_2)(\rho_1\rho_2\rho_3)^{-1}\sin\varphi.
\end{split}
\end{equation} 
Substituting Eq.\eqref{eq:E36} in the amplitude equations \eqref{eq:E34} and considering the real parts, we obtain the following equations:
\begin{equation}\label{eq:E39}
\begin{split}
\tau_0\frac{\partial \rho_1}{\partial t} &= \mu\rho_1+h\rho_2\rho_3\cos\phi-g_1\rho^3_1-g_2(\rho^2_2+\rho^2_3)\rho_1, \\
\tau_0\frac{\partial \rho_2}{\partial t} &= \mu\rho_2+h\rho_3\rho_1\cos\phi-g_1\rho^3_2-g_2(\rho^3_3+\rho^2_1)\rho_2, \\
\tau_0\frac{\partial \rho_3}{\partial t} &= \mu\rho_3+h\rho_1\rho_2\cos\phi-g_1\rho^3_3-g_2(\rho^2_1+\rho^2_2)\rho_3.
\end{split}
\end{equation} 
For $\phi=0~ \mbox{and}~ \phi=\pi$, $h$ becomes positive and negative respectively i.e.  when $h>0$, $H_0$ pattern is stable and for $h<0$, $H_{\pi}$ pattern is stable. Therefore Eq.\eqref{eq:E39} has the following form:
\begin{equation}\label{eq:E40}
\begin{split}
\tau_0\frac{\partial \rho_1}{\partial t} &= \mu\rho_1+|h|\rho_2\rho_3-g_1\rho^3_1-g_2(\rho^2_2+\rho^2_3)\rho_1, \\
\tau_0\frac{\partial \rho_2}{\partial t} &= \mu\rho_2+|h|\rho_3\rho_1-g_1\rho^3_2-g_2(\rho^3_3+\rho^2_1)\rho_2, \\
\tau_0\frac{\partial \rho_3}{\partial t} &= \mu\rho_3+|h|\rho_1\rho_2-g_1\rho^3_3-g_2(\rho^2_1+\rho^2_2)\rho_3.
\end{split}
\end{equation} 
Eq.\eqref{eq:E40} may have following types of solutions:
\subsection{Stability of Stationary State}
The stationary state $(O)$ is given by
$\rho_1=\rho_2=\rho_3=0$ which is stable for $\mu<\mu_o=0$ and unstable for $\mu>\mu_o.$

\subsection{Stability of Stripe Structures}
Stripe patterns $(S)$, given by
$\rho_1=\sqrt{\frac{\mu}{g_1}}\neq0,~\rho_2=\rho_3=0.$ \\
Let us consider $\rho_j=\rho_s+\delta\rho_j,~(j=1,2,3).$ Now place $\rho_1=\rho_s+\delta\rho_1, \rho_2=\delta\rho_2 ~\mbox{and}~ \rho_3=\delta\rho_3$ in Eq.\eqref{eq:E40}(neglecting terms of order $\delta^2$) and substituting the steady state condition, we obtain the following equations:
\begin{equation}\label{eq:E41}
\begin{split}
\tau_0\frac{\partial }{\partial t}(\delta\rho_1) &= (\mu-3g_1\rho^2_s)\delta\rho_1, \\
\tau_0\frac{\partial }{\partial t}(\delta\rho_2) &= (\mu-g_2\rho^2_s)\delta\rho_2+|h|\rho_s\delta\rho_3, \\
\tau_0\frac{\partial }{\partial t}(\delta\rho_3) &= |h|\rho_s\delta\rho_2+(\mu-g_2\rho^2_s)\delta\rho_3
\end{split}
\end{equation} 
The characteristic equation of \eqref{eq:E41} for the stability of stripe structures is as follows:
\begin{equation}\label{eq:E42}
\begin{split}
(\mu-3g_1\rho^2_s-\lambda)\{(\mu-g_2\rho^2_s-\lambda)^2-|h|^2\rho^2_s\}=0,
\end{split}
\end{equation} 
where $\lambda's$ are the associated eigenvalues. Now replace $\rho_s=\sqrt{\frac{\mu}{g_1}},$ in Eq.\eqref{eq:E42}, we get 
\begin{equation}\label{eq:E43}
\begin{split}
(-2\mu-\lambda)\bigg\{(\mu-\frac{g_2}{g_1}\mu-\lambda)^2-|h|^2\frac{\mu}{g_1}\bigg\}=0.
\end{split}
\end{equation}
Therefore, $\lambda_1=-2\mu$ ~ and ~ $\lambda_{2,3}=\mu\bigg(1-\frac{g_2}{g_1}\bigg)\mp|h|\sqrt{\frac{\mu}{g_1}}.$
Stable stripe patterns arise only when all the eigenvalues are negative. This implies 
\begin{equation}\label{eq:E44}
\begin{split}
\mu>\frac{|h|^2g_1}{(g_1-g_2)^2}=\mu_s.
\end{split}
\end{equation}
Therefore the stripe patterns are stable for $\mu>\mu_s$ and unstable for $\mu<\mu_s.$

\subsection{Stability of Hexagonal Structures}

We substitute $\rho_1=\rho_2=\rho_3=\rho_h$ in Eq.\eqref{eq:E40} and get,
\begin{equation}\label{eq:E45}
\begin{split}
\rho^\pm_h=\frac{\{|h|\pm\sqrt{(h^2+4(g_1+2g_2)\mu)}\}}{2(g_1+2g_2)}
\end{split}
\end{equation}
For the real value of $\rho^\pm_h$, the discriminant should be positive, which gives
\begin{equation}\label{eq:E46}
\begin{split}
\mu>\frac{-h^2}{4(g_1+2g_2)}=\mu_{h_1}
\end{split}
\end{equation}
In order to study the linear stability, we again assume $\rho_j=\rho_h+\delta\rho_j,~(j=1,2,3).$ Now we substitute $\rho_1=\rho_h+\delta\rho_1, \rho_2=\rho_h+\delta\rho_2 ~\mbox{and}~ \rho_3=\rho_h+\delta\rho_3$ in Eq.\eqref{eq:E40}(neglecting terms of order $\delta^2$) and substituting the steady state condition, we obtain the following equations:
\begin{equation}\label{eq:E47}
\begin{split}
\tau_0\frac{\partial }{\partial t}(\delta\rho_1) &= a_1\delta\rho_1+a_2\delta\rho_2+a_2\delta\rho_3, \\
\tau_0\frac{\partial }{\partial t}(\delta\rho_2) &= a_2\delta\rho_1+a_1\delta\rho_2+a_2\delta\rho_3, \\
\tau_0\frac{\partial }{\partial t}(\delta\rho_3) &= a_2\delta\rho_1+a_2\delta\rho_2+a_1\delta\rho_3,
\end{split}
\end{equation} 
where
\begin{equation}\label{eq:E48}
\begin{split}
a_1 &= \mu-(3g_1+2g_2)\rho^2_h, \\
a_2 &= |h|\rho_h-2g_2\rho^2_h. 
\end{split}
\end{equation} 
The characteristic equation for the amplitude variation of \eqref{eq:E47} is given by
\begin{equation}\label{eq:E49}
\begin{split}
\lambda^3-3a_1\lambda^2+3(a^2_1-a^2_2)\lambda+(3a_1a^2_2-a^3_1-2a^3_2)=0,
\end{split}
\end{equation}
where $\lambda's$ are eigenvalues. After solving the above cubic equation, we get
\begin{equation}\label{eq:E50}
\begin{split}
\lambda_1 =a_1+2a_2;~~\lambda_{2,3}=a_1-a_2.
\end{split}
\end{equation} 
Now we explore the criterion of linear stability based on the analysis of the eigenvalues, namely $\rho^{+}_h ~ \mbox{and} \rho^{-}_h$. 
\subsection*{\textbf{Case $\rho_h=\rho^{+}_h$:}}
Substitute the value of $\rho^{+}_h$ from \eqref{eq:E45} in $\lambda_1=a_1+2a_2=\mu-3(g_1+2g_2)\rho^{+2}_h+2|h|\rho_h^{+}$, we get
\begin{equation}\label{eq:E51}
\begin{split}
\lambda_1 =-\frac{4\mu(g_1+2g_2)+h^2+|h|\sqrt{(h^2+4\mu(g_1+2g_2))}}{2(g_1+2g_2)}<0.
\end{split}
\end{equation} 
Now we investigate the condition such that eigenvalues $\lambda_{2} ~ \mbox{and}~ \lambda_{3}$ become negative so that one is assured of the stability of hexagonal patterns. Assume and impose $\lambda_{2}=\lambda_{3}<0,$ we obtain
\begin{equation}\label{eq:E52}
\begin{split}
\lambda_2=\lambda_3 &=\mu-3g_1\rho^{+2}_h-|h|\rho_h^{+}\\
& =\frac{4\mu(g_1+2g_2)^2-6g_1h^2-12g_1\mu(g_1+2g_2)-2h^2(g_1+2g_2)-4|h|(2g_1+g_2)\sqrt{(h^2+4\mu(g_1+2g_2))}}{4(g_1+2g_2)^2} \\
& <0
\end{split}
\end{equation}
On simplification of \eqref{eq:E52}, we get
\begin{equation}\label{eq:E53}
\begin{split}
\mu<\frac{h^2(2g_1+g_2)}{(g_1-g_2)^2}=\mu_{h_2}.
\end{split}
\end{equation}
Therefore hexagonal patterns are stable if $\mu<\mu_{h_2}.$

\subsection*{\textbf{Case $\rho_h=\rho^{-}_h$:}}
Again we repeat the process along similar lines and substitute the value of $\rho^{-}_h$ from \eqref{eq:E45} in $\lambda_1=a_1+2a_2=\mu-3(g_1+2g_2)\rho^{-2}_h+2|h|\rho_h^{-}$ and we get,
\begin{equation}\label{eq:E54}
\begin{split}
\lambda_1 =\frac{-4\mu(g_1+2g_2)-h^2+|h|\sqrt{(h^2+4\mu(g_1+2g_2))}}{2(g_1+2g_2)}>0.
\end{split}
\end{equation} 
After a comparison of \eqref{eq:E54} with Eq.\eqref{eq:E51}, we find that the value of $\lambda_1$ becomes positive. Furthermore, we investigate the conditions on the eigenvalues $\lambda_{2} ~ \mbox{and}~ \lambda_{3},$ and we obtain
\begin{equation}\label{eq:E55}
\begin{split}
\lambda_2=\lambda_3 &=\mu-3g_1\rho^{-2}_h-|h|\rho_h^{-}\\
& =\frac{4\mu(g_1+2g_2)^2-6g_1h^2-12g_1\mu(g_1+2g_2)-2h^2(g_1+2g_2)+4|h|(2g_1+g_2)\sqrt{(h^2+4\mu(g_1+2g_2))}}{4(g_1+2g_2)^2} \\
& >0
\end{split}
\end{equation}
On comparing Eq.\eqref{eq:E55} with Eq.\eqref{eq:E52}, we find that Eq.\eqref{eq:E55} becomes positive. Now for the case $\rho_h=\rho^{-}_h$, all the three eigenvalues are positive, confirming that the hexagonal patterns are unstable.

\section{Numerical Simulation}\label{S6}
Next, we present the results on the numerical simulation of model system \eqref{eq:03}, which suggests the importance of the fear effect on the population in the spatio temporal domain. In the following we will enumerate the numerical analysis corresponding to previous sections one by one.

\subsection{Stability of Equilibrium Points}
In Figures (\ref{fig:1}~(a),(b),(c),(d),(e)), we compare the quantitative changes between carrying capacity, fear effect, functional response and self diffusion coefficients with $c_1$, $c_2$, $c_2$, $c_1$ and $c_2$ of the model system \eqref{eq:03}, respectively. We portray Hopf and transcritical bifurcations with respect to two parameters of all possible combinations of system \eqref{eq:2} in Figure \ref{fig:2}. Stability properties of the equilibrium points $E_0$, $E_1$ and $E^{*}$ are found to depend upon nine parameters, i.e. $ r, \alpha, \theta, d, h, k, K, c_2, c_1 $. We depict all possible combinations in different parametric phase-planes in Figure \ref{fig:2}. Parameters in the $y$-axis of extreme left of the sub figures of each row are same for the $y$-axis of the other sub figures of the same row. Similarly, parameters in the $x$-axis of bottom sub figures of each column is same for the $x$-axis of the previous sub figures of the same column. When two parameters are varying in the two axes then the remaining parameters are kept as constants. The parameters are given the values as $r=2, \alpha=0.5, \theta=0.7, d=0.55, h=1, k=90, K=0.5, c_2=0.5, c_1=0.1$. Red and blue lines indicate the Hopf and transcritical bifurcations w.r.t two parameters. In such combinations we discuss the stability domains of the Figures \ref{fig:2} as follows: First we discuss about the stability properties of the region between the red and blue lines. In this region the trivial equilibrium point $E_0$ is a saddle point, the axial equilibrium point $E_1$ is also a saddle point and the coexistence equilibrium point $E^*$ is found to be the stable focus. Second we discuss about the stability properties of the domain which is the other side of red line. In this domain the trivial equilibrium point $E_0$ is a saddle point, the axial equilibrium point $E_1$ is also a saddle point and the coexistence equilibrium point $E^*$ is an unstable focus. Lastly we discuss about the stability properties of the region which is the other side of the blue line. In this region the trivial equilibrium point $E_0$ is a saddle point, the axial equilibrium point $E_1$ is a saddle node and the coexistence equilibrium point $E^*$ becomes infeasible. In Figures (\ref{fig:4}~(a),(b)) we illustrate that the prey and predator populations oscillate (blue for both prey and predator) over unstable coexistence equilibrium point $E^*$, where the portions are coloured in red solid lines for both $x^*$ and $y^*$. Here we observe that the system \eqref{eq:03} experiences Hopf-bifurcation and becomes stable upto the upper threshold value which is coloured in green solid line for both  $x^*$ and $y^*$. The two equilibrium points $E^*$ and $E_1$ intersect each other and we notice that transcritical bifurcation occurs as red dashed lines and green dashed lines, representing unstable and stable domains for both the populations, respectively. Further, increasing the values of both the parameters $c_1$ and $c_2$, we show that $E^*$ becomes an infeasible domain, which is indicated by black solid line and the system converges to prey only equilibrium point $E_1$ (green dashed lines). Now we consider two different parameters $c_1$ and $c_2$. In Figures (\ref{fig:5}~(a),(b)), we show that the model system \eqref{eq:03} is unstable at the coexistence equilibrium point $E^*$ (where $c_1=0.3$ and $c_2=0.2$). Figures (\ref{fig:5}~(c),(d)) illustrate the phenomenon of asymptotic stability at the coexistence equilibrium point $E^*$ for the parameters $c_1=0.6$ and $c_2=0.5$. In Figures (\ref{fig:5}~(e),(f)), we portray that the model system \eqref{eq:03} is stable at prey only equilibrium point $E_1$ for the parameters $c_1=0.85$ and $c_2=0.75$ and this also suggests that the predator population is on the verge of extinction, while the prey population is flourishing. In Figure \ref{fig:8}, we show the Hopf bifurcation phenomenon in the three dimensional space with respect to the bifurcation parameter while portraying the stability analysis of model system \eqref{eq:2} at the coexistence equilibrium point $E^*$.

\subsection{Turing Patterns Associated with One Dimensional Diffusive Model System \eqref{eq:333}}
Now we consider the parameter set as shown in the Table \ref{Table:333} of the model system \eqref{eq:333} and investigate the results by numerical simulation. Figures (\ref{fig:330}~(a) and (b)) display the Turing patterns associated with the prey $(x)$ and predator $(y)$ of the dynamical model system \eqref{eq:333} in the spatial domain. Thus via theorem \eqref{thm:ABC}, we expect Turing instability to occur. This is exactly what is observed in the simulation as shown in Figure \ref{fig:330}. In Figure \ref{fig:111}, we illustrate the Turing and Hopf bifurcations with respect to two parameters of all possible combinations of system \eqref{eq:03}. Parameters such as $d_2, d_1, K, c_2, c_1$ are responsible for the Turing and Hopf bifurcations for the system \eqref{eq:03}. We illustrate the Turing regions with first two rows and in the remaining rows we discuss about Turing and Hopf regions, which have been depicted in Figure \ref{fig:111}. Parameters in the $y$-axis of extreme left of the sub figures of each row are same for the $y$-axis of the other sub figures of the same row. Similarly, parameters in the $x$-axis of bottom sub figures of each column is the same for the $x$-axis of the previous sub figures of the same column. When two parameters are varying in two different axes then the remaining parameters are kept as constant. The parameters are given the values $r=5, \alpha=4.5, \theta=0.06, d=0.4, \Delta h=0.02, k=55, K=0.5, c_2=0.024, c_1=0.01,d_1=0.1, d_2=15.2$. Red and blue lines indicate the Turing and Hopf bifurcations w.r.t two parameters among $d_2, d_1, K, c_2, c_1$. In such combinations, we assume for the sake of easy notation such as 'a' for the Turing region, 'b' for the Hopf-Turing region and 'c' for the Hopf region in Figure \ref{fig:111}. Relations between the real part and imaginary part of the eigenvalue $\lambda_\kappa $ with respect to $\kappa~(\mbox{wave number})$ of the system \eqref{eq:03} on varying the bifurcation parameters $K, c_1$ and $c_2$ respectively have been investigated. Red line represents the real part of the $\lambda_\kappa$ and blue line is for imaginary part of $\lambda_\kappa$. In Figure \ref{fig:10}~(a) the bifurcation parameter \cite{BG15} considered is $K$ and the parametric values of $K$ are taken as follows: (i) $K=0.4$, (ii) $K=0.4852686$ \mbox{and} (iii) $K=0.5$. In Figure \ref{fig:10}~(b) the bifurcation parameter considered is $c_1$ and the parametric values of $c_1$ are taken as follows: (i) $c_1=0.1$, (ii) $c_1=0.47657$ \mbox{and} (iii) $c_1=0.7$. Lastly in the Figure \ref{fig:10}~(c) the bifurcation parameter considered is $c_2$ and the parametric values of $c_2$ are taken as follows: (i) $c_2=0.2687$, (ii) $c_2=0.668699$ \mbox{and} (iii) $c_2=0.8687$. Other parameter sets are given in Table \ref{Table:1}.

\subsection{Spatio-temporal Patterns in two Spatial Dimensions}
Next we consider the evolution of spatio-temporal patterns in which arise in the two spatial dimensions $X$ and $Y$ for \eqref{eq:03} as discussed in \ref{S4} and \ref{S5}. In Figure \ref{fig:222}, the numerical solutions of the predator-prey system \eqref{eq:03} are plotted. The initial data are small spatial perturbations of the stationary solutions $x^*$ and $y^*$ of the spatially homogeneous system. On varying the diffusion parameters $D_1$ and $D_2$ from $0.01$ to $20$ simultaneously, one is led to the four basic one-dimensional dynamics such as chaos, intermittent chaos, smooth oscillatory state and stationary behaviour covering almost all of the domain. One can clearly observe from the first panel of the Figure \ref{fig:222}~(a), a chaotic state of the system \eqref{eq:03} which is portrayed along with the associated patterns corresponding to the diffusion parameters $D_1=0.01=D_2$ in the second and third panels respectively. In the first panel of the Figure \ref{fig:222}~(b), intermittent chaotic stage of the system \eqref{eq:03} is depicted along with the pattern formations of the system corresponding to the diffusion parameters $D_1=0.1=D_2$ in the second and third panel respectively. If we follow the same process we obtain  smooth oscillatory state of the system \eqref{eq:03} in the first panel of the Figure \ref{fig:222}~(c) along with the pattern formations of the system corresponding to the diffusion parameters $D_1=01=D_2$ in the second and third panel respectively.  In Figure \ref{fig:222}~(d), stable steady state solution of the system \eqref{eq:03} is illustrated in the first panel along with the pattern formations of the system associated with the diffusion parameters $D_1=20=D_2$ in the second and third panels respectively. In the first panel of Figure \ref{fig:444}, we show the time-series solutions of the system \eqref{eq:03} at the mesh grid (50,50). We observe the change of the dynamical behaviours of the system \eqref{eq:03} such as chaos, smooth oscillatory state and steady state with changing the bifurcation parameter as $c_2=0.55$, $c_2=0.6$ and $c_2=0.75$, respectively. We also illustrate the pattern formations of the prey $(x)$ and predator $(y)$ (in $100\times 100$ $XY$-domain) corresponding to each stage of the transition from chaotic to stable phase of the system \eqref{eq:03} in the 2nd, and 3rd panels of the Figure \ref{fig:444}. \\

It is a well known fact that the Turing pattern formations near the bifurcation threshold parameter are described by amplitude equations \cite{GP17}. However if the bifurcation parameter is far away from the bifurcation threshold parameter, amplitude equations fail to explain the pattern phenomena. In the first panel of Figures (\ref{fig:11}~(a),(b)), the bifurcation parameter  is $K = 0.47$ and the hexagonal structures arise in the whole area. Significantly we notice the condition that when $K$ is very close to $K_T=0.4852686$, the hexagonal patterns develop at a very slow rate. This phenomenon is called critical slow down. These patterns are called as \emph{spots}. In the second panel of Figures (\ref{fig:11}~(a),(b)), with the decreasing of bifurcation parameter to $K = 0.42$ some stripes arise along with spots and these patterns are termed as \emph{spots-stripes}. Further, on decreasing $K$ to $K = 0.4$, there is certain growth of stripes patterns from spots-stripes formations in the model system \eqref{eq:03}. Indeed the third panel of Figures (\ref{fig:11}~(a),(b)) indicates that spots are decreasing and stripes patterns emerge. Since for the last two cases, the bifurcation parameter $K$ is far away from the bifurcation threshold parameter $K_T$, numerical simulations cannot be derived from the theoretical analysis. From Figures (\ref{fig:11}~(a), (b)), we observe that the model system \eqref{eq:03} shows change of transition of structural formations of random perturbations from spots $\rightarrow$ spots-stripes $\rightarrow$ stripes. The parameter set is given in Table \ref{Table:1}. In Figures (\ref{fig:12}~(a), (b), (c)), we notice that the random perturbations make the spots-stripes formation and the process ends with making spots patterns. In Figures (\ref{fig:13}~(a), (b), (c)), we notice that the random perturbations manage to form spots-stripes structures and end with spots-stripes patterns. The random perturbations of Figures (\ref{fig:14}~(a), (b), (c)), starts with spots-stripes formations, and as the processes go on, spots patterns are decreasing and end with stripes patterns. We consider the the step size of the space for the Figures(\ref{fig:11}, \ref{fig:12}, \ref{fig:13} and \ref{fig:14}) as $ \Delta h=0.4$ and other parameters are given in Table \ref{Table:1} \\

Now we analyse the Turing patterns to interpret the dynamical model system \eqref{eq:03} in the spatial domain. We consider the parameter set as shown in Table \ref{Table:1} of the model system \eqref{eq:03} and investigate the results by numerical simulation with the help of parameters as follows:
\begin{itemize}
	\item [$\bullet$] The system size is $100 \times 100$ square space units. 
	\item [$\bullet$] Time step size, $\Delta t=0.09$.
	\item [$\bullet$] Initial conditions:
	\begin{equation*}
	\begin{split}
	\raggedleft
	x(X,Y,0) &=x^*-\theta_1 \times \psi_1,\\
	y(X,Y,0) &=y^*-\theta_2 \times \psi_2, 
	\end{split}
	\end{equation*}
	where $\theta_1 = 5 \times 10^{-9}$, $\theta_2 = 4 \times 10^{-9}$ and $\psi_i(i=1,2)$ are random numbers which are uniformly distributed in $(-1.5,1.5)$.\\
\end{itemize}

In this section we consider the parameter set as shown in the Table \ref{Table:2} of the model system \eqref{eq:03} and investigate the results by numerical simulation with respect to their boundary conditions with the help of the parameters as follows:
\begin{itemize}
	\item [$\bullet$] The system size is $100 \times 100$ square space units. 
	\item [$\bullet$] Step size of the space, $\Delta h=0.15$.
	\item [$\bullet$] Time step size, $\Delta t=0.05$.
\end{itemize}

From Figures (\ref{fig:15}) and (\ref{fig:19}), the spatial patterns of the model system \eqref{eq:03} are illustrated w.r.t. different boundary conditions, which are listed in Table \ref{Table:2}. In Figure(\ref{fig:15}), we observe that at the step size of the space, $\Delta h=0.15$ with $10$ iterations, the pattern starts forming for the model system \eqref{eq:03}. Here we investigate different pattern formations of the model system \eqref{eq:03} on increasing the iterations. At iterations $100$, we observe that the prey and predator do not change that much in the further iterations. Not much changes are visible in further iterations. We repeat the same process as above for the Figure(\ref{fig:19}) with same number of iterations. We notice that the pattern formation is looking bigger at final iteration as compared to the initial iteration. So, varying time and boundary condition while keeping the other parameter values fixed, we see changes in the dynamics of the system \eqref{eq:03}, making the patterns in the form of spirals, which suggests that the system is transmitting from a chaotic state to a stable state.

\section{Conclusions}\label{S7}

In this paper, we have carried out an extensive study of the modified Rosenzweig-MacArthur predator-prey model by incorporating the effects of habitat complexity on the carrying capacity and fear effect of the functional response of the prey and predator. Our investigations confirm that habitat complexity plays a major role in the reduction of predation rates due to decreased encounter rates between predator and prey as confirmed by several experiments \cite{JB14}. It helps to reduce the predator-prey interaction as well as the available space for the interaction of species \cite{JB14}. Due to the fear effect on prey and habitat complexity, we notice an anti-predator behaviour among the prey which may enhance the survival rate of prey. The stability properties have been elaborately described for the model \eqref{eq:2}. We have analyzed the boundary equilibrium points and investigated the stability properties of the coexistence equilibrium point of the model system \eqref{eq:2}. Further, we have carried out a global stability analysis of the model \eqref{eq:2} and focused our attention to describe Hopf bifurcation, limit cycle and transcritical bifurcation . Through analytical and numerical investigations, we have further shown the cost of habitat complexity $c_1,c_2$ and fear effect $K$ can prevent the occurrence of limit cycle oscillations and increase the stability of the system \eqref{eq:2}. \\

At coexistence equilibrium, our analysis suggests that there is no impact of the level of fear effect $K$ on prey population and on the other hand on increasing  $K$, we notice that there is a declination of predator population. However this does not suggest that the predator population goes extinct due to the level of fear effect \cite{SS18,WZZ16}. Because of habitat complexity, foraging intention and encounter rates between predator-prey interaction will decrease. The available space for interaction of species also gets minimized. Reduced foraging intention of predator due to increasing habitat complexity gradually introduces fearlessness in prey. So the behaviour of prey can change drastically in the presence of predator. We  observe that prey equilibrium density is independent of $c_1$, and so habitat complexity $c_1$ has no impact on the prey density. However it has a certain impact on predator density. We have discussed elaborately based on the numerical simulation of spatio-temporal dynamics on pattern formations of system \eqref{eq:03}. From the numerical investigation it is clear that the changes in the stability domain of the model system are certainly visible after incorporating the diffusive parameters in the model system. We observe in detail about the diffusive pattern formations of predator-prey interaction in spatial domain and investigate the effect of dynamics of the diffusive model system \eqref{eq:03}. We have investigated the modest changes of self diffusion parameters $D_1$ and $D_2$ of the model system \eqref{eq:03} in one dimension which can lead to dramatic changes in the qualitative dynamics of solutions. We observe that a remarkable change in the dynamics of the system from chaotic to regular behaviour occurs, when the diffusivity parameters are increased to a larger value. We have identified Turing pattern formations of the spatially diffusive predator-prey model system  and carried out a stability analysis by amplitude equations with the help of standard multiple-scale analysis and identified three types of Turing patterns i.e. spots, spots-stripes and stripes near the bifurcation threshold parameter where the model system \eqref{eq:03} shows complex dynamical pattern formations. A slight perturbation on the bifurcation parameter makes changes in the dynamics of system \eqref{eq:03} and we described this phenomenon by a weak nonlinear analysis. Numerical simulation suggests that the level of fear effect $K$ has an important role to play in the spatial domain of predator-prey model system and we displayed three different types of pattern formations (spots patterns, spots-stripes patterns and stripe patterns). \\

From the ecological point of view, spots patterns indicate that prey ($x$) belongs to isolated domains with low density and rest of the domains are considered as high density. So this region is safer than other regions for prey ($x$). This significance remains the same for the other species too of the system. Overall we have proposed a predator-prey model by incorporating the effects of habitat complexity on the carrying capacity and fear effect of prey and predator's functional response and have shown that fear effect of prey makes an anti-predator behaviour including habitat complexity which helps prey to survive in ecosystem. The study of spatio temporal dynamics evolves around the pattern formations of diffusive model system in the spatial domain. Pattern formations enrich the dynamics of reaction-diffusion model system on predator-prey populations and provide a certain impact to ecologists.  \\

\section*{Funding: Research of M.L. is supported by DST-SERB by a Distinguished Fellowship and research of Debaldev Jana is supported by SERB, Govt. of India fund (MTR/2019/000788).}

\section*{Conflict of Interest: The authors declare that they have no conflict of interest.}

\newpage

\begin{table}[H]
	\caption{Parameter sets for the pattern structures of the diffusive model \eqref{eq:333}.}\label{Table:333}
	{\scriptsize$$\begin{array}{|c|c|c|c|}
		\hline
		\mbox{Parameter set} & \mbox{Diffusion Parameters} & \mbox{Coexistence Equilibrium Points} & \mbox{Figure}\\
		\hline
		r=5, k=55, \alpha=5 & & &  \\
		h=0.02, K=0.5, \theta=0.06 &  D_1=0.001, D_2=0.002 & (0.33,0.84) & \ref{fig:330} \\
		d=0.08, c_1=0.7, c_2=0.2 &  &  &    \\ 
		\hline
		\end{array}
		$$}
\end{table}

\begin{table}[H]
	\caption{Parameter sets for the pattern structures of the diffusive model \eqref{eq:03}.}\label{Table:1}
	{\scriptsize$$\begin{array}{|c|c|c|c|c|c|}
		\hline
		\mbox{Parameter set} & \mbox{Initial Condition} & \mbox{Diffusion Parameters} & \mbox{Coexistence Equilibrium Points} & \mbox{Figure}\\
		\hline
		r=5, k=55, \alpha=5 & & & & \\
		h=0.02, \theta=0.06 & (0.5,0.4) & D_1=0.1, D_2=15.2 & (0.33,0.84) & \ref{fig:10}~(a) \\
		d=0.08, c_1=0.7, c_2=0.2 &  &  &   & \\ 
		\hline
		r=35, k=55, \alpha=5  & & & & \\
		h=0.02, \theta=0.06 & (0.5,0.4) & D_1=1.1, D_2=0.2 & (0.34,3.50) & \ref{fig:10}~(b) \\
		d=0.08, c_1=0.7, c_2=0.2 &  &  &   & \\ 
		\hline
		r=0.308, k=90, \alpha=5 & & & & \\
		h=0.1, \theta=0.7 & (0.5,0.4) & D_1=0.02, D_2=0.01 & (34.11,5.99) & \ref{fig:10}~(c) \\
		d=0.08, c_1=0.7, c_2=0.2 &  &  &   & \\ 
		\hline
		r=5, k=55, \alpha=5 & & & & \\
		H=0.02, \theta=0.06 & (0.5,0.4) & D_1=0.1, D_2=15.2 & (0.33,0.84) & \ref{fig:11} \\
		d=0.08, c_1=0.7, c_2=0.2 & & & & \\ 
		\hline
		r=5, k=55, \alpha=5 & & & & \\
		K=0.47, H=0.02, \theta=0.06 & (0.5,0.4) & D_1=0.1, D_2=15.2 & (0.33,0.84) & \ref{fig:12} \\
		d=0.08, c_1=0.7, c_2=0.2 & & & & \\ 
		\hline
		r=5, k=55, \alpha=5 & & & & \\
		K=0.42, H=0.02, \theta=0.06 & (0.5,0.4) & D_1=0.1, D_2=15.2 & (0.33,0.84) & \ref{fig:13} \\
		d=0.08, c_1=0.7, c_2=0.2 & & & & \\ 
		\hline
		r=5, k=55, \alpha=5 & & & & \\
		K=0.4, H=0.02, \theta=0.06 & (0.5,0.4) & D_1=0.1, D_2=15.2 & (0.33,0.84) & \ref{fig:14} \\
		d=0.08, c_1=0.7, c_2=0.2 & & & & \\ 
		\hline
		\end{array}
		$$}
\end{table}

\begin{table}[H]
	\caption{Parameter sets for the pattern structures of the diffusive model \eqref{eq:03} with boundary conditions.}\label{Table:2}
	{\scriptsize$$\begin{array}{|c|c|c|c|c|}
		\hline
		\mbox{Parameter set} & \mbox{Boundary condition} & \mbox{Initial Condition} & \mbox{Figure}\\
		\hline
		r=5, k=55, \alpha=5 & & & \\
		h=0.02, K=0.5, \theta=0.06 & 0.5\times (\tan(X-40)^2-\cos(Y+40)^2)<50 & & \\
		d=0.08, c_1=0.7, c_2=0.2 & 0.4\times (\cos((X-47)^2-(Y+47)^2))<50 & (0.5,0.4) & \ref{fig:15} \\
		D_1=0.1, D_2=0.1 &  & & \\
		\hline
		r=5, k=55, \alpha=5 & & & \\
		h=0.02, K=0.5, \theta=0.06 & 0.5\times (\tan((X-25).(Y-25)))>0.5 & & \\
		d=0.08, c_1=0.7, c_2=0.2 & 0.4\times \log(\tan(X-25)^2-\sec(Y-25)^2)<25 & (0.5,0.4) & \ref{fig:19} \\
		D_1=0.1, D_2=0.1 &  & & \\
		\hline
		\end{array}
		$$}
\end{table}

\begin{figure}[ht!]\centering
	{\includegraphics[width=7in, height=5in]{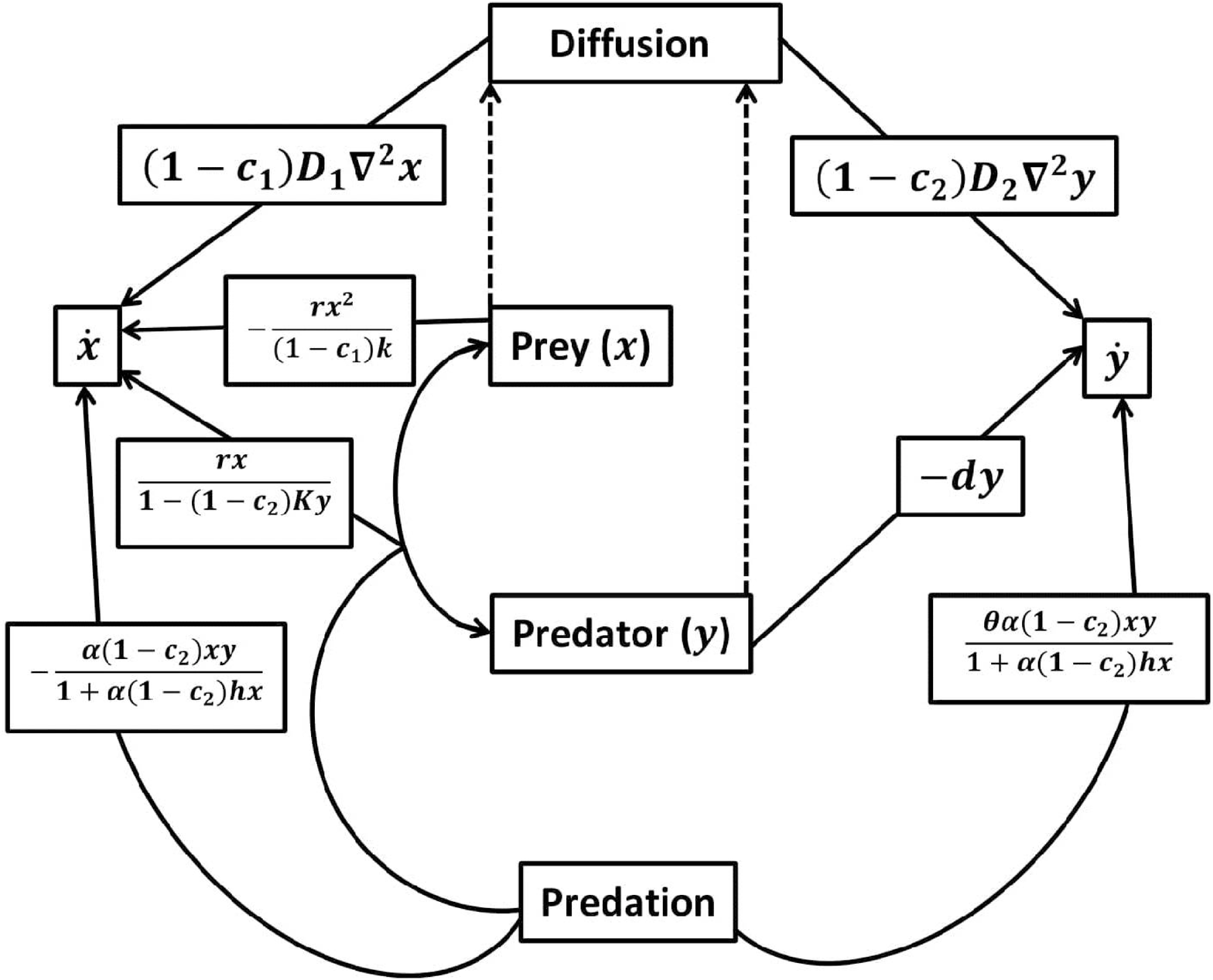}}  
	\caption{Schematic diagram of prey-predator evolution.}\label{fig:100}
\end{figure}

\begin{figure}[ht!]\centering
	{\includegraphics[width=7in, height=5in]{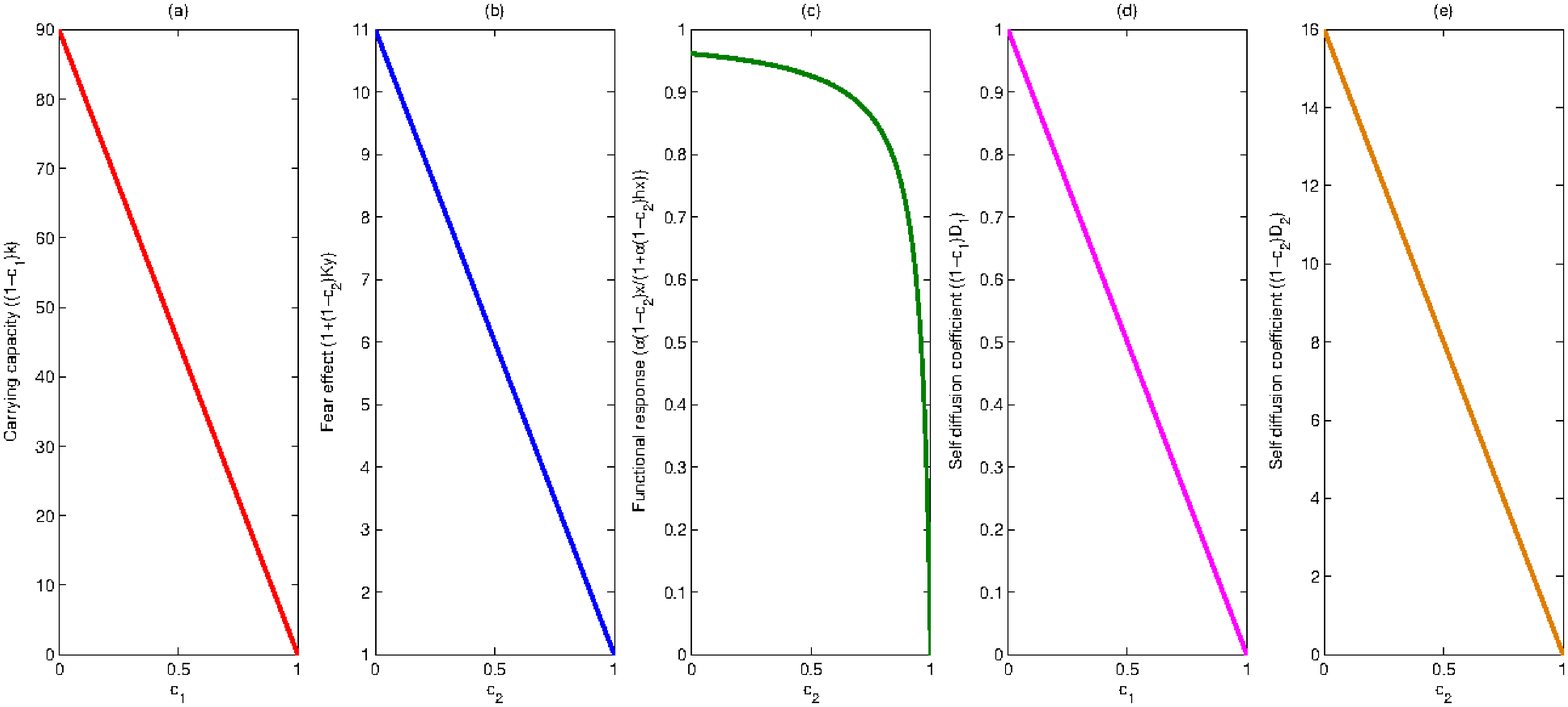}}  
	\caption{Quantitative changes of carrying capacity (figure a) w.r.t $c_1$, fear effect (figure b), functional response (figure c) w.r.t $c_2$ and self diffusion coefficients (figure d and figure e) w.r.t $c_1$ and $c_2$ respectively. Parameters are given in the text.}\label{fig:1}
\end{figure}

\begin{figure}[ht!]\centering
	{\includegraphics[width=7.5in, height=7.5in]{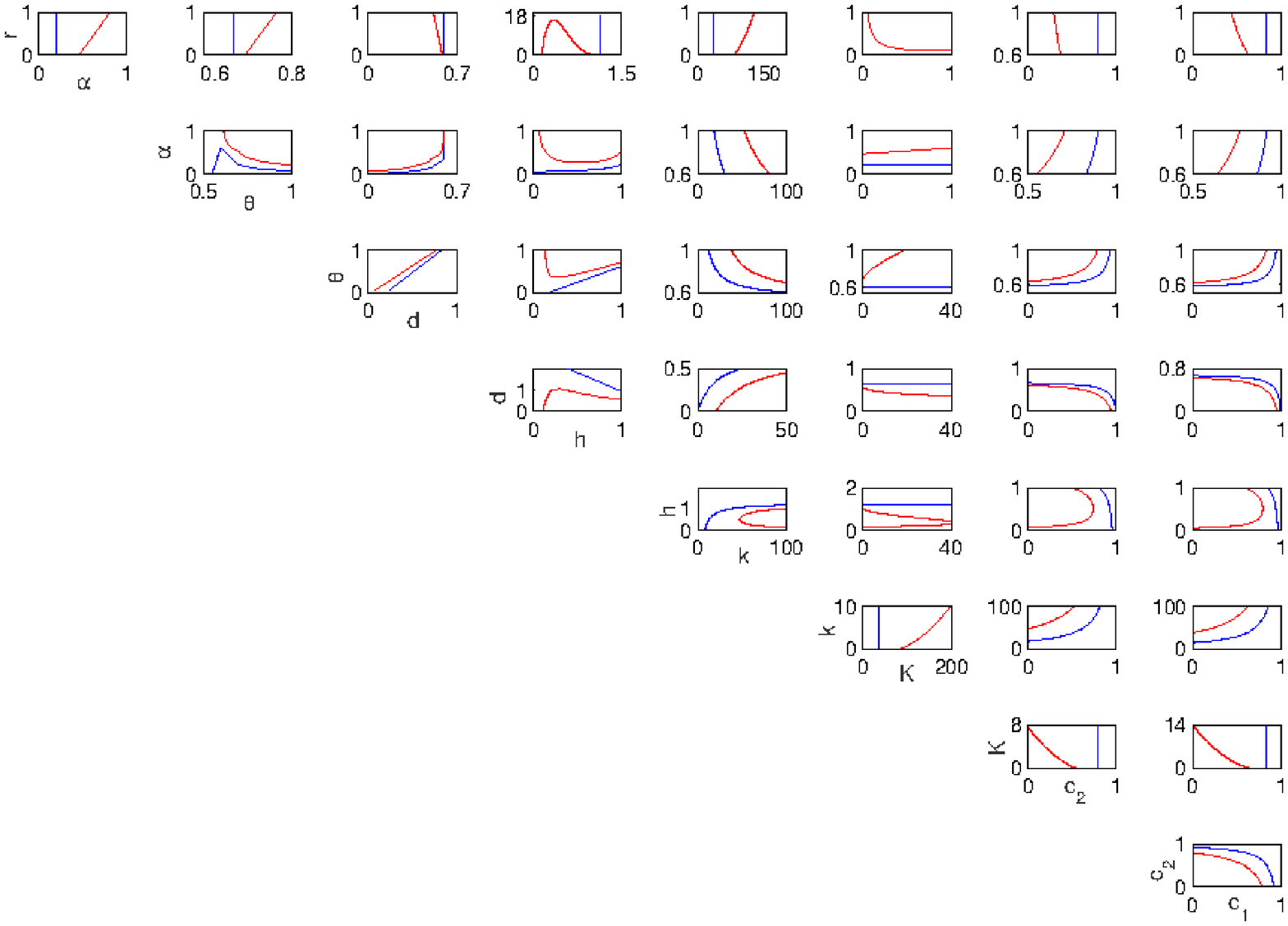}}  
	\caption{Two parametric bifurcation or stability domain diagrams in two parameter phase-planes of system (\ref{eq:2}) w.r.t all combinations of system parameters. Parameters in the $y$-axis of extreme left figures with Red and Blue lines indicate the Hopf and transcritical bifurcations w.r.t respective two parameters. we illustrate the changes between Hopf and Transcritical bifurcations together. The first image of the first row describes the dynamical changes between $(\alpha,r)$, the second image of the first row describes the dynamical changes between $(\theta,r)$, the third one for $(d,r)$, next is for $(h,r)$ …, and the last one for $(c_1,r)$. Similarly, in the second row, we describe the dynamical changes between Hopf and Transcritical bifurcations with respect to the parameters sets such as $(\theta,\alpha)$, $(d,\alpha)$, $(h,\alpha)$ …, and the last one is $(c_1,\alpha)$. This process follows as above and in the eighth row, we depict the changes of Hopf and Transcritical bifurcations between the parameter set $(c_1,c_2)$. Baseline parameters are in the text.}\label{fig:2}
\end{figure}


\begin{figure}[ht!]\centering
	\subfloat[]{\includegraphics[width=6.5in, height=3.5in]{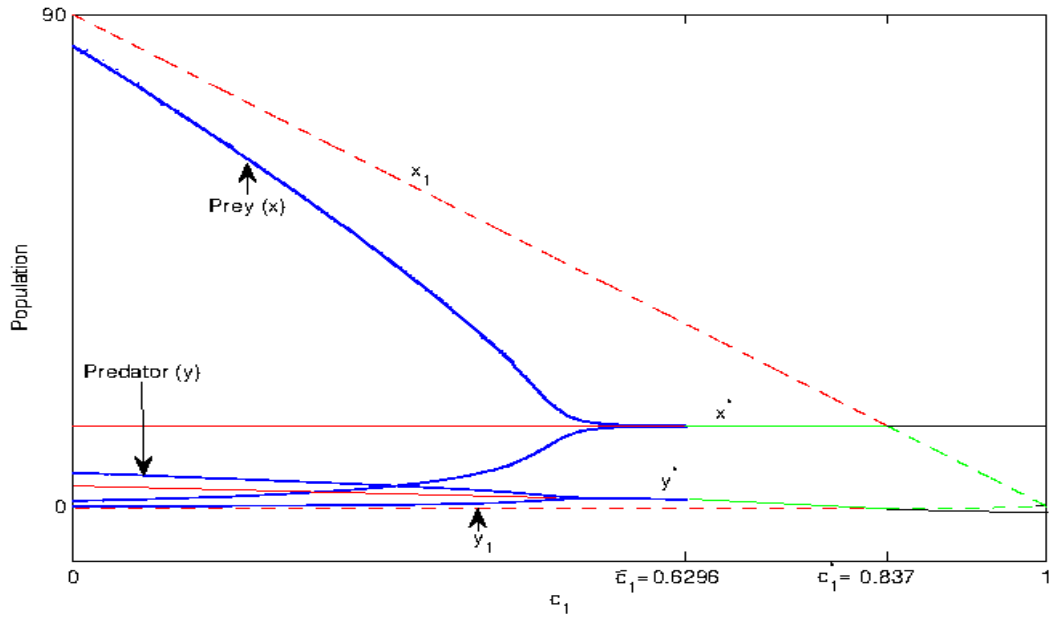}}\\
	\subfloat[]{\includegraphics[width=6.5in, height=3.5in]{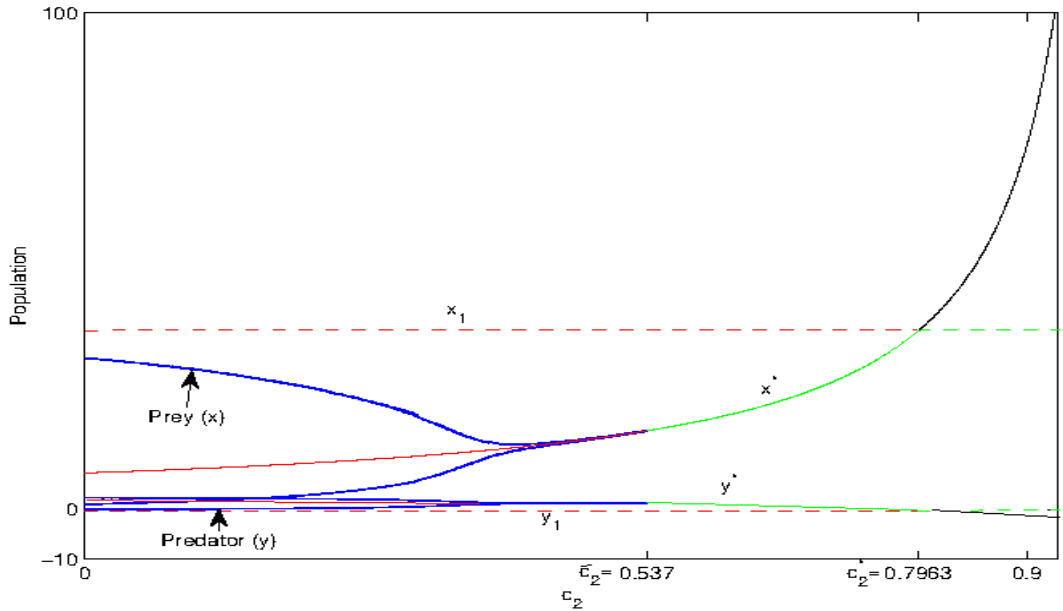}}\\  
	\caption{(a) Bifurcation diagram of system (\ref{eq:2}) with $c_1$ as the bifurcation
		parameter. Solid and dashed lines are respectively for coexistence equilibrium $E^*$ and prey only equilibrium $E_1$. \emph{Black colour} represents infeasibility of $E^*$. \emph{Red} and \emph{green} for unstable and stable status of each equilibrium points. A Hopf bifurcation of $E^*$ occurs at $\bar{c_1}=0.6296$. $E^*$ is unstable focus on $0<\bar{c_1}<0.6296$ and asymptotically stable on $0.6296<{c^*_1}<0.837$. A transcritical bifurcation occurs at ${c^*_1}=0.837$. (b) Bifurcation diagram of system (\ref{eq:2}) with $c_2$ as the bifurcation
		parameter. Solid and dashed lines are respectively for coexistence equilibrium $E^*$ and prey only equilibrium $E_1$. \emph{Black colour} represents infeasibility of $E^*$. \emph{Red} and \emph{green} are for unstable and stable status of each equilibrium points. A Hopf bifurcation of $E^*$ occurs at $\bar{c_2}=0.537$. $E^*$ is unstable focus on $0<\bar{c_2}<0.537$ and asymptotically stable on $0.537<{c^*_2}<0.7963$. A transcritical bifurcation occurs at ${c^*_2}=0.7963$. Other parameters are in the text.}\label{fig:4}
\end{figure}

\begin{figure}[ht!]\centering
	\subfloat[]{\includegraphics[width=2.4in, height=2.5in]{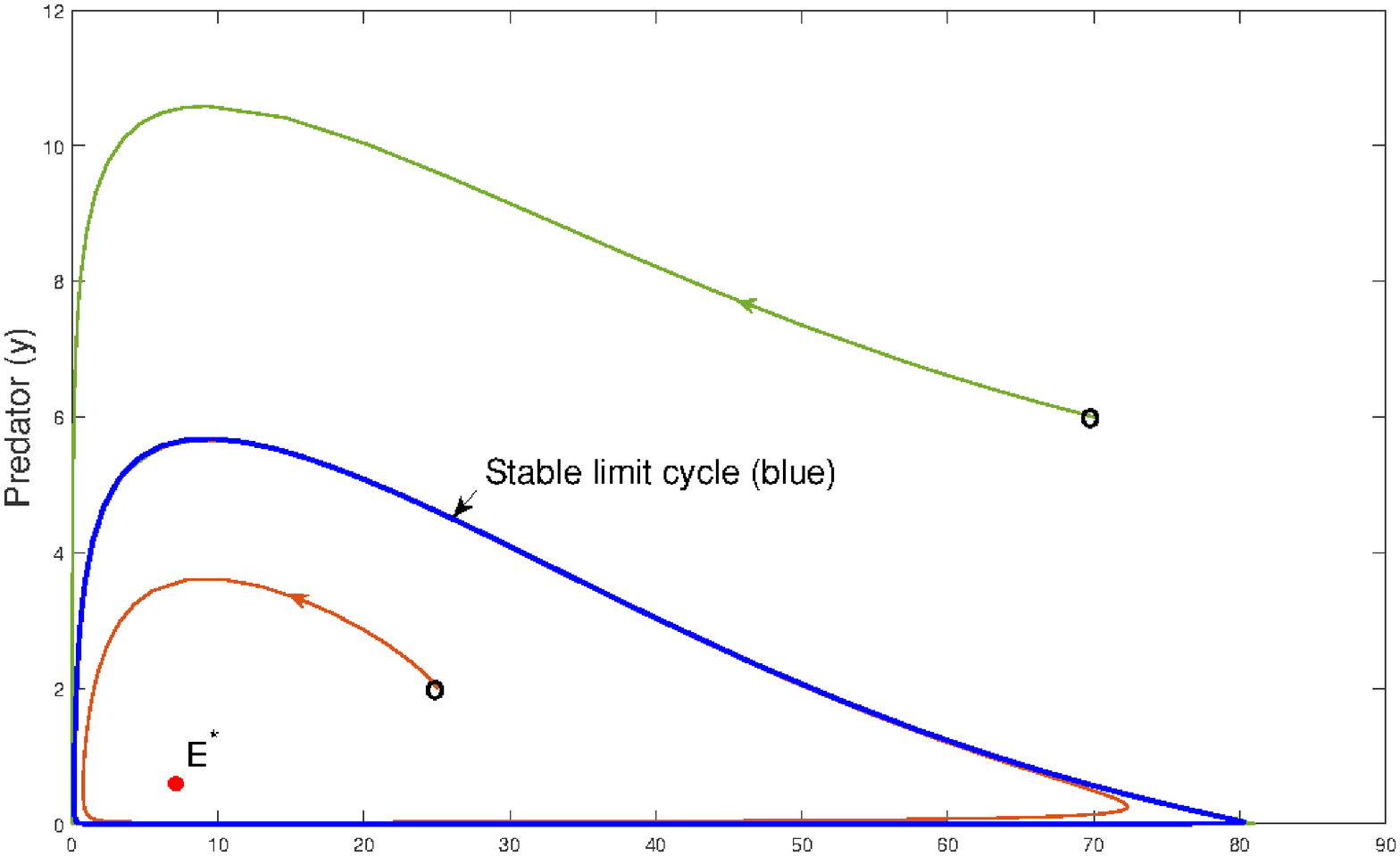}}  
	\subfloat[]{\includegraphics[width=2.4in, height=2.5in]{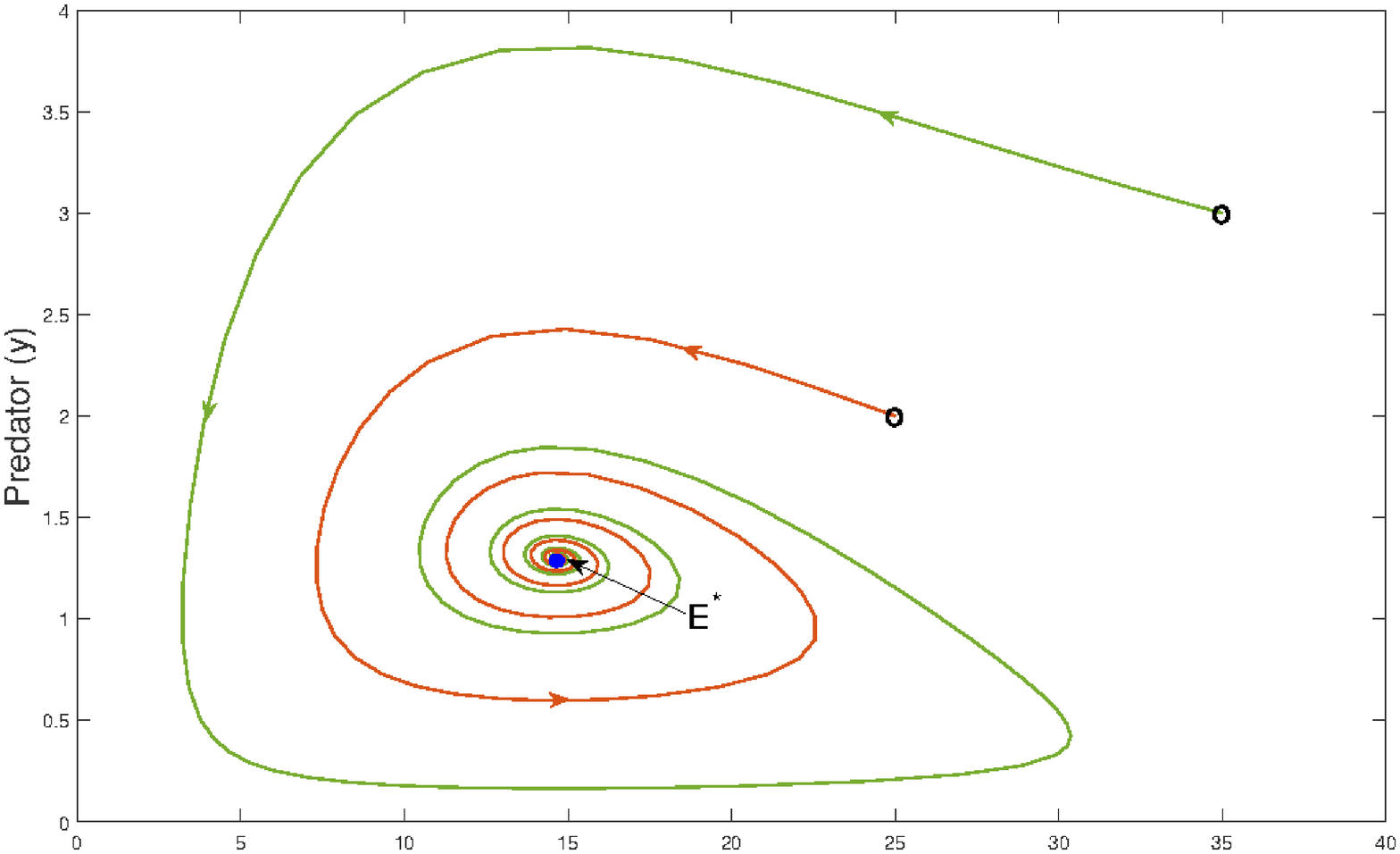}}
	\subfloat[]{\includegraphics[width=2.4in, height=2.5in]{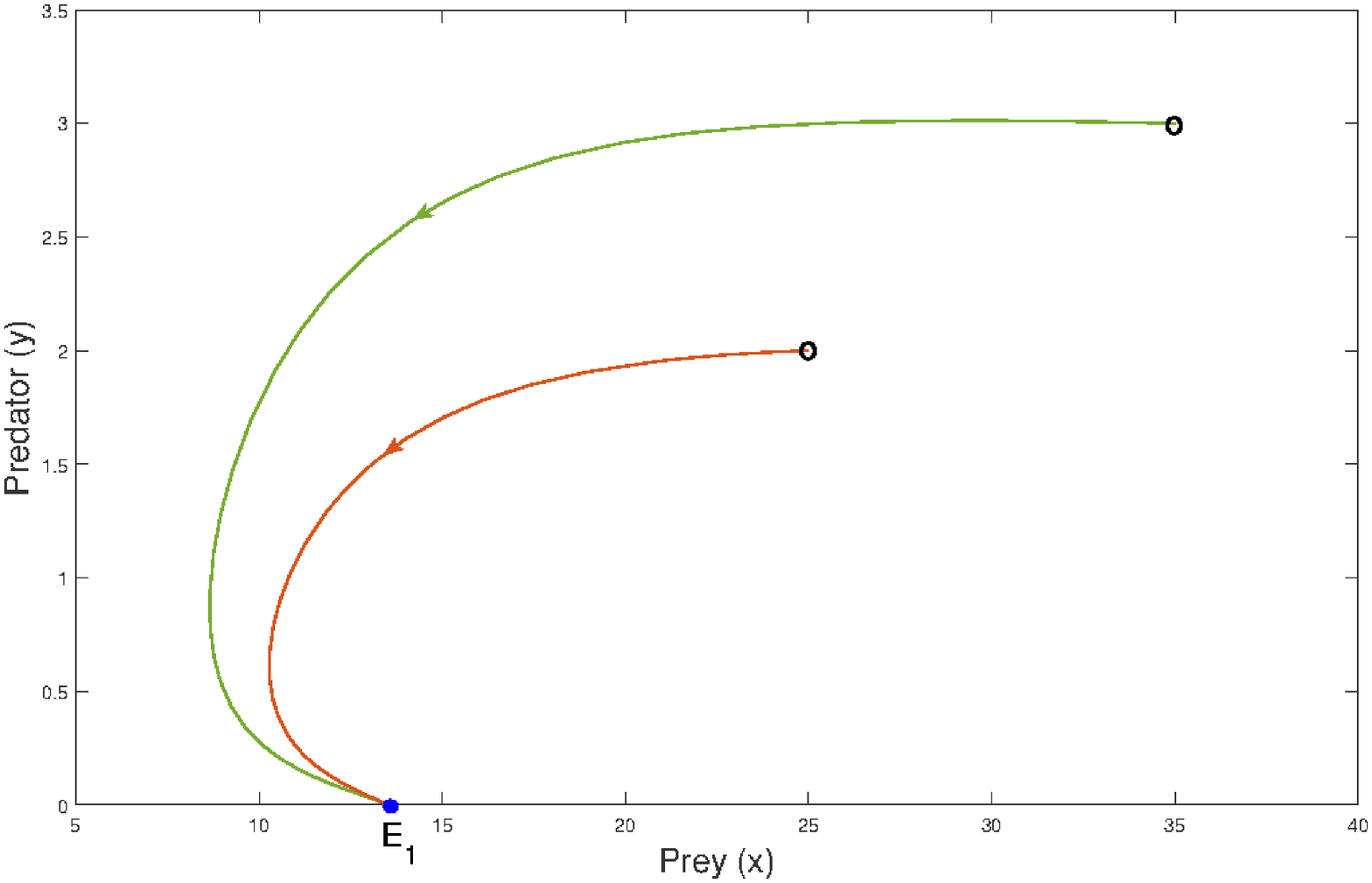}}\\
	\caption{Time series diagram of system (\ref{eq:2}) for different choices of $c_1$ and $c_2$. (a) Coexistence equilibrium point $E^*$ is unstable focus for $c_1=0.3,c_2=0.2$, (b) coexistence equilibrium point $E^*$ is asymptotically stable for $c_1=0.6,c_2=0.5$, (c) prey only equilibrium point $E_1$ is stable for $c_1=0.85,c_2=0.75$. Other parameters are in the text.}\label{fig:5}
\end{figure}



\begin{figure}[ht!]\centering
	{\includegraphics[width=7.0in, height=4.5in]{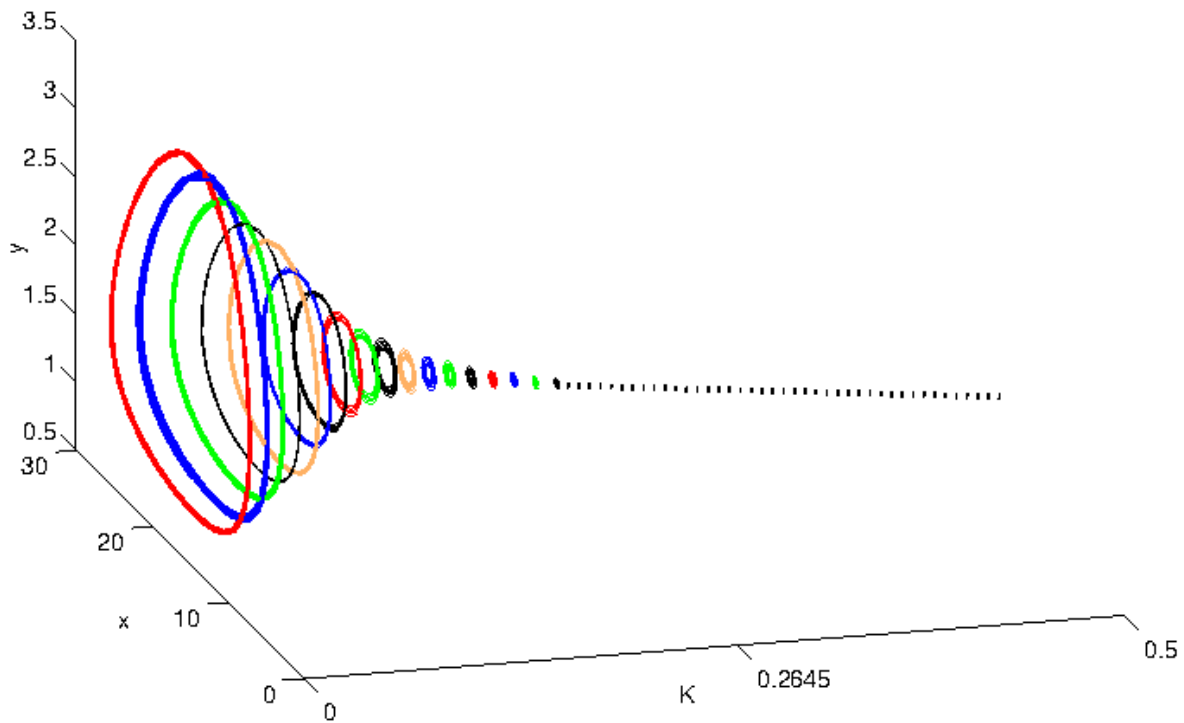}}  \\
	\caption{Hopf-bifurcation diagram of system (\ref{eq:2}) with respect to
		the bifurcation parameter $K$ is drawn in the three-dimensional
		space $(K, x, y)$ when $c_1=0.6,c_2=0.5$. This figure shows that
		the coexistence equilibrium $E^*$ is an unstable focus for $K<0.2645$, where the system converges to a stable limit cycle (depicted by different colour cycles for different values of $K$), stable focus for $K>0.2645$ (depicted by black dotted line) and a Hopf-bifurcation occurs at $K=0.2645$. Other parameters are as in the text.}\label{fig:8}
\end{figure}

\begin{figure}[ht!] 
	\raggedleft
	\centering
	\subfloat[]{\includegraphics[width=7in, height=3in]{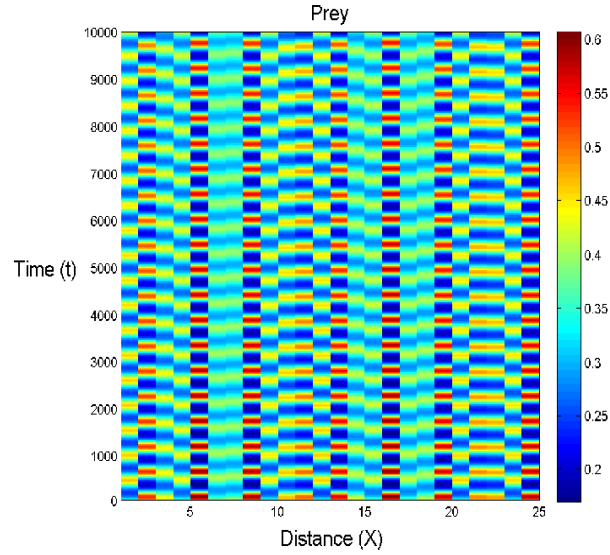}}  \\
	\subfloat[]{\includegraphics[width=7in, height=3in]{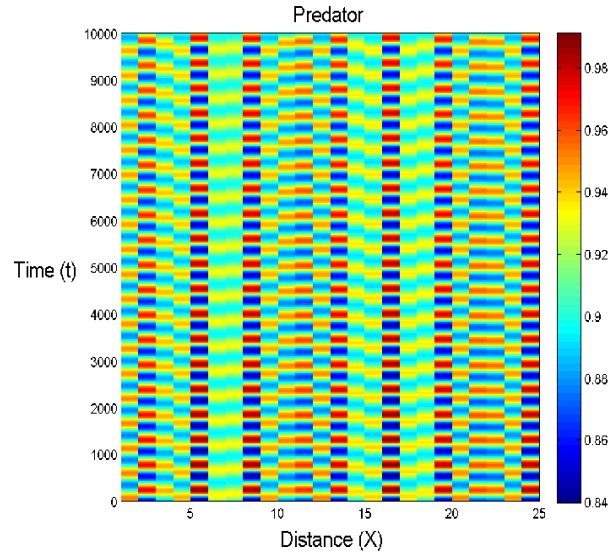}}  \\
	
	\caption{Images of pattern formation in one dimensional diffusion of the time evolution of prey $(x)$ and predator $(y)$ in the $(X-t)$, i.e, spatial $X$-direction and time plane of the diffusive model \eqref{eq:333}. Spatial Turing patterns are observed. The total time has been taken as $t=10,000$. Here we use a perturbation around the initial condition of the spatially homogeneous equilibrium point $(x^*,y^*)=(0.33,0.84)$ by $0.1 \times \cos^2(10X)$. Parameter sets are given in Table \ref{Table:333} .} \label{fig:330}
\end{figure}

\begin{figure}[ht!]\centering
	{\includegraphics[width=7in, height=5.5in]{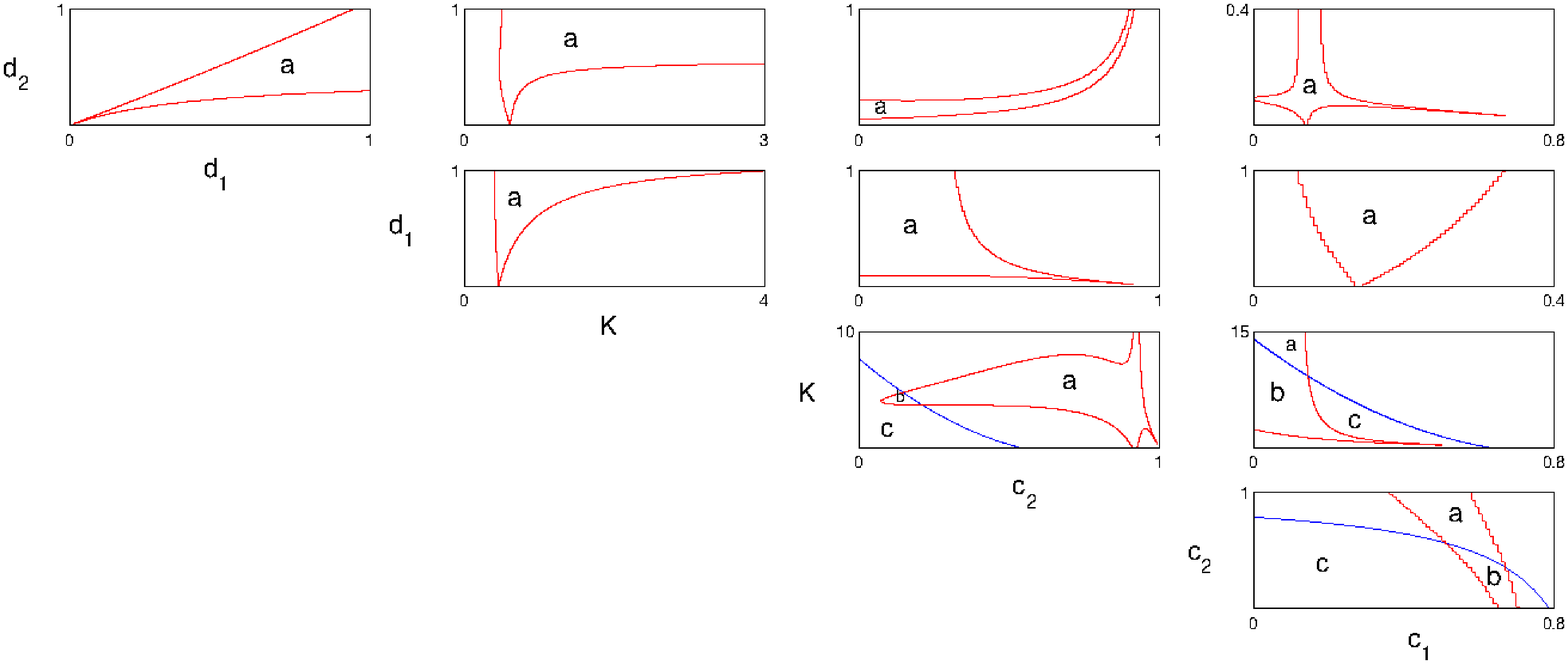}}  
	\caption{Two parametric bifurcation or stability domain diagrams in two parameter phase-planes of system (\ref{eq:2}) w.r.t all combinations of system parameters. Red and blue lines indicate the Turing and Hopf bifurcations w.r.t respective two parameters. Here 'a' denotes the Turing region, 'b' is the Hopf-Turing region and 'c' be the Hopf region. Here, we specifically investigate the changes between Hopf and Turing bifurcations with respect to the parameters sets such as $(d_1,d_2)$, $(K,d_2)$, $(c_2,d_2)$ and $(c_1,d_2)$ in the first row of the Figure. In the second row, we make changes between $(K,d_1)$, $(c_2,d_1)$ and $(c_1,d_1)$. In the third row, we choose the parameters sets such as $(c_2,K)$ and $(c_1,K)$. In the fourth row, we see changes between $(c_1,c_2)$. Parameters set are given in the text.}\label{fig:111}
\end{figure}

\begin{figure}[ht!]
	\centering
	\subfloat[]{\includegraphics[width=6.9in, height=2.55in]{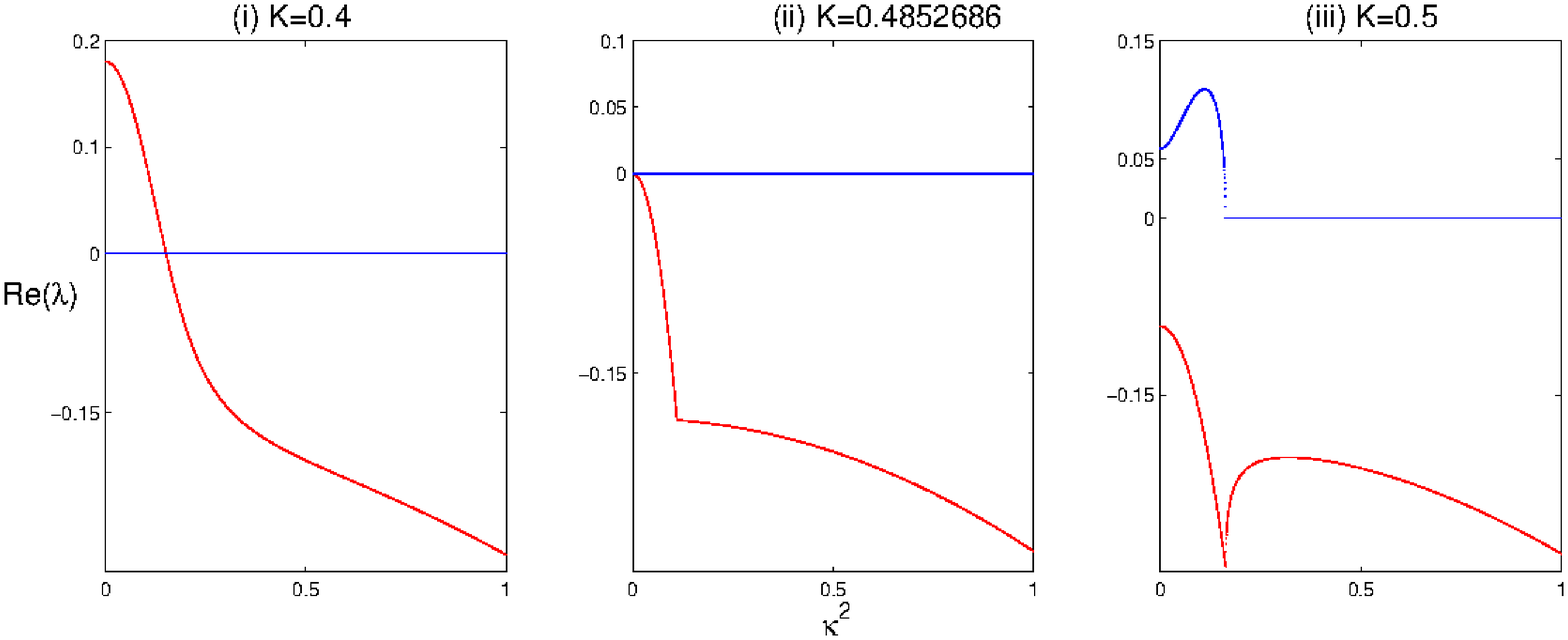}}  \\
	\subfloat[]{\includegraphics[width=6.9in, height=2.55in]{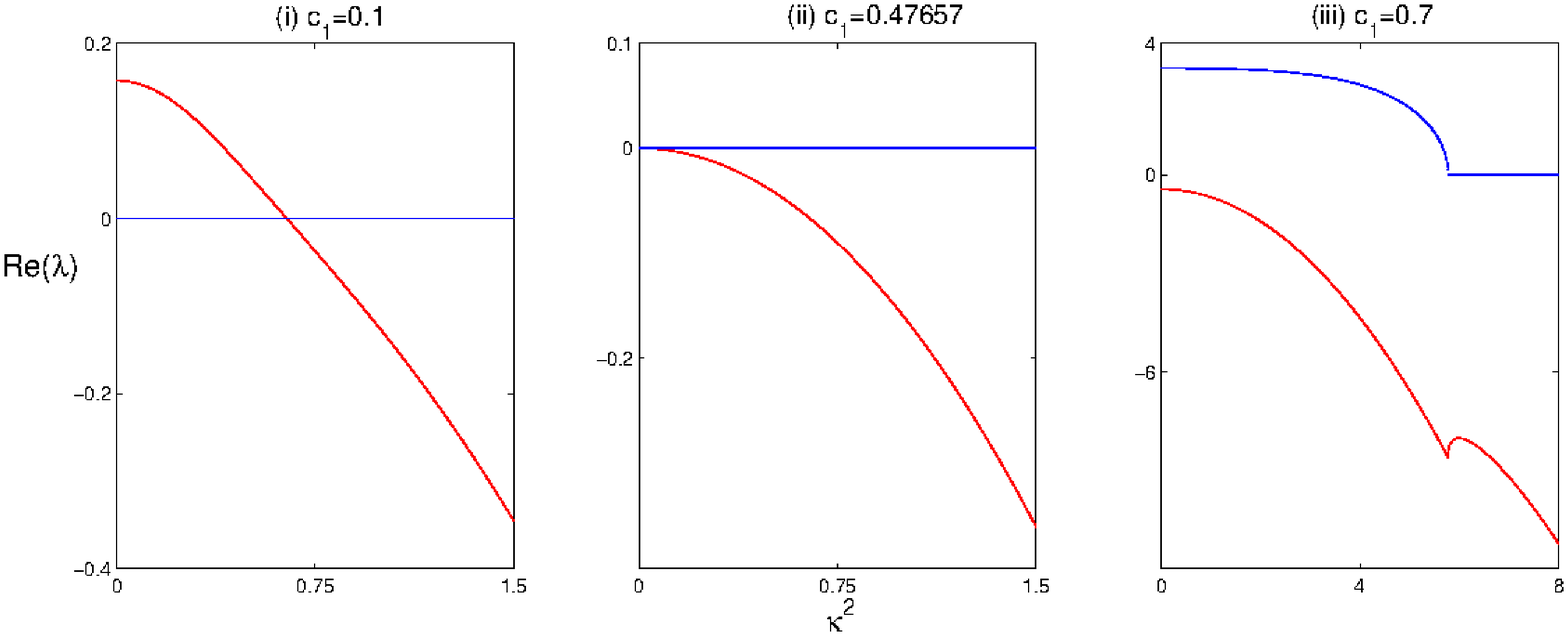}}  \\
	\subfloat[]{\includegraphics[width=6.9in, height=2.55in]{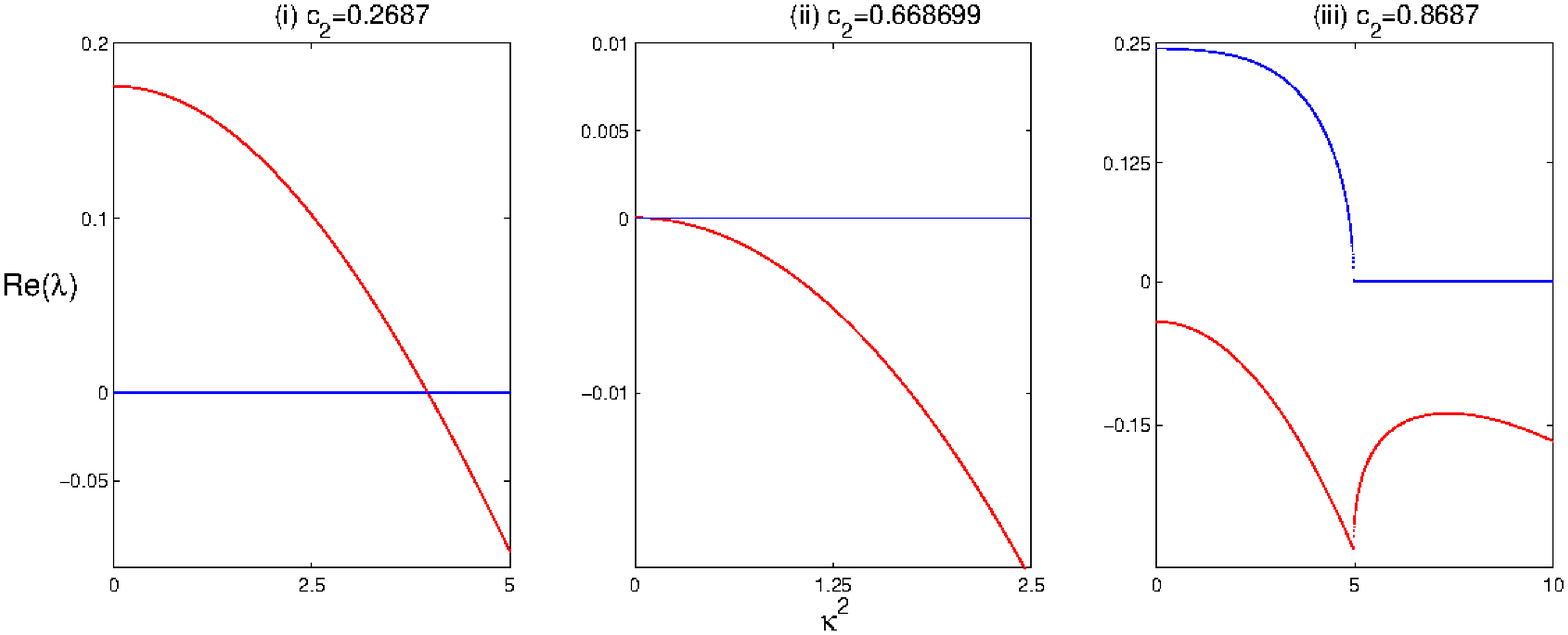}} \\
	\caption{Relation between $Re(\lambda)~\mbox{(real part of the eigen value $ \lambda $)}$ and $\kappa^2~(\mbox{wave number})$ on varying the bifurcation parameters $K$, $c_1$ \mbox{and} $c_2$ respectively. \emph{Red line} represents the real part of the $\lambda$ and \emph{Blue line} is for the imaginary part of the $\lambda$. (a) In the first row (i) $K=0.4$, (ii) $K=0.4852686$ \mbox{and} (iii) $K=0.5$. (b) In the second row (i) $c_1=0.1$, (ii) $c_1=0.47657$ \mbox{and} (iii) $c_1=0.7$. (c) In the third row (i) $c_2=0.2687$, (ii) $c_2=0.668699$ \mbox{and} (iii) $c_2=0.8687$. Other parameter sets are given in Table \ref{Table:1} .  \\}
	\label{fig:10}
\end{figure}

\begin{figure}[ht!]\centering
	\subfloat[]{\includegraphics[width=6.6in, height=1.7in]{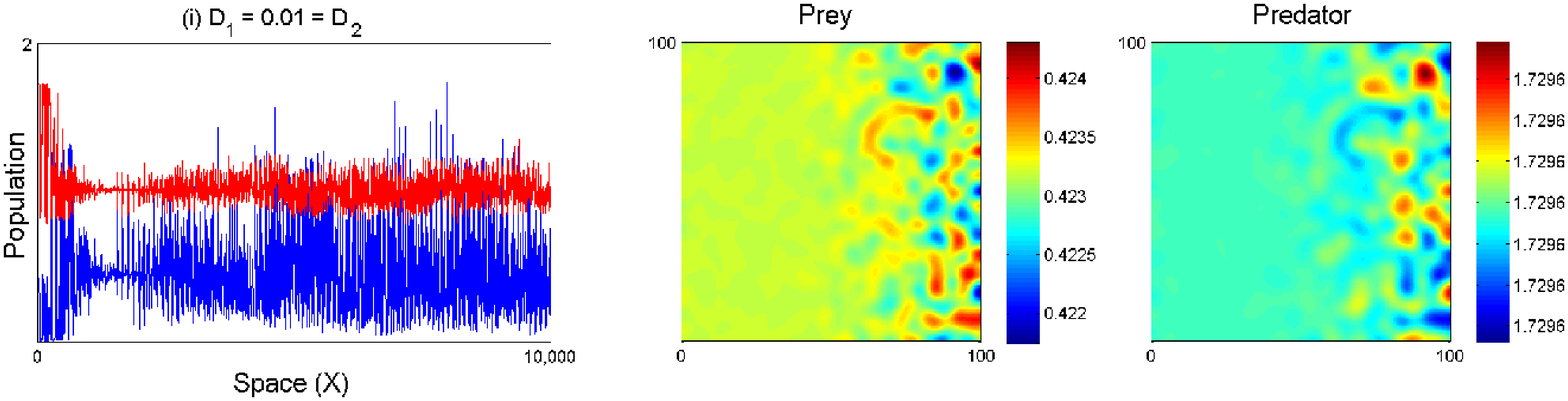}}  \\
	\subfloat[]{\includegraphics[width=6.6in, height=1.7in]{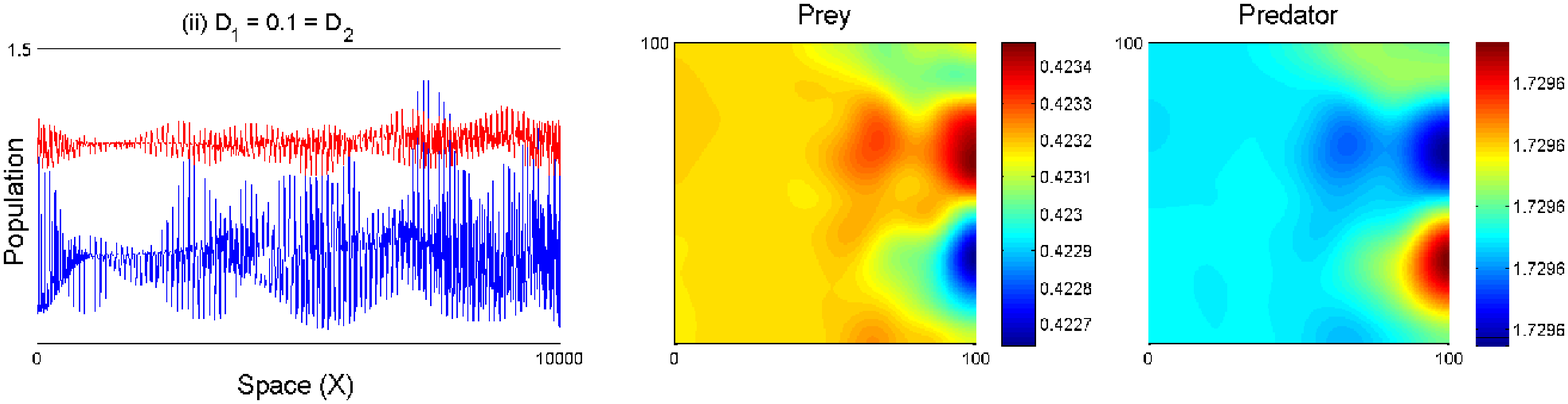}}  \\
	\subfloat[]{\includegraphics[width=6.6in, height=1.7in]{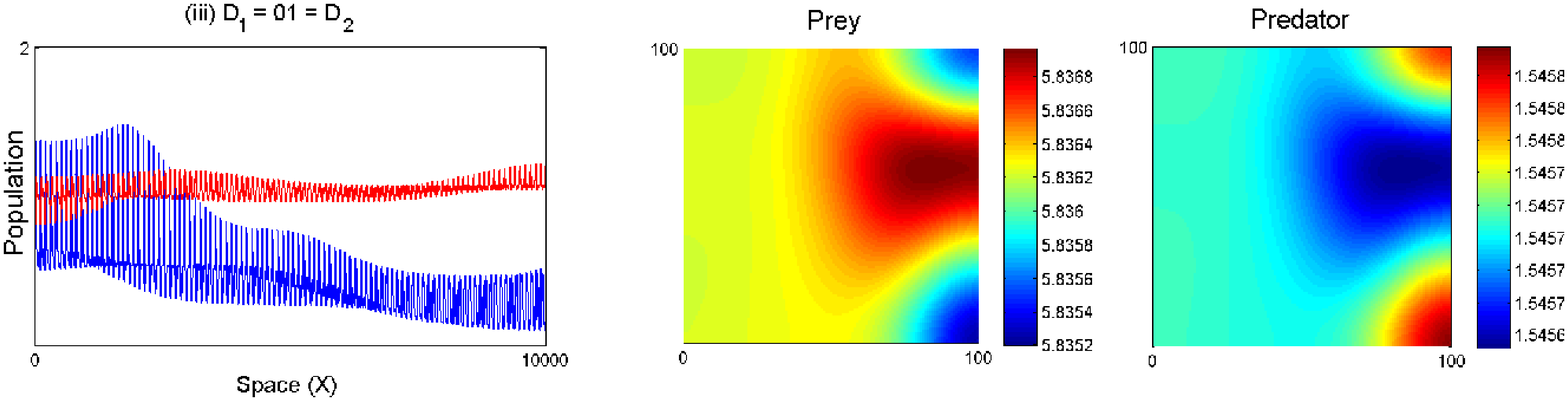}}  \\
	\subfloat[]{\includegraphics[width=6.6in, height=1.7in]{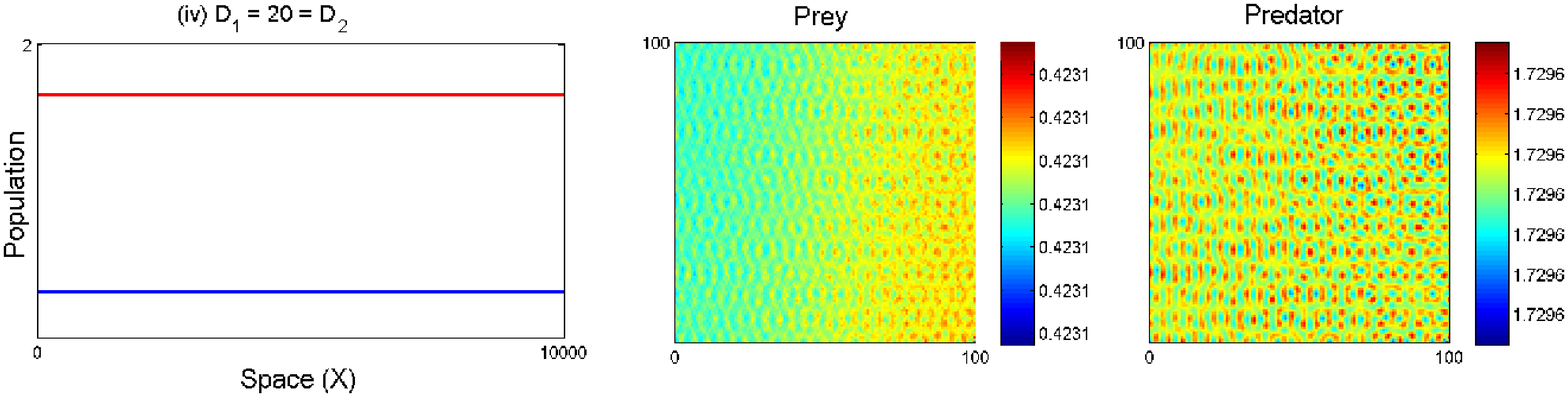}} \\
	\caption{In the 1st panel we portray space series generated at time $T=500$ and for different values of $D_1$ and $D_2$ of the system (\ref{eq:03}) i.e. (a) $D_1=0.01=D_2$, (b) $D_1=0.1=D_2$, (c) $D_1=01=D_2$ and (d) $D_1=20=D_2$ showing the effect of diffusivity constants $D_1$ and $D_2$ on the dynamics of the model of the system (\ref{eq:03}). \emph{Blue line} represents prey ($x$) population and \emph{Red line} is for predator ($y$) population. Parameter values and initial data: $r=4, \alpha=5, \theta=0.06, d=0.08, h=0.02, k=70, K=0.5, c_1=0.7, c_2=0.4$, space-step size ($\Delta h$)=$1$, time-step size ($\Delta t$)=$0.1$, $x^*=0.33$, $y^*=0.84$. In the 2nd and 3rd panels we depict the pattern formations of the prey and predator in the spatial $(X-Y)$ plane.}\label{fig:222}
\end{figure}

\begin{figure}[ht!]\centering
	\subfloat[]{\includegraphics[width=6.6in, height=2.4in]{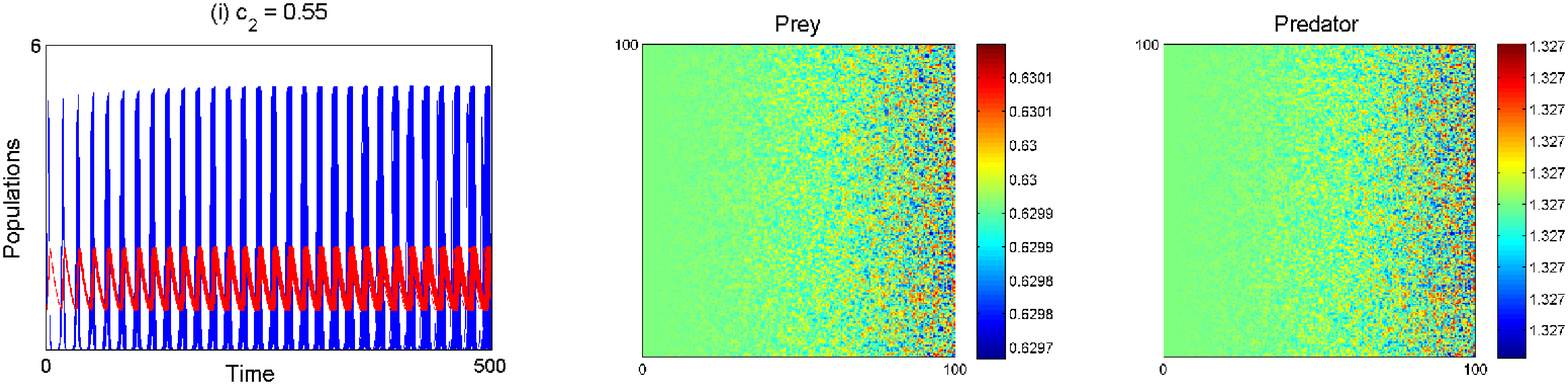}}  \\
	\subfloat[]{\includegraphics[width=6.6in, height=2.4in]{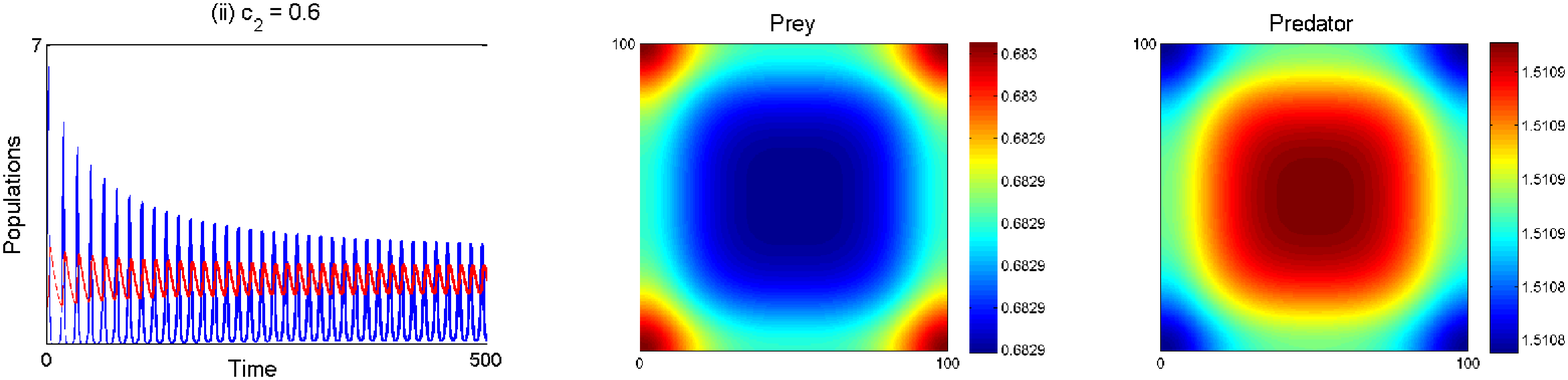}}  \\
	\subfloat[]{\includegraphics[width=6.6in, height=2.4in]{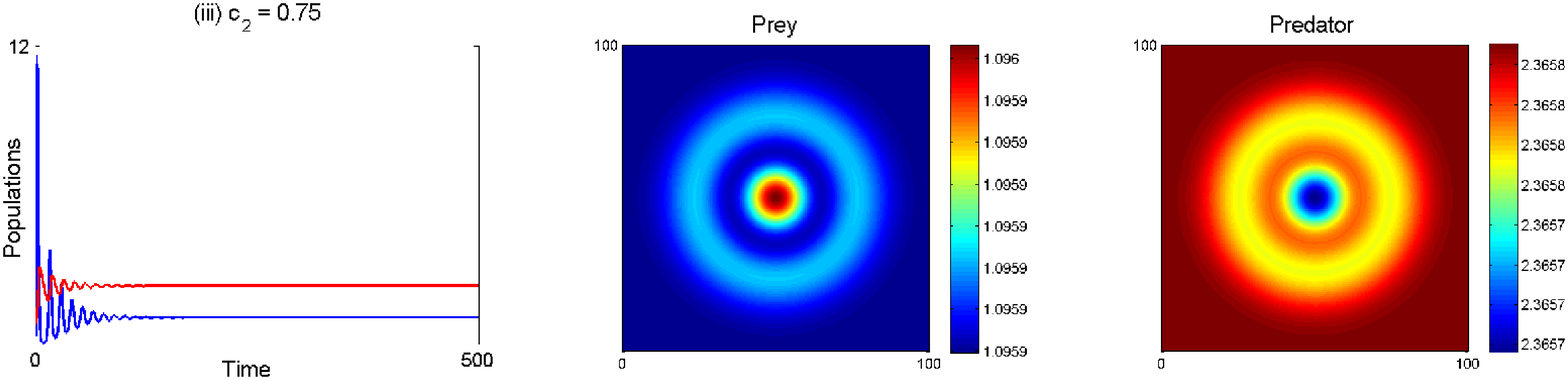}}  \\
	\caption{In the 1st panel we illustrate the time-series solutions of the system \eqref{eq:03} generated at the mesh grid (50,50) for $0<t<T=500$ with diffusion parameters $D_1=D_2=0.0005$ and for different values of bifurcation parameter $\beta$ i.e. (a) $c_2=0.55$, (b) $c_2=0.6$ and (c) $c_2=0.75$ showing the effect of the dynamics of the model of the system (\ref{eq:03}). \emph{Blue line} represents prey population and \emph{Red line} is for predator population. Parameter values and initial data: $r=4, \alpha=5, \theta=0.06, d=0.08, h=0.02, k=70, K=0.5, c_1=0.7$, space-step $(\Delta h)$=$0.3$, time-step $(\Delta t)$=$0.2$, $x^*=0.33$, $y^*=0.84$. In the 2nd and 3rd panels we depict the pattern formations of the prey and predator in the spatial $(X-Y)$-plane.}\label{fig:444}
\end{figure}

\begin{figure}[ht!]
	\centering
	\subfloat[]{\includegraphics[width=6.8in, height=3in]{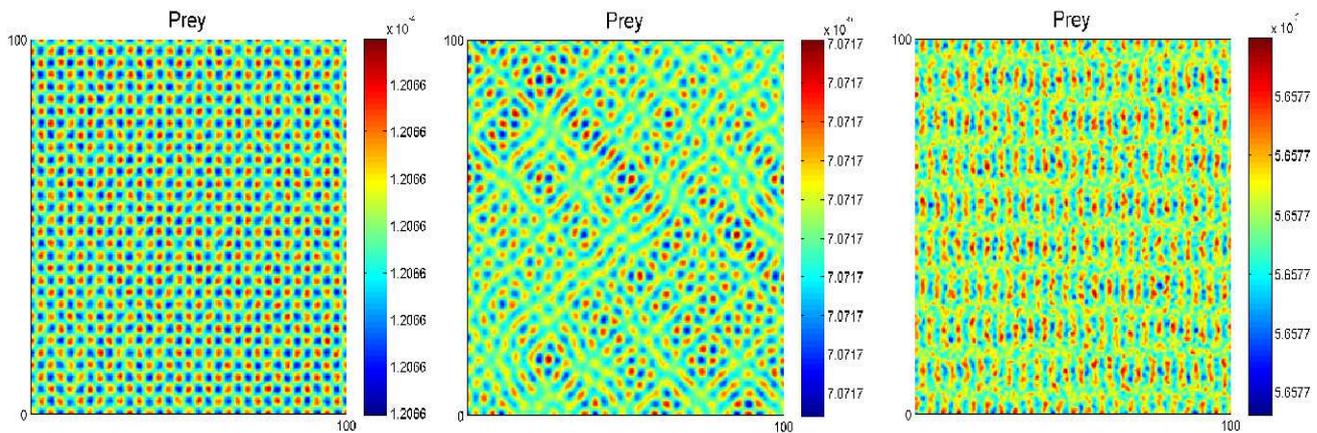}} \\
	\subfloat[]{\includegraphics[width=6.8in, height=3in]{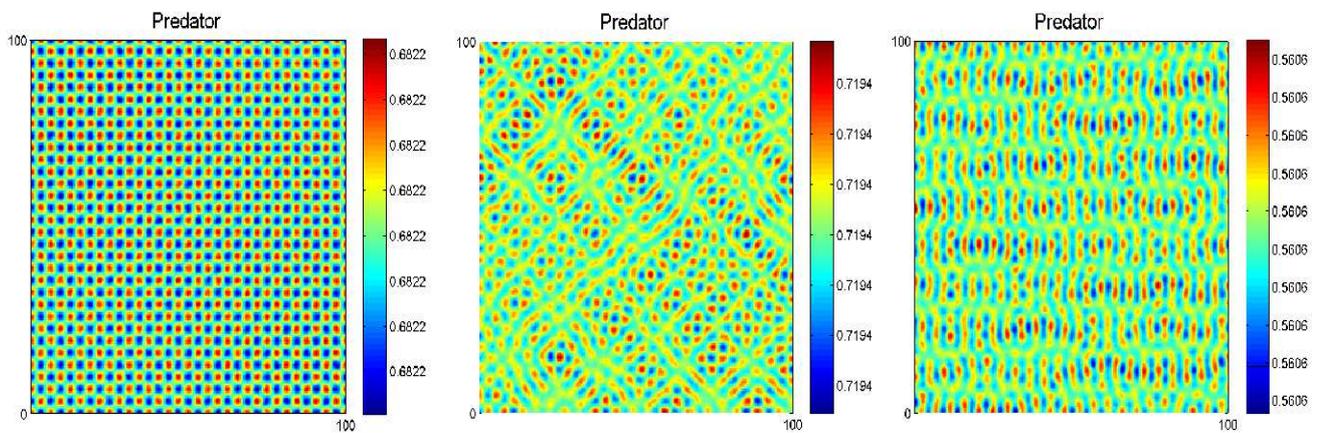}} \\
	\caption{Snapshots of Turing patterns of the time evolution of  prey $(x)$ and predator $(y)$ in the \textit{XY}-plane of the diffusive model \eqref{eq:03}, where the step size of the space, $\Delta h=0.4$. In first panel, time has been taken as $T=8000$ and time has been taken as $T=5000$ in both second and third panels respectively. Parameter sets are given in Table \ref{Table:1} .}
	\label{fig:11}
\end{figure}

\begin{figure}[ht!]
	\centering
	\subfloat[]{\includegraphics[width=7.0in, height=2.6in]{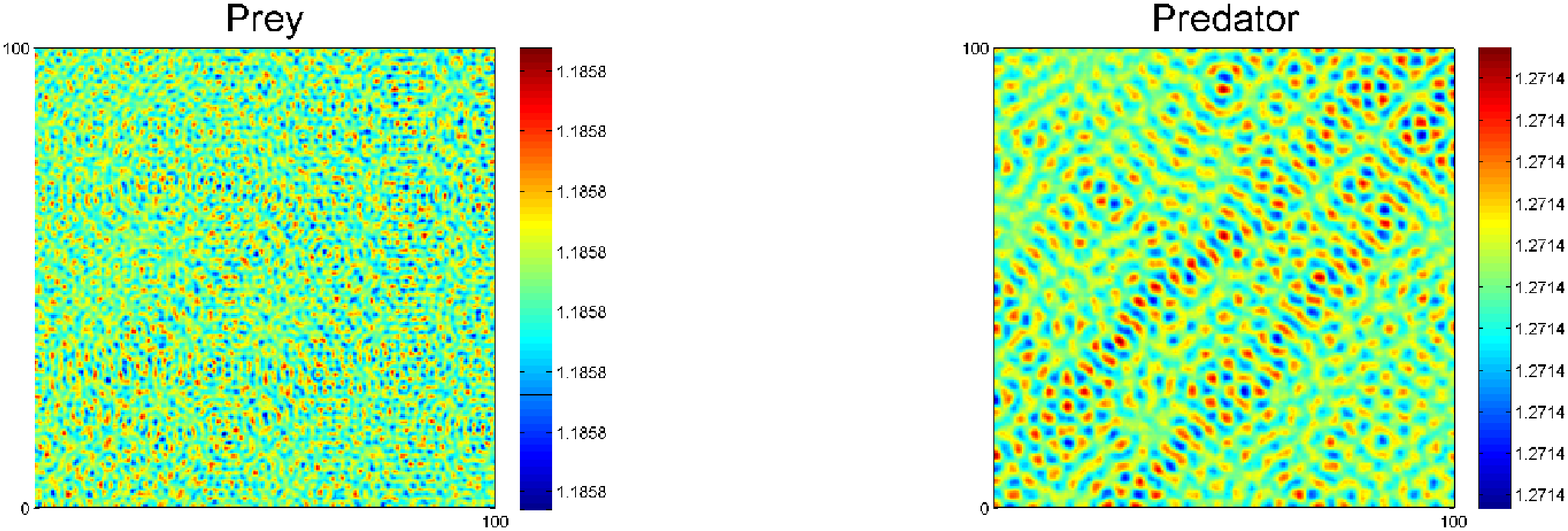}} \\
	\subfloat[]{\includegraphics[width=7.0in, height=2.6in]{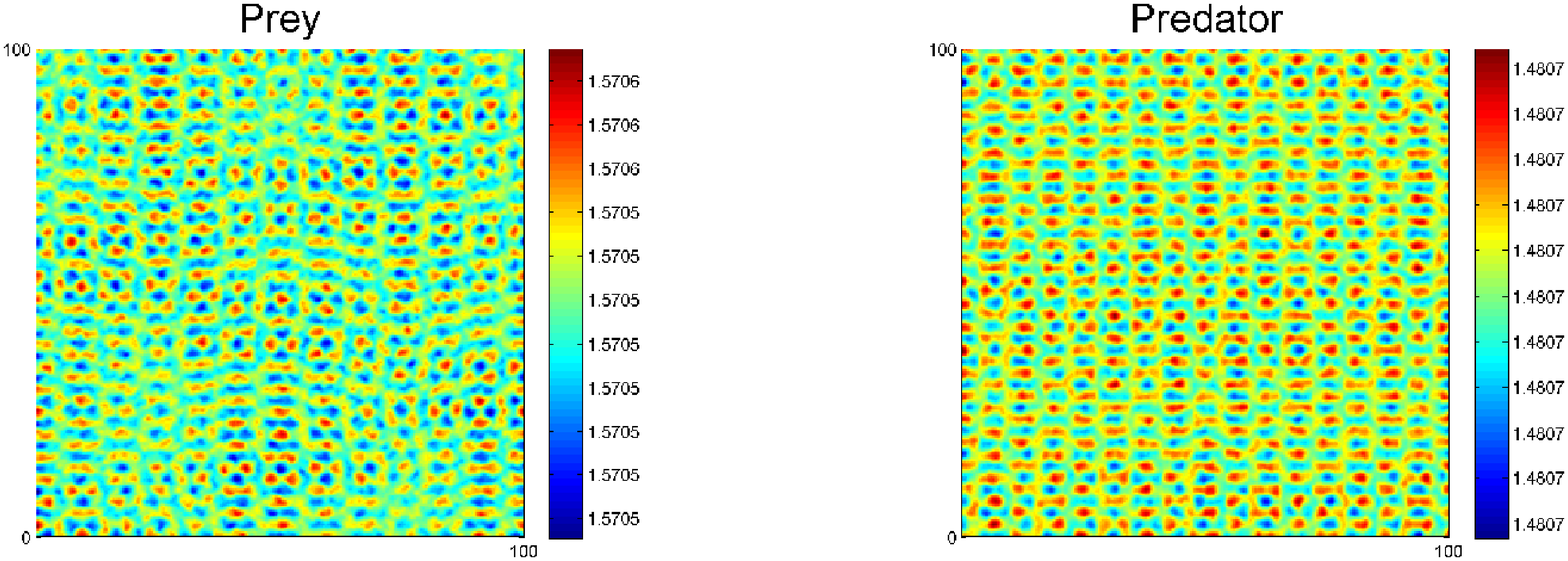}} \\
	\subfloat[]{\includegraphics[width=7.0in, height=2.6in]{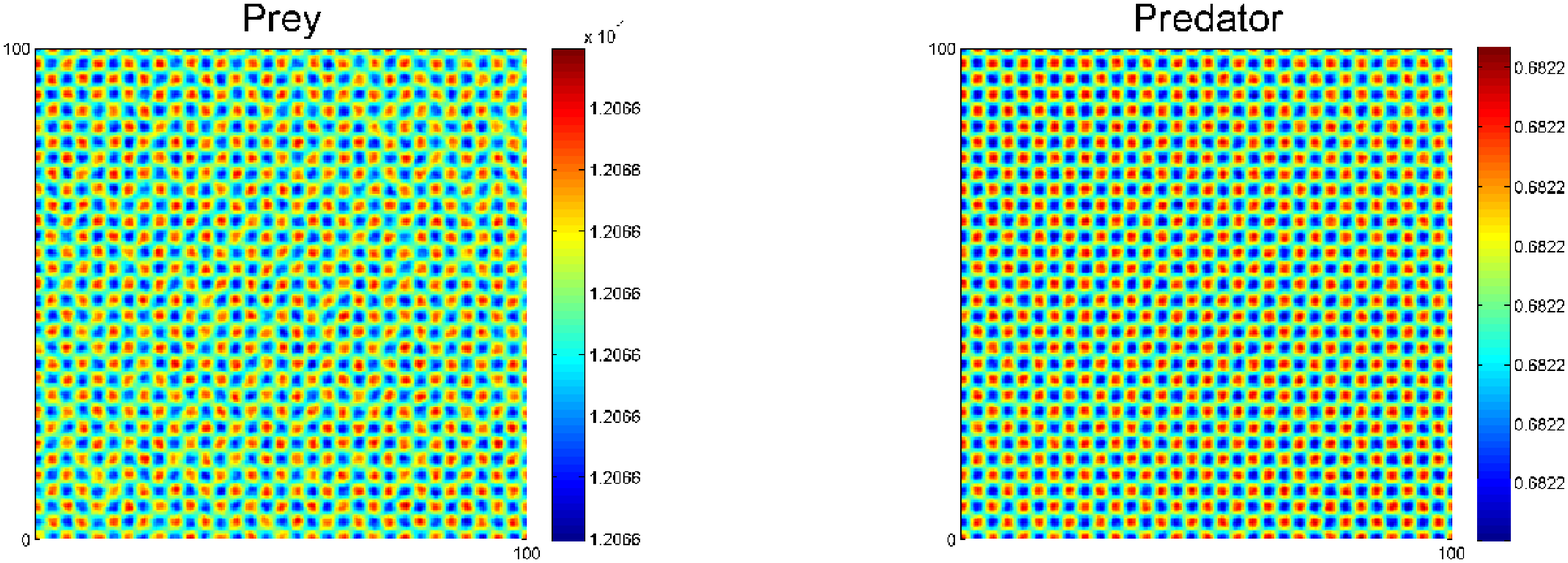}} \\
	\caption{Snapshots of contour pictures of the time evolution of prey $(x)$ and predator $(y)$ in the \textit{XY}-plane of the diffusive model \eqref{eq:03}, where the step size of the space, $\Delta h=0.4$ and the time has been taken as (a) $T=100$, (b) $T=4000$ and (c) $T=8000.$ Parameter sets are given in Table \ref{Table:1} .}
	\label{fig:12}
\end{figure}

\begin{figure}[ht!]
	\centering
	\subfloat[]{\includegraphics[width=7.0in, height=2.6in]{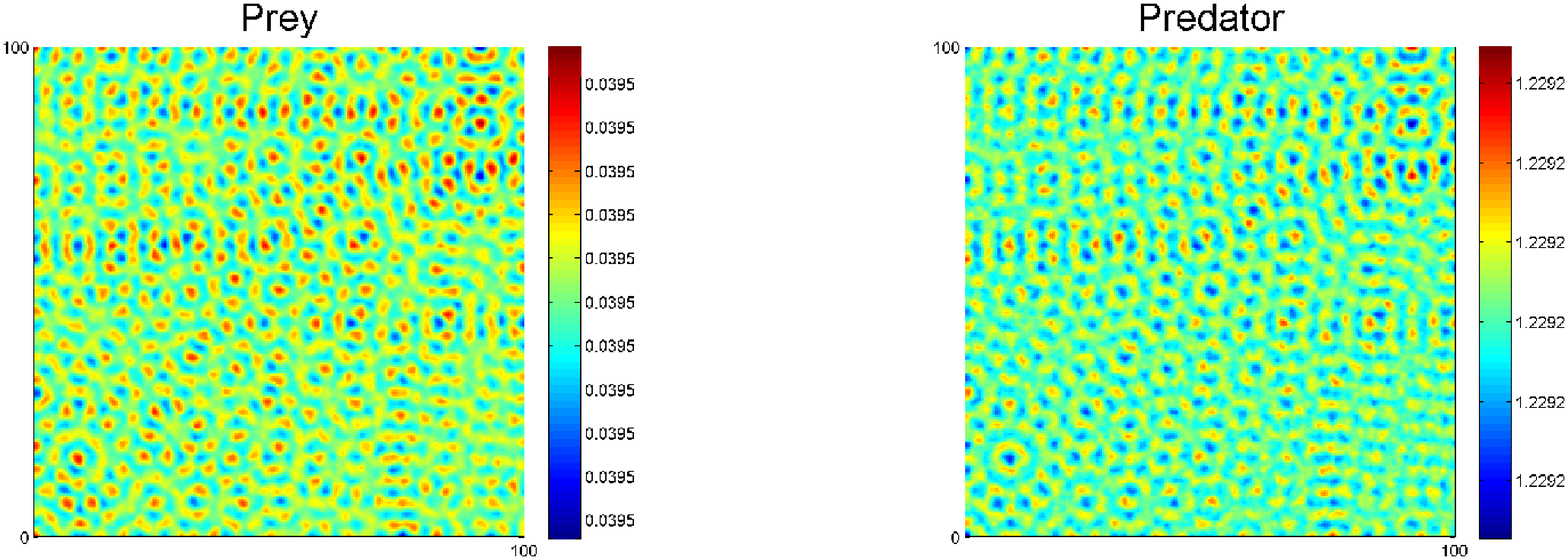}} \\
	\subfloat[]{\includegraphics[width=7.0in, height=2.6in]{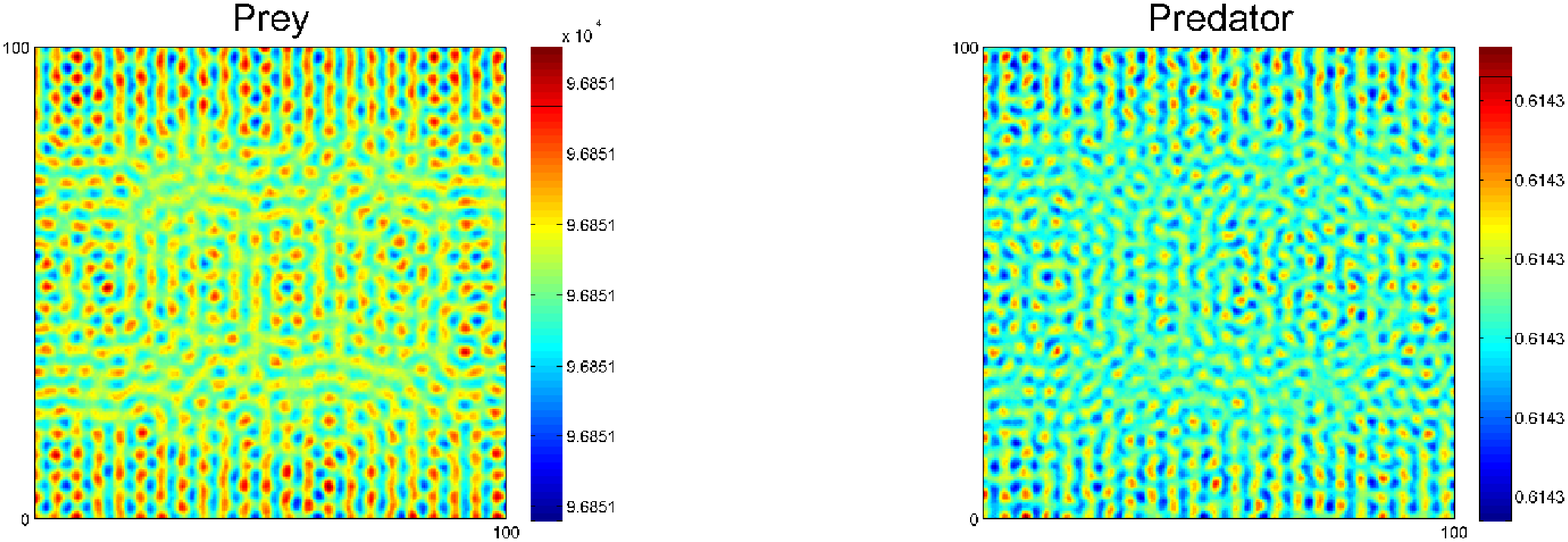}} \\
	\subfloat[]{\includegraphics[width=7.0in, height=2.6in]{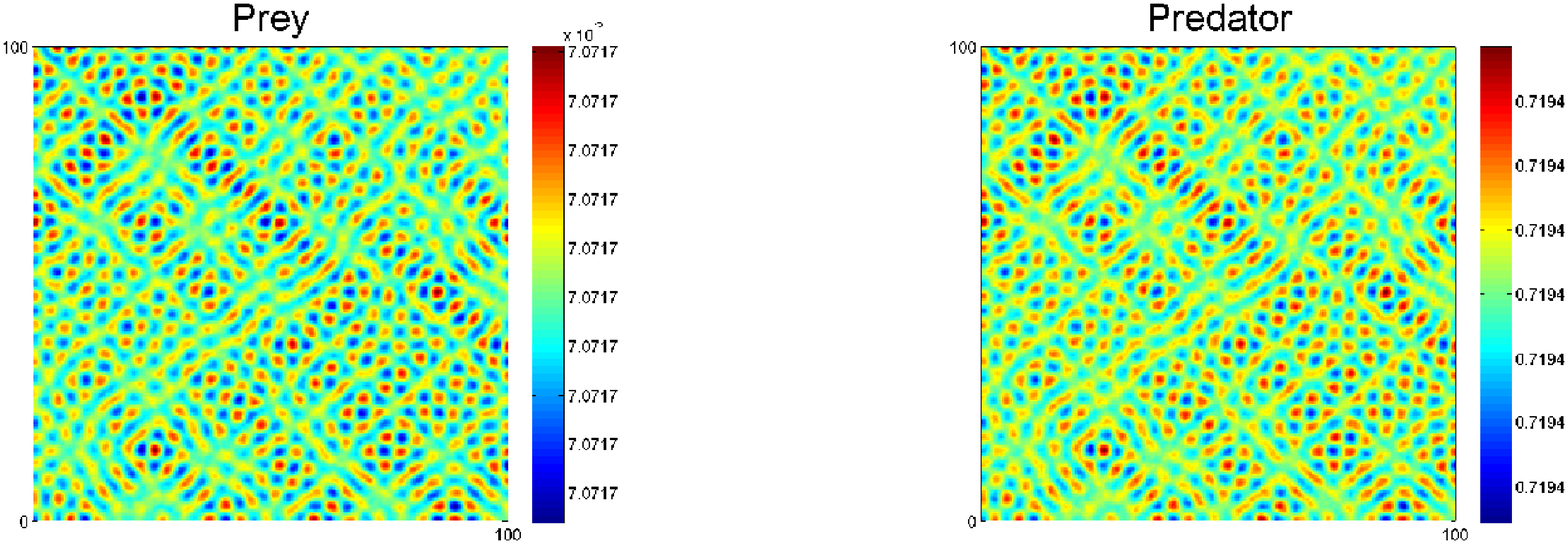}} \\
	\caption{Snapshots of contour pictures of the time evolution of prey $(x)$ and predator $(y)$ in the \textit{XY}-plane of the diffusive model \eqref{eq:03}, where the step size of the space, $\Delta h=0.4$ and the time has been taken as (a) $T=100$, (b) $T=2000$ and (c) $T=5000.$ Parameter sets are given in Table \ref{Table:1} .}
	\label{fig:13}
\end{figure}

\begin{figure}[ht!]
	\centering
	\subfloat[]{\includegraphics[width=7.0in, height=2.6in]{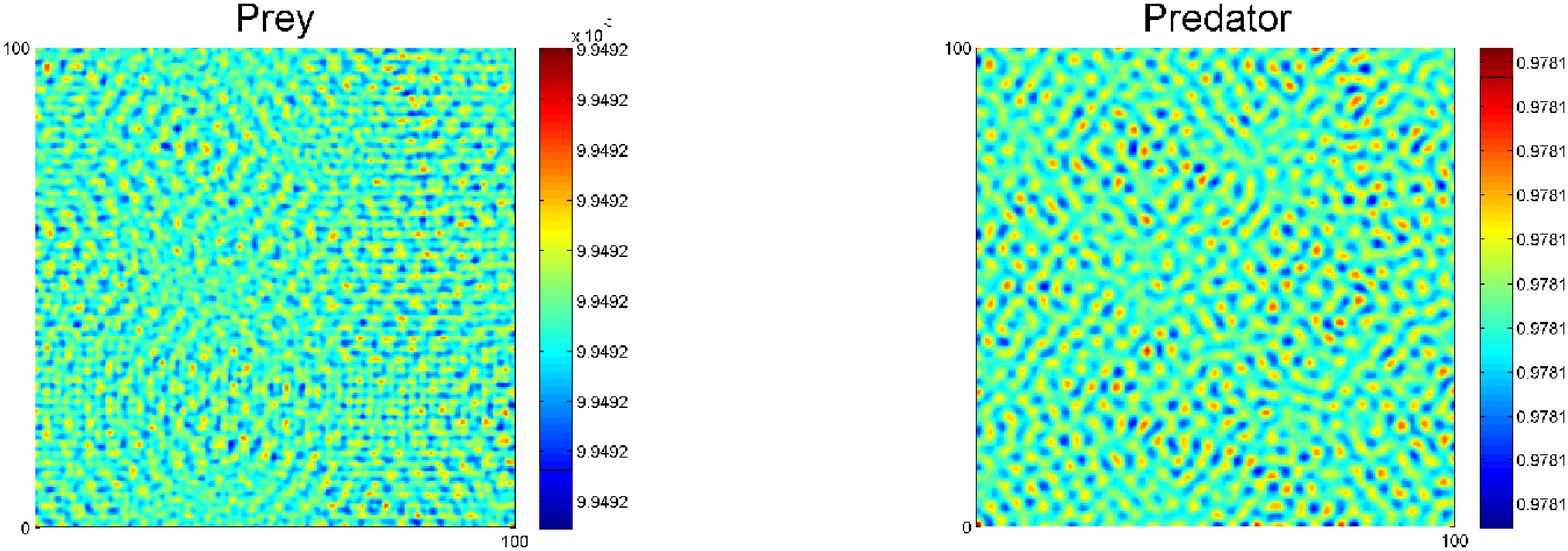}} \\
	\subfloat[]{\includegraphics[width=7.0in, height=2.6in]{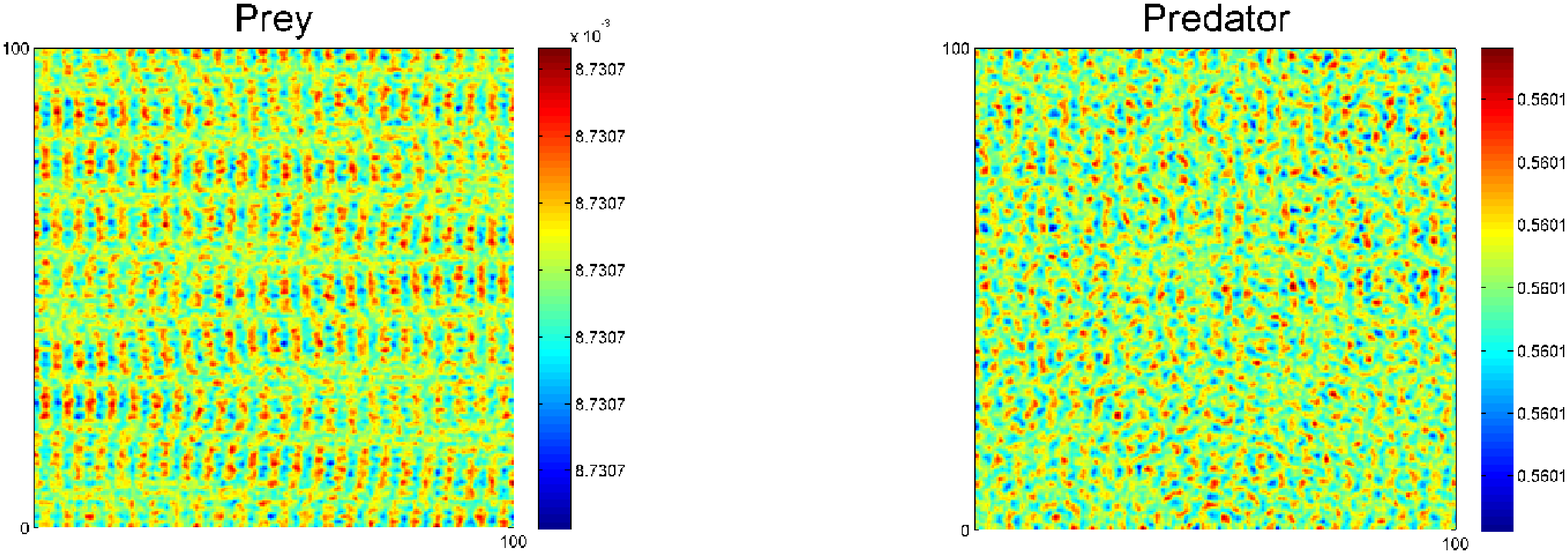}} \\
	\subfloat[]{\includegraphics[width=7.0in, height=2.6in]{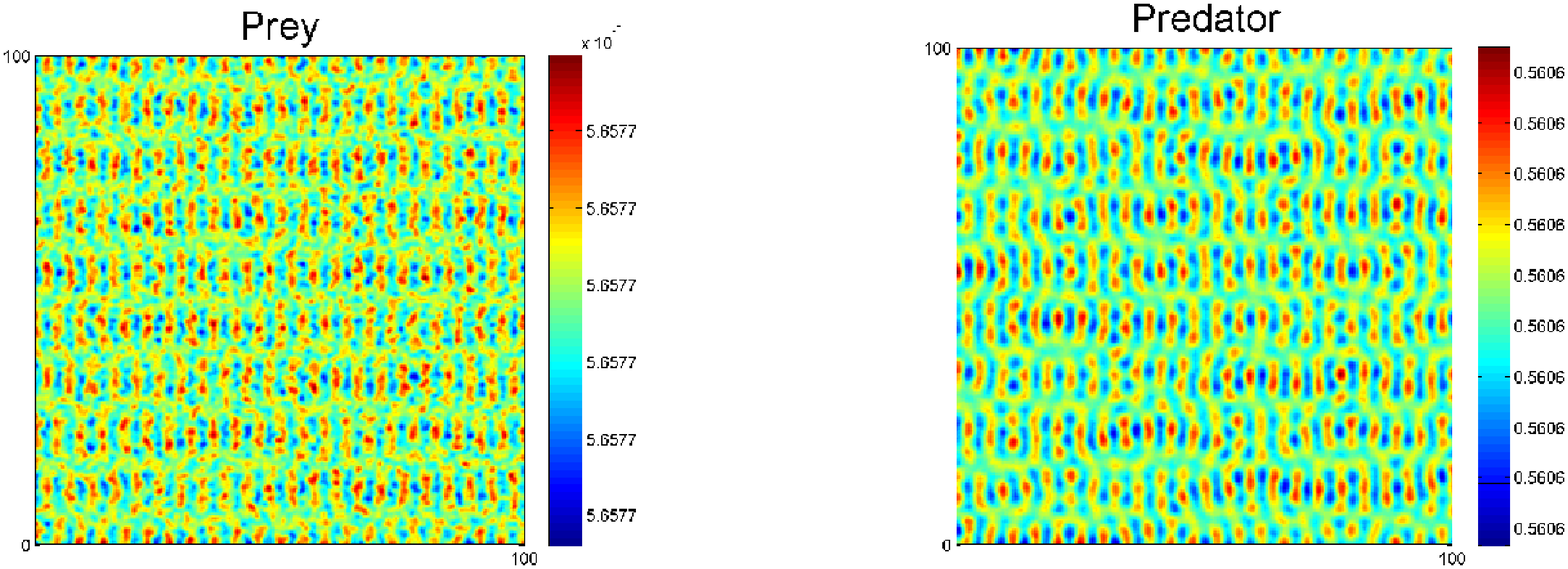}} \\
	\caption{Snapshots of contour pictures of the time evolution of prey $(x)$ and predator $(y)$ in the \textit{XY}-plane of the diffusive model \eqref{eq:03}, where the step size of the space, $\Delta h=0.4$ and the time has been taken as (a) $T=100$, (b) $T=2000$ and (c) $T=5000.$ Parameter sets are given in Table \ref{Table:1} .}
	\label{fig:14}
\end{figure}

\begin{figure}[ht!]
	\centering
	\subfloat[]{\includegraphics[width=7.0in, height=3.0in]{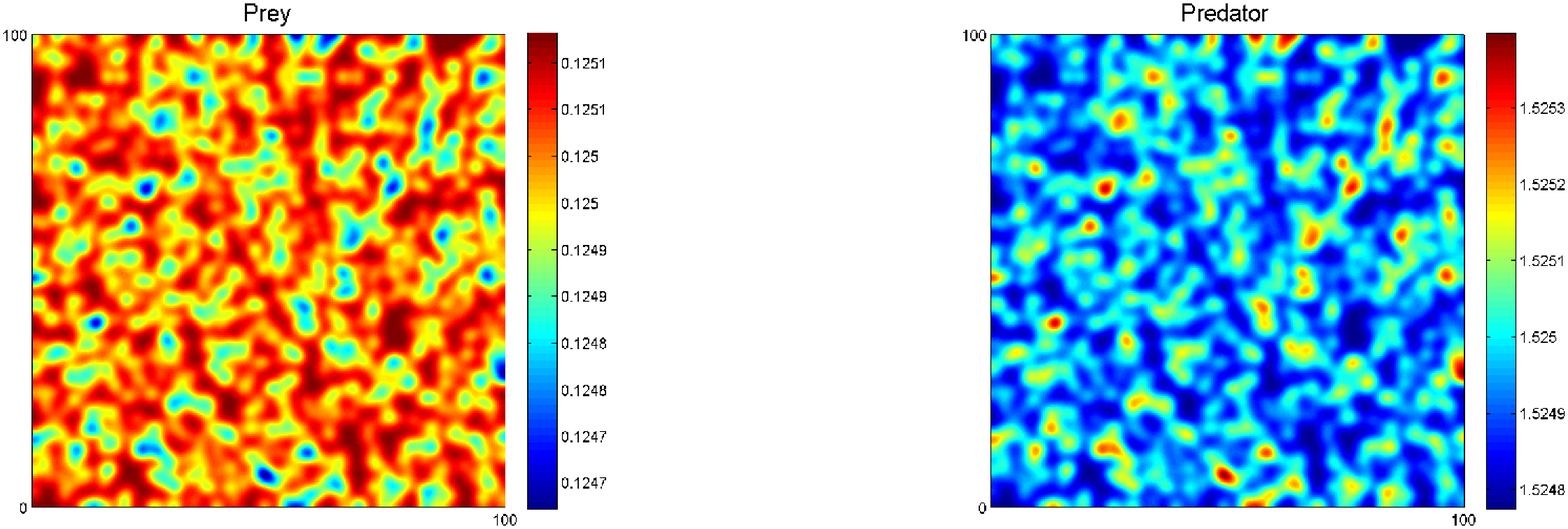}} \\
	\subfloat[]{\includegraphics[width=7.0in, height=3.0in]{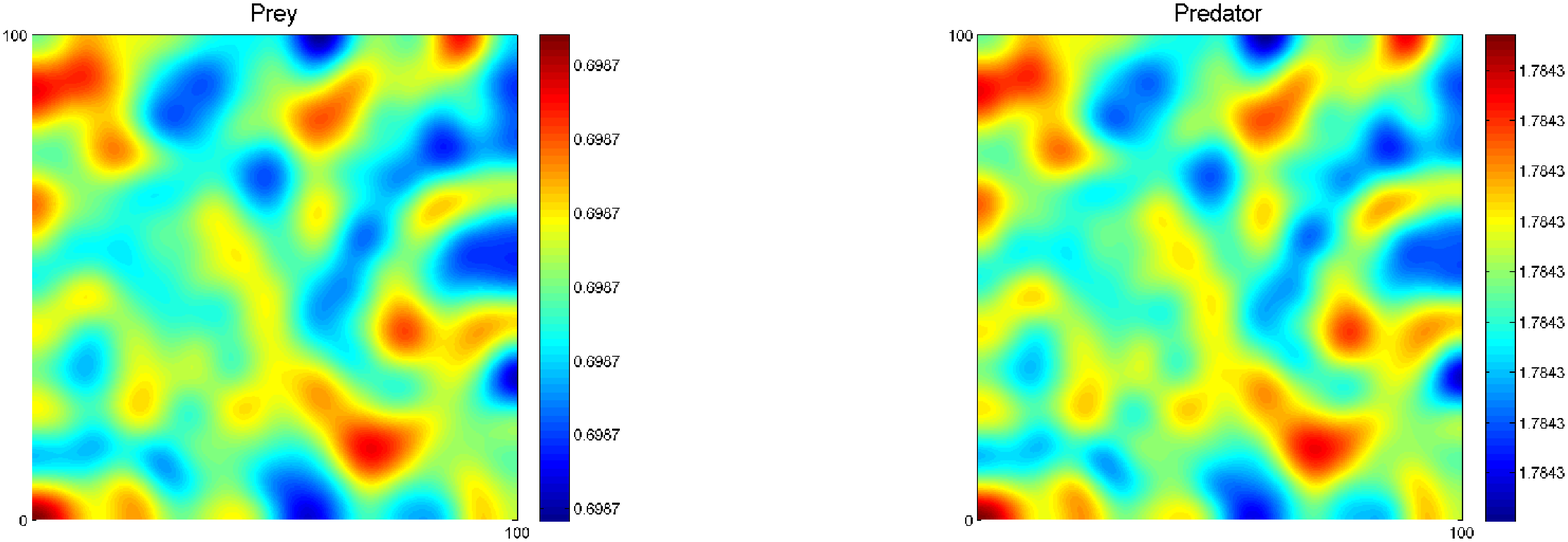}} \\
	\caption{Snapshots of contour pictures of the time evolution of prey $(x)$ and predator $(y)$ in the \textit{XY}-plane of the diffusive model \eqref{eq:03}. In figs (a) and (b), the time has been taken as $T=10$ and $T=100$ respectively. Parameter sets are given in Table \ref{Table:2}.}
	\label{fig:15}
\end{figure}

\begin{figure}[ht!]
	\centering
	\subfloat[]{\includegraphics[width=7.0in, height=3.0in]{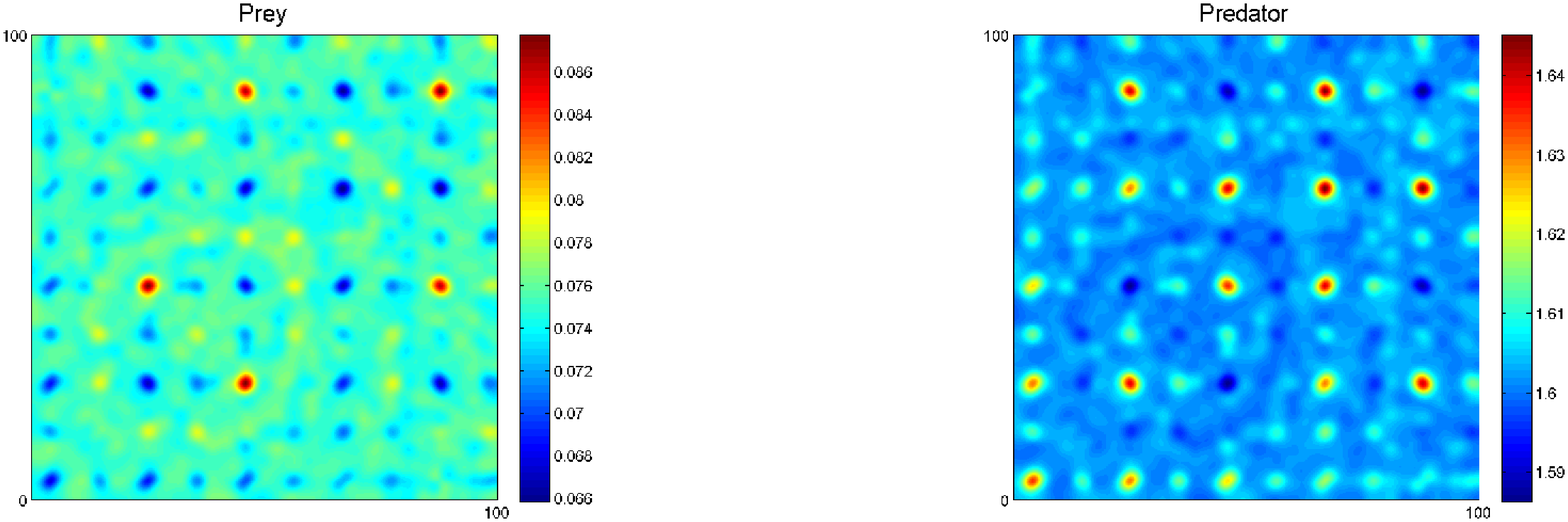}} \\
	\subfloat[]{\includegraphics[width=7.0in, height=3.0in]{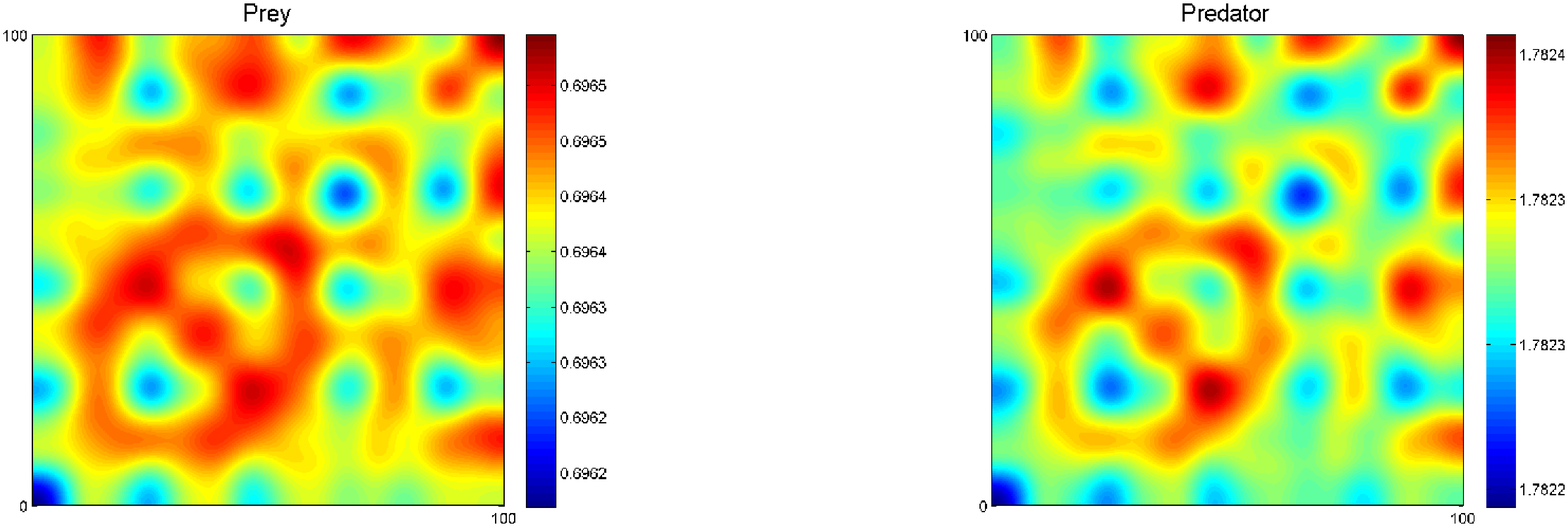}} \\
	\caption{Snapshots of contour pictures of the time evolution of prey $(x)$ and predator $(y)$ in the \textit{XY}-plane of the diffusive model \eqref{eq:03}. In figs (a) and (b), the time has been taken as $T=10$ and $T=100$ respectively. Parameter sets are given in Table \ref{Table:2} .}
	\label{fig:19}
\end{figure}

\end{document}